\documentclass[11pt]{article}

\usepackage[a4paper,margin=1in]{geometry}
\usepackage[T1]{fontenc}
\usepackage{lmodern}
\usepackage{amsmath,amssymb,amsthm,bm}
\usepackage{booktabs}
\usepackage{array,tabularx}
\usepackage{graphicx}
\usepackage{caption}
\AddToHook{env/table/begin}{\normalsize}
\usepackage{float}
\usepackage{microtype}
\usepackage[authoryear,longnamesfirst]{natbib}
\usepackage[hidelinks]{hyperref}

\hypersetup{
  pdftitle={Calibrated Horizon-Weighted Local Projection Designs for Markov Switchbacks},
  pdfauthor={Makoto Nakakita and Teruo Nakatsuma},
  pdfsubject={Dynamic experimental design; local projections; Markov switchbacks}
}

\graphicspath{{figures/}}
\newtheorem{proposition}{Proposition}
\newtheorem{lemma}{Lemma}
\newtheorem{corollary}{Corollary}
\newtheorem{assumption}{Assumption}

\title{Calibrated Horizon-Weighted Local Projection Designs for Markov Switchbacks}
\author{Makoto Nakakita\\
RIKEN\\
\texttt{makoto.nakakita@riken.jp}
\and
Teruo Nakatsuma\\
Keio University}
\date{}

\begin{document}
\maketitle

\begin{abstract}
We study temporal assignment design for Markov switchback experiments when the reported object is a dynamic local-projection target. We develop a calibrated selector that chooses the feasible persistence minimizing the covariance, HAC, residual-bootstrap, or realized-schedule risk of the estimator and reporting object specified before the experiment. A balanced homoskedastic Markov benchmark yields a closed form because the lagged-assignment information matrix is AR(1)-Toeplitz with a tridiagonal inverse. The benchmark maps local-projection reporting weights into persistence recommendations within a prespecified first-order Markov class. Field recommendations replace the benchmark covariance with residualized, serially dependent, pilot-calibrated, or randomization-based risk. A semi-synthetic Low Carbon London evaluation uses observed half-hourly baseline dynamics and known injected responses to assess design risk. It evaluates the covariance calculations under realistic load autocovariance and identifies when calibrated covariance selection should replace the homoskedastic Markov formula. Near-boundary designs use randomization-first inference when many-spell normal approximations are unsupported.
\end{abstract}



\section{Introduction}
Many economic experiments are dynamic. A demand-response signal can reduce electricity use immediately, trigger rebound later in the day, and reshape the load profile over subsequent periods. A platform pricing experiment can affect current participation, future search behavior, and congestion in later periods. A marketing intervention can have delayed effects that are invisible in a contemporaneous average treatment effect. In these settings, the object of interest is often a dynamic response curve rather than a single scalar effect.

Local projections provide a natural empirical language for reporting such responses. Researchers estimate horizon-specific effects $(g_0,\ldots,g_H)$, cumulative effects such as $\sum_{h=0}^H g_h$, or shape contrasts that distinguish immediate reductions from later rebound. Recent work has clarified the causal interpretation of local projections and related time-series objects: potential-system frameworks provide nonparametric content for time-series experiments, local projections, impulse responses, and SVARs \citep{CarlsonShephard2026}, and policy-path local projection frameworks define dynamic counterfactuals under sequential interventions \citep{Wang2026}.

We ask a complementary design question. If the researcher plans to report a dynamic response curve or a horizon-weighted functional of it, how should the treatment assignment path be designed? Existing switchback and time-series experimental-design papers largely target average effects, carryover robustness, covariate balance, or generic mean squared error. \citet{BojinovSimchiLeviZhao2023}, for example, study optimal switchback designs under assumptions on carryover order and provide randomization-based inference. Recent time-series A/B testing work considers full-history adaptive designs optimized for generic treatment-effect MSE \citep{WuEtAl2026}. These are natural objectives for many experiments, but they are not the same as optimizing precision for the multi-horizon local-projection object that the researcher ultimately reports.

We make the reported dynamic causal object the design target. We introduce the horizon-weighted local projection (HW-LP) criterion. Let $g=(g_0,\ldots,g_H)'$ be a design-invariant response curve and let $\theta_c=c'g$ be the scalar object of interest. The design risk is the asymptotic variance of $c'\hat g$ under the assignment mechanism. More generally, for a positive semidefinite matrix $W$, the risk is
\[
    \mathcal R_W(\pi)=\operatorname{tr}\{W V_\pi\},
\]
where $V_\pi$ is the design-dependent covariance matrix of the estimated response curve. This criterion covers isolated horizon effects, cumulative effects, delayed-response windows, rebound contrasts, and full-curve reporting.

The main analytical results are developed for balanced Markov switchback designs. Let $A_t\in\{0,1\}$, $\Pr(A_t=1)=1/2$, and $\Pr(A_t=A_{t-1})=s$. Write $r=2s-1$, so that $r=0$ is iid randomization, $r>0$ is persistent switching, and $r<0$ is alternating switching. For $X_t=(A_t-1/2,\ldots,A_{t-H}-1/2)'$, the population information matrix is
\[
    Q_H(r)=\frac14 (r^{|i-j|})_{i,j=0}^H.
\]
In the homoskedastic finite-memory model, $V(r)=\sigma^2 Q_H(r)^{-1}$. Because the inverse of this AR(1)-Toeplitz matrix is tridiagonal, the HW-LP risk for the contrast $c'g$ has the closed form
\[
    \mathcal R_c(r)=4\sigma^2\frac{a+br^2-2dr}{1-r^2},
\]
where $a=\sum_{j=0}^H c_j^2$, $b=\sum_{j=1}^{H-1}c_j^2$, and $d=\sum_{j=1}^{H}c_{j-1}c_j$. When $0<|d|<(a+b)/2$, the interior stationary point has the closed form
\[
    r_c^{\mathrm{int}}=\frac{(a+b)-\sqrt{(a+b)^2-4d^2}}{2d}.
\]
If $d=0$, iid assignment is optimal in the benchmark class; if $|d|=(a+b)/2$, the variance-only infimum is approached at the boundary and the feasible optimum is the nearest allowed endpoint.

The formulas imply that optimal switchback persistence is target-specific: isolated immediate effects favor iid assignment, smooth cumulative targets favor persistence, and field deployment should replace the closed-form benchmark with calibrated covariance selection when the benchmark assumptions fail. Within the benchmark Markov class, iid assignment is optimal for isolated horizon effects. Persistent treatment episodes are variance-favored for smooth cumulative effects. Alternating or moderate-persistence designs can be optimal for oscillating and rebound contrasts. The feasible set matters: at half-hour frequency, for example, $r=0.95$ implies an expected episode length of $2/(1-r)=40$ half-hours, while $r=0.7$ implies 6.67 half-hours. When carryover may extend beyond the estimated horizon, excessive persistence can also increase omitted-lag bias. We therefore add a bias-augmented sensitivity criterion that trades off target variance against worst-case omitted-carryover bias without claiming to be an exact finite-sample MSE under every omitted-tail data-generating process.

Figure~\ref{fig:risk-curves} summarizes how the reporting target changes the variance-minimizing persistence.
\begin{figure}[t]
    \centering
    \includegraphics[width=0.95\textwidth]{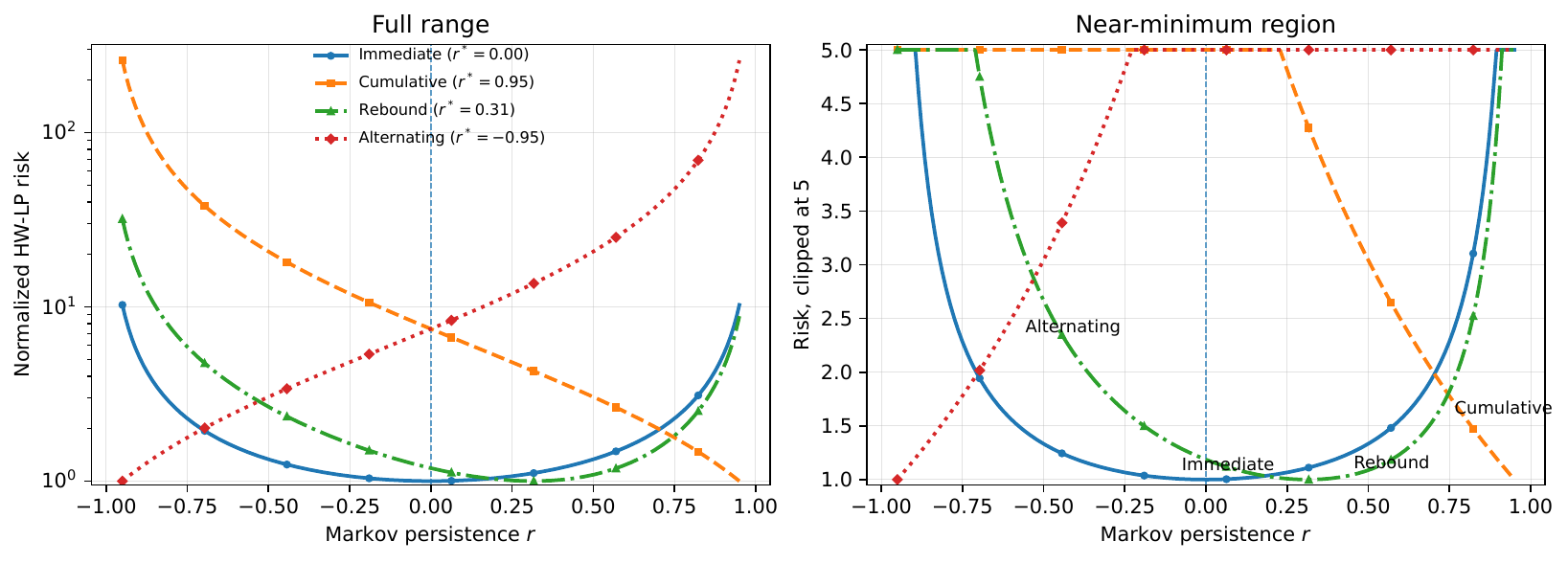}
    \caption{Target-specific risk and variance-minimizing persistence. Normalized HW-LP risk varies with the reporting target: isolated horizons favor iid assignment, smooth cumulative targets favor persistence, and oscillating contrasts can favor alternating designs.}
    \label{fig:risk-curves}
\end{figure}

We make three contributions. First, we formulate a horizon-weighted LP design criterion for a pre-specified reduced-form reporting object and show why path adjustment fixes that object while naive LP coefficients are design-dependent projections. Second, we derive a closed-form Markov benchmark for binary switchbacks and explain its frequency-domain interpretation, sparse active-share variants, and scope limits. The AR(1)-Toeplitz precision calculation is inherited from optimal input design; our contribution is the LP reporting-object translation, the binary Markov feasible class, and the operational mapping from target shape to persistence. Third, we turn the benchmark into a calibrated implementation rule by adding HAC and residualized covariance selection, pilot-regret bounds, local-bias carryover sensitivity, randomization-first near-boundary inference, and semi-synthetic covariance stress tests. The closed form is used as an interpretable benchmark for the persistence margin; field recommendations are based on calibrated covariance and randomization diagnostics when residual dependence, calendar controls, or feasibility constraints depart from the balanced homoskedastic Markov benchmark.

Section \ref{sec:target} defines the dynamic response targets and path-adjusted local projections. Section \ref{sec:markov} derives the HW-LP criterion and closed-form optimal persistence, with extensions in Appendix~\ref{app:design_extensions}. Section \ref{sec:robust} studies omitted-carryover misspecification. Section \ref{sec:simulations} presents estimation, design selection, and inference, with details in Appendix~\ref{app:inference_details}. Section \ref{sec:lcl} reports the empirical calibration, and Section~\ref{sec:limitations} discusses interpretation and external validity.

\section{Related literature}\label{sec:related}
The analysis relates to six literatures. The first is the literature on local projections and dynamic causal response objects. Local projections were introduced by \citet{Jorda2005} as a flexible way to estimate impulse responses. Subsequent work clarified that LPs and VARs share the same population impulse-response estimands under unrestricted lag structures, while LP inference can remain robust in persistent settings and at long horizons \citep{PlagborgMollerWolf2021,MontielOleaPlagborgMoller2021}. Other work studies the efficiency and finite-sample behavior of LP estimators, including comparisons with VAR-based impulse responses, smoothing across horizons, and practical inference procedures \citep{KilianKim2011,BarnichonBrownlees2019,InoueJordaKuersteiner2026}. Recent causal foundations show when time-series regressions, local projections, and policy-path responses can be interpreted as dynamic causal effects rather than purely predictive objects \citep{RambachanShephard2025,CarlsonShephard2026,Wang2026}. Related work on nonlinear environments shows that linear LPs with observed shocks or proxies can still identify weighted averages of causal effects, which is useful for interpreting response curves outside fully linear models \citep{KolesarPlagborgMoller2025}. Macroeconomic applications using LP-style response curves further illustrate why researchers often care about cumulative, delayed, and state-dependent dynamic effects rather than a single contemporaneous coefficient \citep{Ramey2016,AuerbachGorodnichenko2013,RameyZubairy2018}.

The second literature studies randomized experiments and temporal assignment. General potential-outcome treatments of randomized experiments emphasize that the estimand and assignment mechanism should be specified jointly \citep{ImbensRubin2015,AtheyImbens2017}. \citet{BojinovShephard2019} extend potential-outcome reasoning to single time-series experiments and exact randomization tests, while \citet{BojinovRambachanShephard2021} develop finite-population dynamic causal effects for panel experiments. Switchback designs with carryover have been analyzed by \citet{BojinovSimchiLeviZhao2023}, while minimax designs under habituation and clustered switchbacks under spatiotemporal interference are studied by \citet{BasseDingToulis2023} and \citet{JiaKallusYu2023}. More recent work connects regression-based and design-based inference, balances lagged outcomes and covariates through sequential rerandomization, and uses full-history reinforcement learning to optimize generic treatment-effect MSE \citep{LinDing2025,ZengEtAl2026,WuEtAl2026}. The design objective differs from contemporaneous average-effect, balance-score, and global-MSE criteria: it optimizes the covariance of the horizon-weighted response object that the researcher plans to report.

A related set of papers studies interference and exposure mappings in randomized experiments. Temporal carryover in this paper can be viewed as within-unit interference over time, so the exposure-mapping perspective is conceptually useful. Work on two-stage randomized designs, general interference, and unknown interference formalizes how treatment assignment can affect outcomes through exposure paths rather than through a single contemporaneous treatment indicator \citep{HudgensHalloran2008,AronowSamii2017,SavjeAronowHudgens2021}. Randomization tests under network interference also show how design-based inference changes when the null hypothesis is not sharp under the original assignment mechanism \citep{AtheyEcklesImbens2018}. The framework here is narrower: the exposure path is temporal, the response object is a local-projection target, and the main design margin is assignment persistence.

The third literature is optimal input design and optimal experimental design for dynamic systems. Classical optimal design gives general criteria for choosing informative experiments, including equivalence and optimality ideas that underlie many information-matrix calculations \citep{Kiefer1959,Fedorov1972,Pukelsheim2006}. Dynamic input-design work already shows that the input spectrum and autocorrelation shape the covariance of dynamic-parameter estimators; early and textbook treatments include \citet{Mehra1974}, \citet{GoodwinPayne1977}, \citet{Zarrop1979}, and \citet{Ljung1999}. The AR(1)-Toeplitz precision calculation used below is also a standard Gaussian-Markov precision fact, familiar from state-space and GMRF treatments \citep{RueHeld2005}. These precision-matrix results are standard. The contribution here is narrower: we translate the OID logic into a binary switchback problem whose target is a pre-specified LP reporting object rather than a transfer-function parameter; whose feasible class is constrained by run length, active share, and verifiability; and whose field recommendation can replace the homoskedastic information matrix by the residualized covariance of the estimator actually used. Thus the closed-form calculation below is an inherited precision-matrix benchmark, while the HW-LP contribution is the switchback-specific reporting, feasibility, and calibrated-implementation layer.

Applied design texts and Bayesian or nonlinear-design reviews emphasize that the target, loss function, and prior or model class determine the optimal design \citep{AtkinsonDonevTobias2007,ChalonerVerdinelli1995,PronzatoPazman2013}. Our calculations are adjacent to that literature, but the reported LP contrast is itself the object: it can be defined without specifying a full dynamic outcome model and remains meaningful when the design covariance is estimated by HAC, pilot, realized-schedule, or residual-bootstrap methods. This distinction is why the Toeplitz calculation is framed as a transparent benchmark and not as a claim of global dominance over continuous-valued, multisine, adaptive, finite-sequence, or fully optimized input designs.

The fourth literature is target-specific and data-driven switchback design. Empirical-Bayes and prior-data approaches study how carryover, seasonality, autocorrelation, and simultaneous experiments affect switchback MSE \citep{XiongChinTaylor2024}. Sequentially rerandomized switchbacks use lagged outcomes and covariates to improve balance over time \citep{ZengEtAl2026}. Online experimentation work highlights the operational importance of logging, randomization, long-run metrics, and platform constraints in large-scale A/B tests \citep{KohaviLongbothamSommerfieldHenne2009,KohaviTangXu2020,BakshyEcklesBernstein2014}. Long-term metric design in online platforms is especially close in spirit to the idea that the reported object should shape the experiment \citep{HohnholdOBrienTang2015}. The HW-LP criterion uses the same broad principle, but the tuning target is a pre-specified response curve, cumulative effect, delayed window, or rebound contrast.

The fifth literature is residential electricity demand response and dynamic pricing. The Low Carbon London data were produced by the United Kingdom's first residential dynamic time-of-use electricity-pricing trial and include smart-meter readings and tariff information \citep{StrbacEtAl2024,Schofield2015}. Reviews and field experiments show that households respond to time-varying prices, but the magnitude and timing of responses depend on tariff design, enabling technologies, and event duration \citep{FaruquiSergici2010,NewshamBowker2010,Wolak2011}. Behavioral and informational interventions also generate dynamic response patterns, including persistence, habituation, and rebound-like behavior \citep{Allcott2011,AllcottRogers2014,ItoIdaTanaka2018}. Additional studies of marginal-price perception, time-of-use tariff participation, and distributional effects emphasize that load shifting is shaped by information, salience, and household constraints \citep{Ito2014,Ozaki2018,YunusovTorriti2021}. The Low Carbon London exercise is not a causal estimate of the observed tariff; it uses realistic high-frequency load dynamics to evaluate how different assignment paths perform when the response path is known.

Finally, we contribute to the broader literature on experiment planning under multiple possible reporting targets. Researchers often enter a dynamic experiment knowing that they may report an immediate effect, cumulative effect, delayed window, or full response curve. The target-menu criterion introduced below provides a simple way to pre-specify this uncertainty without changing the estimand after observing the data. In that sense, the contribution is not another post-estimation adjustment for local projections, but a design-stage rule that ties the randomization path to the dynamic objects that will be reported.

\section{Identification and local-projection estimands}\label{sec:target}
\subsection{Dynamic response targets and finite-memory identification}
Potential-outcome notation defines the causal object before the LP design criterion is introduced, because the assignment path is dynamic and temporal carryover is a form of within-unit interference over time.

Let $a_{1:T}$ denote a candidate binary assignment path and let $Y_t(a_{1:T})$ be the potential outcome at time $t$. Our main target is an aggregate or single-series dynamic response, not a household-level heterogeneous-treatment-effect parameter. In applications with many households, $Y_t(a_{1:T})$ is the aggregate potential outcome under the common tariff or platform policy path. This convention rules out cross-sectional interference only if the target is an individual-level effect. If network, feeder, or congestion spillovers are part of the aggregate load response, they are included in the aggregate potential outcome; we then design for that aggregate response rather than for a no-spillover household estimand. This is a scope choice, not a proof that household-level interference is irrelevant. If the scientific object is an individual effect, a spillover effect, or an exposure-specific effect, the design should be expanded to a two-stage or clustered switchback design and an explicit spatial-temporal exposure mapping \citep{HudgensHalloran2008,AronowSamii2017,SavjeAronowHudgens2021,Leung2022}.

A further aggregation condition is required for design invariance. If the observed aggregate is a fixed linear average of household potential outcomes under the common policy path, then the aggregate response vector is the corresponding average of household response paths and does not change with the Markov persistence used to estimate it. If compliance, opt-out behavior, appliance automation, or missingness makes the set of effective contributors depend on the assignment path, the aggregate estimand becomes design-weighted. The HW-LP criterion still optimizes precision for the reported aggregate object, but the object should then be described as a policy-path response for the induced compliance regime rather than as a simple population average of household elasticities. This is the Jensen-style aggregation caveat: nonlinear participation and heterogeneous compliance can move the estimand when the design changes.

\begin{assumption}[Potential-outcome identification]\label{ass:po_identification}
For each experiment length $T$ and candidate assignment path $a_{1:T}$:
\begin{enumerate}
    \item[(i)] \emph{Well-defined aggregate exposure.} The observed outcome satisfies $Y_t=Y_t(A_{1:T})$. Either cross-sectional interference is absent for the unit-level estimand, or the reported estimand is the aggregate potential outcome that includes any spillovers induced by the common assignment path.
    \item[(ii)] \emph{No anticipation.} For any two paths with $a_{1:t}=a'_{1:t}$, $E\{Y_t(a_{1:T})\}=E\{Y_t(a'_{1:T})\}$.
    \item[(iii)] \emph{Sequential design exogeneity.} The assignment rule is implemented from the pre-specified randomization device and logged pre-treatment states. Conditional on the design history used by the rule, the innovation in $A_t$ is independent of future potential outcomes.
    \item[(iv)] \emph{Finite-memory linear projection.} After residualizing deterministic calendar controls and pre-treatment covariates, the mean potential-outcome path admits the projection
    \[
        E\{Y_t(a_{1:T})\}=\mu_t+\sum_{\ell=0}^{H_0} g_\ell a_{t-\ell},
    \]
    with $a_j=0$ for pre-experiment indices. The closed-form benchmark sets the reporting horizon $H\ge H_0$ or treats omitted lags through the local-bias sensitivity criterion in Section~\ref{sec:robust}.
\end{enumerate}
\end{assumption}

Assumption~\ref{ass:po_identification} separates identification from design. The potential-outcome restrictions define the response vector. The Markov assignment mechanism then determines the information matrix for estimating that vector. The response vector in the aggregate design is an aggregate reduced-form path response: it is the response of the chosen aggregation unit to a common assigned path. It is not an individual treatment effect, a structural price elasticity, a welfare primitive, or a household-level spillover estimand. If the scientific target is household elasticity, network spillover, feeder-level exposure, or structural welfare, the exposure mapping and design must be changed before applying the HW-LP criterion. The main asymptotic theory is written in a stationary super-population language because it yields compact covariance formulas. A finite-population or fixed-schedule reading is also possible: condition on the potential-outcome schedule and evaluate the randomization distribution induced by the assignment path. The realized-schedule diagnostic in Proposition~\ref{prop:realized_schedule} uses exactly this finite-design reading. The Low Carbon London exercise is semi-synthetic: it fixes the observed baseline load path, injects known response paths, and randomizes assignments to evaluate design risk. It therefore validates the variance and covariance arithmetic under realistic baseline dynamics, not causal recovery of an unknown historical tariff response.

Consider a single time-series experiment with binary assignment $A_t\in\{0,1\}$. Under Assumption~\ref{ass:po_identification}, the initial analytical setting can be written as the finite-memory additive response model
\begin{equation}
    Y_t=\mu_t+\sum_{\ell=0}^{H} g_\ell A_{t-\ell}+\varepsilon_t.
    \label{eq:finite_memory}
\end{equation}
The vector $g=(g_0,\ldots,g_H)'$ is the dynamic response curve. Throughout the analytical section, deterministic controls, intercepts, calendar effects, and pre-treatment covariates are treated as partialled out. Equivalently, all regressions below can be read after applying a Frisch--Waugh--Lovell residualization step to the outcome and to the lagged assignment vector. A researcher may care about the full curve or about a functional
\begin{equation}
    \theta_c=c'g=\sum_{\ell=0}^H c_\ell g_\ell.
    \label{eq:theta_c}
\end{equation}
The contrast $c=e_0$ targets the contemporaneous effect; $c=\mathbf 1$ targets the cumulative effect; weights concentrated on medium-run horizons target delayed effects; sign-changing weights target rebound or shape contrasts.

The notation below uses $Q_H$ for the assignment information matrix, $V_H$ for the covariance of the path-adjusted LP estimator, $\mathcal R$ for the feasible persistence set, $p_0$ for an active-share budget, and $r$ for assignment persistence.

The key distinction is between the causal target and the assignment design. The target $g$ or $c'g$ is fixed. The assignment mechanism changes the information available about that target.

The terminology ``local projection target'' refers to the response object the researcher plans to report: a vector of horizon-specific effects or a weighted functional of that vector. In the finite-memory additive model, the path-adjusted local projection and the distributed-lag representation estimate the same design-invariant response components. We use the distributed-lag form in the analytical benchmark because it exposes the assignment information matrix and yields closed-form design rules. In more general nonlinear, state-dependent, or heteroskedastic environments, the same design principle applies to the covariance matrix of the chosen local-projection estimating equations:
\[
    \hat r_W=\arg\min_{r\in\mathcal R}\operatorname{tr}\{W\hat V(r)\},
\]
where $\hat V(r)$ can be estimated from a pilot experiment, residual bootstrap, or calibrated baseline model. Thus the finite-memory results are not meant to exhaust all local-projection settings; they provide a transparent benchmark and comparative statics for choosing assignment persistence. The reported scalar object is $\theta_c=c'g$, where $c$ encodes the pre-specified horizon weights. Naive LP coefficients and path-adjusted coefficients coincide only under finite-memory restrictions that remove future-assignment mixtures; outside that benchmark, the calibrated implementation replaces the closed-form covariance by $\widehat V(r)$ for the estimator actually used.

\subsection{Path-adjusted local projections}
The first requirement for a design criterion is that the estimand does not change when the assignment persistence changes. In the finite-memory benchmark, the path-adjusted LP satisfies this requirement. Let $X_t=(\widetilde A_t,\ldots,\widetilde A_{t-H})'$ denote the centered lag vector and let $Q_X=E[X_tX_t']$.

\begin{lemma}[Target invariance of path-adjusted LP]\label{lem:target_invariance}
Suppose the finite-memory benchmark can be written after centering and residualizing deterministic controls as
\[
    Y_t=\alpha+X_t'g+\varepsilon_t,
    \qquad E[X_t\varepsilon_t]=0.
\]
If $Q_X$ is nonsingular, then the population coefficient from the path-adjusted LP, equivalently the distributed-lag regression of $Y_t$ on $X_t$, is $g$. Thus the response vector $g$ is invariant to the assignment persistence whenever the lagged-assignment covariance is full rank. Assignment design changes the information matrix used to estimate $g$, not the target itself.
\end{lemma}

A naive local projection of $Y_{t+h}$ on $A_t$, by contrast, can mix the effect of the initial assignment with the effects of subsequent assignments when the design induces persistence. The following proposition makes the source of the problem explicit.

\begin{proposition}[Naive local projections are design-dependent mixtures]\label{prop:naive_mixture}
Let $\widetilde A_t=A_t-EA_t$ and suppose the assignment process satisfies
\[
    E[\widetilde A_t\widetilde A_{t-k}]=\sigma_A^2r^{|k|},\qquad |r|<1.
\]
Consider the finite-memory response model written in centered assignments,
\[
    Y_{t+h}=\mu_h+\sum_{\ell=0}^{H}g_\ell \widetilde A_{t+h-\ell}+\varepsilon_{t+h},
    \qquad E[\varepsilon_{t+h}\widetilde A_t]=0.
\]
The population coefficient from the naive local projection of $Y_{t+h}$ on $\widetilde A_t$ is
\[
    \beta_h^{\mathrm{naive}}(r)
    =\frac{\operatorname{Cov}(Y_{t+h},\widetilde A_t)}{\operatorname{Var}(\widetilde A_t)}
    =\sum_{\ell=0}^{H}g_\ell r^{|h-\ell|}.
\]
Thus, unless $r=0$ or the response has only one nonzero component, the naive LP coefficient is a persistence-dependent mixture of nearby dynamic response components rather than the design-invariant component $g_h$.
\end{proposition}

The proposition does not imply that the naive LP coefficient is inconsistent for its projected estimand. Rather, the projected estimand itself changes with assignment persistence. This leads to a simple reporting rule. If the pre-analysis plan defines a design-invariant dynamic response component, cumulative effect, or shape contrast, the analysis should use the path-adjusted LP or the equivalent distributed-lag regression. If the scientific object is instead the projection induced by the implemented policy path--for example, the average response to the whole bundled tariff path rather than to a single initial assignment--then the naive LP can be the appropriate estimand, but the estimand should be described as design-specific. Our design problem is the first case: the response object is fixed before choosing assignment persistence.

Proposition~\ref{prop:naive_mixture} and Lemma~\ref{lem:target_invariance} are the reason the design target is defined through path-adjusted local projections. The goal is not to optimize precision for a response object that changes with the assignment persistence, but to fix the dynamic response component and then choose the assignment path that estimates it well. To isolate dynamic response components, we use the path-adjusted local projection, equivalently the distributed-lag representation
\begin{equation}
    Y_t=\alpha+g_0 A_t+g_1 A_{t-1}+\cdots+g_H A_{t-H}+u_t.
    \label{eq:path_adjusted}
\end{equation}
In richer models, the horizon-$h$ regression can be written as
\begin{equation}
    Y_{t+h}=\alpha_h+\beta_h A_t+\Gamma_h'A_{t+1:t+h}+\Lambda_h'H_t+u_{t+h,h},
    \label{eq:future_path_adjusted}
\end{equation}
where future assignments and history controls adjust for design-induced path correlation. The finite-memory representation in \eqref{eq:path_adjusted} is the workhorse for the closed-form design results below.

\section{Markov assignment and the HW-LP information criterion}\label{sec:markov}
\subsection{Assignment information and design feasibility}
We focus first on balanced Markov switchback designs:
\begin{equation}
    \Pr(A_t=1)=\frac12,\qquad \Pr(A_t=A_{t-1})=s.
    \label{eq:balanced_markov}
\end{equation}
Define $r=2s-1$ and $\widetilde A_t=A_t-1/2$. The parameter $r$ is the assignment persistence: $r=0$ is iid Bernoulli assignment, $r>0$ favors longer treatment episodes, and $r<0$ favors alternation. Designs are chosen from a feasible interval $\mathcal R=[\underline r,r_{\max}]\subset(-1,1)$, which can encode minimum dwell-time, fatigue, safety, switching-cost, or operational constraints. Since $s=(1+r)/2$, the expected treatment-episode length is
\[
    E[\text{run length}]=\frac{1}{1-s}=\frac{2}{1-r}.
\]
This translation is useful in applications: at half-hour frequency, $r=0.7$ corresponds to about 3.33 hours, while $r=0.95$ corresponds to about 20 hours. Let $X_t=(\widetilde A_t,\widetilde A_{t-1},\ldots,\widetilde A_{t-H})'$.

\begin{proposition}[Information matrix under unbalanced Markov switchback]\label{prop:unbalanced}
Let $p=\Pr(A_t=1)\in(0,1)$, define $\widetilde A_t=A_t-p$, and let $r$ be the second eigenvalue of the two-state transition matrix. Equivalently,
\[
    \Pr(A_t=1\mid A_{t-1}=1)=p+(1-p)r,
    \qquad
    \Pr(A_t=1\mid A_{t-1}=0)=p(1-r).
\]
The transition probabilities are valid when
\[
    r_{\min}(p)\equiv \max\{-p/(1-p),-(1-p)/p\}\le r\le 1.
\]
The deterministic boundary values can lead to non-mixing or nearly singular designs. All inverse and large-sample covariance statements below therefore use feasible persistence sets contained in the interior, $\mathcal R_p\subset(r_{\min}(p),1)$. For the stationary chain,
\begin{equation}
    E[\widetilde A_t\widetilde A_{t-k}]=p(1-p)r^{|k|},
    \qquad
    Q_H(p,r)=p(1-p)(r^{|i-j|})_{i,j=0}^{H}.
    \label{eq:unbalanced_Q}
\end{equation}
Consequently, if $p$ is fixed by operational constraints, the variance-only optimal persistence for a contrast $c'g$ is the balanced-case rule while the information scale is multiplied by $p(1-p)$. The run-length formulas are asymmetric when $p\ne1/2$. Conditional on entering the active state, the active spell length $L_1$ is geometric on $\{1,2,\ldots\}$ with exit probability
\[
    \Pr(A_{t+1}=0\mid A_t=1)=(1-p)(1-r),
\]
so
\[
    E[L_1]=\{(1-p)(1-r)\}^{-1}.
\]
Similarly, the inactive spell length $L_0$ satisfies $E[L_0]=\{p(1-r)\}^{-1}$. When $p=1/2$, active and inactive spell lengths coincide and reduce to the balanced formula $2/(1-r)$.
\end{proposition}

This unbalanced extension is important in demand-response settings, where high-price or event signals are often sparse. Sparse assignment shares do not change the Toeplitz persistence logic, but they reduce information by the factor $p(1-p)$ and make active-run-length constraints explicit.

\begin{corollary}[Budgeted sparse-signal HW-LP designs]\label{cor:budgeted_sparse}
Fix an active-share budget $p_0\in(0,1)$ and consider the unbalanced Markov class
\[
    \Pi(p_0,\mathcal R)=\{\text{stationary two-state Markov assignments}: \Pr(A_t=1)=p_0,r\in\mathcal R\},
\]
where $\mathcal R\subset(r_{\min}(p_0),1)$ is a compact interior feasible set.
In the homoskedastic finite-memory benchmark, the contrast risk is
\[
    \mathcal R_c(p_0,r)=\frac{\sigma^2}{p_0(1-p_0)}\,c'R_H(r)^{-1}c,
\]
where $R_H(r)=(r^{|i-j|})_{i,j=0}^{H}$. Hence, conditional on the active-share budget $p_0$, the variance-optimal persistence is the same closed-form rule as in Proposition~\ref{prop:optimal}. If an experiment imposes a maximum expected active spell length $L_{1,\max}$, where a spell is conditional on starting in the active state and includes the first active period, then the feasible set must also satisfy
\[
    r \le 1-\{(1-p_0)L_{1,\max}\}^{-1}.
\]
If it also imposes a maximum expected inactive spell length $L_{0,\max}$, the additional constraint is $r\le 1-\{p_0L_{0,\max}\}^{-1}$. Thus sparse demand-response designs can use the same HW-LP persistence logic after separating the event budget $p_0$ from the temporal persistence parameter $r$, while the feasible set records operational restrictions on active and inactive spell lengths.
\end{corollary}

\begin{corollary}[Active-share budgets and information scale]\label{cor:active_share_budget}
Suppose the active share can be chosen subject to $p\in[\underline p,\bar p]$ and $\bar p\le 1/2$, while the feasible persistence set is fixed. In the homoskedastic benchmark, the contrast variance is proportional to $[p(1-p)]^{-1}$. Hence the variance-efficient active share is the largest feasible share, $p=\bar p$. If the feasible set contains $1/2$, the variance-efficient share is $p=1/2$. Operational constraints therefore enter in two separable ways: the active-share budget controls the information scale, while the persistence parameter controls how that information is distributed across horizons.
\end{corollary}

For sparse demand-response signals, strong negative persistence may be infeasible because the chain cannot alternate frequently while keeping a very small active share. The lower bound $r_{\min}(p_0)$ is therefore not a technicality: it rules out assignment patterns that would violate the event budget. This matters most for sign-changing or oscillating targets, for which the unconstrained benchmark may prefer negative persistence. If the active share is $p_0=0.10$, for example, the Markov feasibility lower bound is only $r_{\min}=-0.111$, so strongly alternating designs are unavailable regardless of their variance ranking in the balanced benchmark. Positive persistence is usually feasible, but active- and inactive-spell caps may impose upper bounds below one.

Table~\ref{tab:sparse-budget} reports the information scale and active-run constraints for representative active-share budgets.

\begin{table}[htbp]
    \centering
    \caption{Sparse-signal information scale and active-run constraints.}
    \label{tab:sparse-budget}
    \begin{tabular}{rrrr}
\toprule
Active share $p_0$ & Info. scale $4p_0(1-p_0)$ & $r_{\max}$, 3h cap & $r_{\max}$, 6h cap \\
\midrule
0.05 & 0.190 & 0.825 & 0.912 \\
0.10 & 0.360 & 0.815 & 0.907 \\
0.14 & 0.482 & 0.806 & 0.903 \\
0.25 & 0.750 & 0.778 & 0.889 \\
0.50 & 1.000 & 0.667 & 0.833 \\
\bottomrule
\end{tabular}

    \begin{minipage}{0.92\linewidth}
    Notes: The information scale is $4p_0(1-p_0)$, normalized to one at balanced assignment $p_0=1/2$. The last two columns report the largest feasible persistence $r$ implied by $E[L_1]\le L_{1,\max}$ for $L_{1,\max}=6$ and $12$ active half-hours, where \(L_1\) counts the active spell including the first active period. These calculations separate the event budget from the temporal persistence choice.
    \end{minipage}
\end{table}

Population information scaling does not by itself describe the finite-path risk of a sparse experiment. Appendix Table~\ref{tab:sparse_finite_path} therefore simulates the cumulative-target design at the three-hour active-spell cap using the exact realized lagged-assignment Gram matrix. With $T=2{,}000$ and $p_0=0.05$, the median risk ratio is 1.034 and the p90 ratio is 1.817, while the p10 path contains only 12 active spells; at $T=17{,}520$ the p90 ratio falls to 1.197. None of the simulated Gram matrices is singular. Thus sparse short experiments should screen candidate randomization paths before launch even when the population Markov information matrix is nonsingular.

\begin{proposition}[Information matrix under balanced Markov switchback]\label{prop:info}
Suppose $A_t$ follows the stationary balanced Markov chain in \eqref{eq:balanced_markov}. Then
\[
    Q_H(r)\equiv E[X_tX_t']=\frac14\left(r^{|i-j|}\right)_{i,j=0}^{H}.
\]
For $|r|<1$,
\[
    Q_H(r)^{-1}=\frac{4}{1-r^2}
    \begin{pmatrix}
        1 & -r & 0 & \cdots & 0\\
        -r & 1+r^2 & -r & \cdots & 0\\
        0 & -r & 1+r^2 & \cdots & 0\\
        \vdots & \vdots & \vdots & \ddots & -r\\
        0 & 0 & 0 & -r & 1
    \end{pmatrix}.
\]
\end{proposition}

\emph{Proof.} See Appendix~\ref{app:proofs}.

\begin{lemma}[Uniform mixing for compact Markov design classes]\label{lem:uniform_markov_mixing}
Fix $p\in[\underline p,\bar p]\subset(0,1)$ and let the feasible persistence set satisfy
\[
    \mathcal R_p\subset\{r:r_{\min}(p)+\eta\le r\le r_{\max}\},
    \qquad \eta>0,\quad r_{\max}<1.
\]
Define the modulus bound $\rho_{\max}=\sup_{p,r\in\mathcal R_p}|r|<1$. Then the stationary two-state Markov assignment process is geometrically alpha-mixing uniformly over the feasible class: there is a finite constant $C$ such that
\[
    \alpha_A(k;p,r)\le C \rho_{\max}^k,
    \qquad k\ge1.
\]
The same bound holds for the finite lag vector $X_t=(\widetilde A_t,\ldots,\widetilde A_{t-H})$ up to replacing $k$ by $k-H$. If the outcome innovation or residual process is independent of assignment conditional on calendar controls, or is jointly alpha-mixing with coefficients summable uniformly over the same class, then standard product and measurable-transform inequalities for mixing arrays imply that the score process $S_t(r)=X_t(r)\varepsilon_t(r)$ satisfies the uniform mixing part of Assumption~\ref{ass:hac}; see, for example, \citet[Ch.~14]{Davidson1994} and \citet[Ch.~3]{White2001}. This is separate from the nonsingularity of $Q_H(r)$; compact interior feasible sets give both uniform mixing and uniformly bounded inverse information matrices, but for different reasons.
\end{lemma}

\begin{proposition}[Design-dependent covariance]\label{prop:covariance}
Let $X_t=(\widetilde A_t,\ldots,\widetilde A_{t-H})'$ and suppose that, after the stated residualization of deterministic controls, the benchmark outcome satisfies
\[
    Y_t=X_t'g+\varepsilon_t .
\]
Assume the balanced Markov assignment has $|r|<1$ and $T^{-1}\sum_{t=1}^T X_tX_t'\to_p Q_H(r)$. Suppose further that the errors are conditionally mean-zero, conditionally homoskedastic, and conditionally serially uncorrelated given the full assignment path:
\[
    E[\varepsilon_t\mid A_{-\infty:\infty}]=0,\qquad
    E[\varepsilon_t^2\mid A_{-\infty:\infty}]=\sigma^2,\qquad
    E[\varepsilon_t\varepsilon_s\mid A_{-\infty:\infty}]=0\quad(t\ne s).
\]
Then
\[
    \sqrt T(\hat g-g)\Rightarrow N(0,V_H(r)),\qquad V_H(r)=\sigma^2Q_H(r)^{-1}.
\]
More generally, if the score process has long-run covariance
\[
    \Omega_H(r)=\sum_{k=-\infty}^{\infty}E[X_t\varepsilon_t\varepsilon_{t-k}X_{t-k}'],
\]
then
\[
    \sqrt T(\hat g-g)\Rightarrow N\{0,Q_H(r)^{-1}\Omega_H(r)Q_H(r)^{-1}\}.
\]
\end{proposition}

\subsection{Covariance estimation under serial dependence}
The clean covariance in Proposition~\ref{prop:covariance} is a benchmark. In applications with deterministic calendar controls or other pre-specified residualization variables $Z$, the relevant finite-sample information is not the raw Toeplitz matrix but the residualized matrix
\[
    Q_{H,Z}(r)=\operatorname*{plim}_{T\to\infty}T^{-1}X_H(r)'M_ZX_H(r),\qquad
    M_Z=I-Z(Z'Z)^{-1}Z'.
\]
The closed-form $Q_H(r)$ rule is exact for the no-control homoskedastic benchmark and remains useful as a persistence diagnostic. The calibrated implementation below instead evaluates the residualized design matrix, or its bootstrap/HAC analogue, for the realized or pilot-calibrated schedule. Thus calendar residualization is not treated as innocuous: it is precisely one reason we recommend switching from the analytical rule to the calibrated selector when finite-sample covariance differs materially from the Toeplitz benchmark.

The following proposition records the primitive conditions used for the calibrated implementation. It is stated for a compact persistence set and a standard Bartlett/Newey--West estimator; other kernels can be used under the usual moment and bandwidth restrictions familiar from data-dependent HAC procedures \citep{Andrews1991,AndrewsMonahan1992}.

\begin{assumption}[Uniform HAC primitives]\label{ass:hac}
For each $r\in\mathcal R$, where $\mathcal R$ is a fixed compact interior persistence set satisfying Lemma~\ref{lem:uniform_markov_mixing}, let $S_t(r)=X_t(r)\varepsilon_t(r)$ be the residualized path-adjusted LP score. The array $\{S_t(r):r\in\mathcal R\}$ is strictly stationary, has uniformly bounded $4+\delta$ moments, and is alpha-mixing with coefficients satisfying $\sum_{k\ge1} k^{1/2}\alpha(k)^{\delta/(4+\delta)}<\infty$ uniformly over $r\in\mathcal R$. The map $r\mapsto S_t(r)$ is Lipschitz in $L_2$ uniformly in $t$. The long-run covariance
\[
    \Omega_H(r)=\sum_{k=-\infty}^{\infty}E[S_t(r)S_{t-k}(r)']
\]
exists and $Q_H(r)$ is nonsingular uniformly on $\mathcal R$.
\end{assumption}

\begin{proposition}[Uniform HAC covariance for HW-LP risk]\label{prop:hac_uniform}
Let Assumption~\ref{ass:hac} hold. Let
\[
    \widehat\Omega_H(r)=\sum_{|k|\le b_T}\left(1-\frac{|k|}{b_T+1}\right)
    \widehat\Gamma_k(r),
    \qquad
    \widehat\Gamma_k(r)=T^{-1}\sum_{t=|k|+1}^{T}S_t(r)S_{t-|k|}(r)'
\]
be the Bartlett HAC estimator, and assume $b_T\to\infty$ and $b_T/T\to0$. Then
\[
    \sup_{r\in\mathcal R}\|\widehat\Omega_H(r)-\Omega_H(r)\|=o_p(1).
\]
If also $\sup_{r\in\mathcal R}\|\widehat Q_H(r)-Q_H(r)\|=o_p(1)$, then
\[
    \sup_{r\in\mathcal R}\left\|\widehat Q_H(r)^{-1}\widehat\Omega_H(r)\widehat Q_H(r)^{-1}
    -Q_H(r)^{-1}\Omega_H(r)Q_H(r)^{-1}\right\|=o_p(1).
\]
Consequently, $\sup_{r\in\mathcal R}|\operatorname{tr}\{W\widehat V(r)\}-\operatorname{tr}\{WV(r)\}|=o_p(1)$ for any fixed finite-dimensional $W$.
\end{proposition}

\noindent The proposition provides the baseline form of the calibrated rule used below. For a finite grid of candidate persistences, the same conclusion requires only pointwise HAC consistency for each grid point. For a continuous persistence interval, the Lipschitz and mixing conditions give the stochastic equicontinuity needed to pass from pointwise consistency to uniform consistency. The uniform statement is for a fixed compact interior set; it is not a claim that the same bandwidth works as $r_{\max}$ drifts to one. If a sequence of feasible sets has $r_{\max,T} \uparrow 1$, the truncation lag must cover the growing dependence length, for example by requiring $b_T(1-r_{\max,T})\to\infty$ and $b_T/T\to0$. This also makes clear why the effective sample size is closer to $T_{\mathrm{eff}}(r)=T(1-r)/(1+r)$ than to $T$ when $r$ is large.

In implementation we therefore use a candidate-specific bandwidth diagnostic rather than a single iid-oriented bandwidth. The starting value is
\[
    b_{0T}(r)=\left\lceil 1.3 T^{1/3}\frac{1+|r|}{1-|r|}\right\rceil,
\]
and the finite-sample default rounds the capped value to the nearest integer according to
\[
    b_T(r)=\max\!\left\{1,\left\lfloor
    \min\{b_{0T}(r),\,T/4,\,0.25T_{\mathrm{eff}}(r)\}+\frac12
    \right\rfloor\right\},
\]
with sensitivity checks at one half and twice the capped value. If the cap binds, the report treats the design as a high-persistence finite-sample warning rather than as evidence that a very long HAC window is reliable. This formula is not an optimal-bandwidth theorem; it is a conservative operational rule that expands the lag window as assignment persistence approaches the boundary while keeping the truncation lag close to one quarter of the effective sample or below. The effective-sample-size column is likewise a scalar AR(1) variance-inflation diagnostic, not an exact sample size for every vector LP contrast; its role is to flag when persistence makes iid-oriented HAC choices unreliable. Figure~\ref{fig:bandwidth-diagnostics} reports the scale for the LCL sample length.

\begin{figure}[htbp]
\centering
\includegraphics[width=0.86\textwidth]{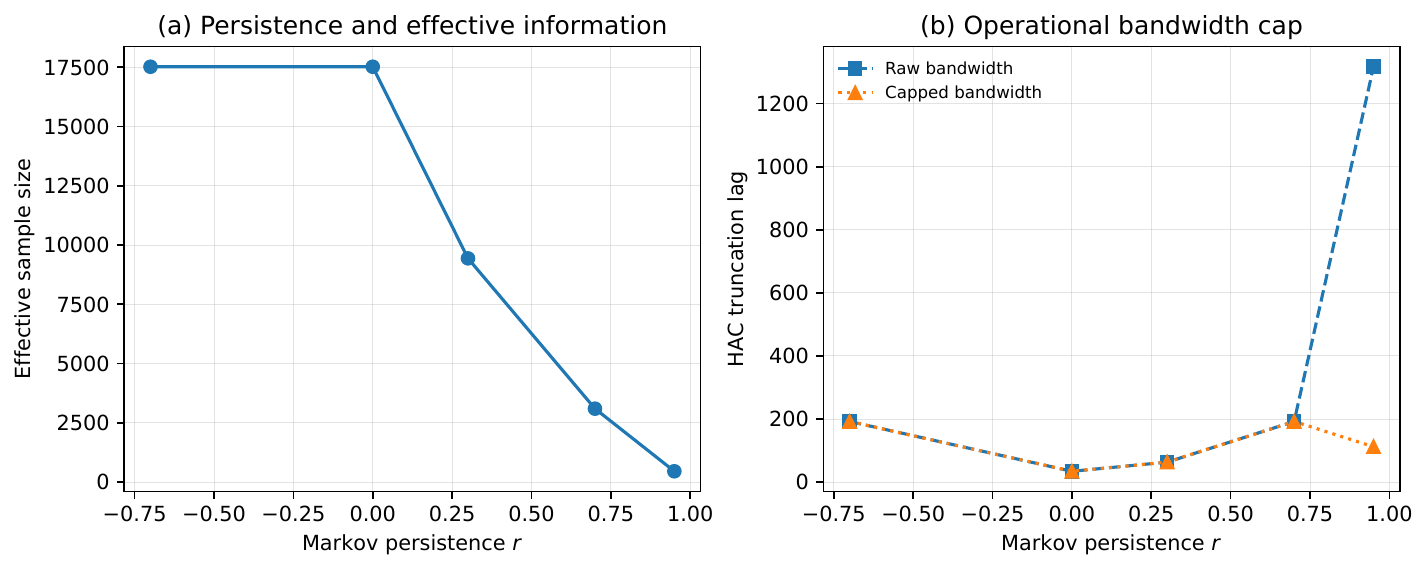}
\caption{Persistence, effective sample size, and candidate-specific HAC bandwidths for $T=17{,}520$ half-hours. The effective-sample-size curve uses the scalar AR(1) variance-inflation diagnostic $T(1-r)/(1+r)$ for positive persistence and is capped at $T$ for negative persistence. The capped bandwidth is the nearest integer to the minimum of $b_{0T}(r)$, $T/4$, and $0.25T_{\mathrm{eff}}(r)$, so cap binding is visible as a finite-sample warning rather than as a recommendation for an extremely long HAC window.}
\label{fig:bandwidth-diagnostics}
\end{figure}

The LCL calibrated implementation also examines sensitivity to the HAC truncation lag. Table~\ref{tab:bandwidth_sensitivity} recomputes the residualized aggregate cumulative-target selector under bandwidths equal to one quarter, one half, one, two, and four times the capped candidate-specific rule. The selected persistence remains at the upper fine-grid endpoint in all five cases, while the relative risk of iid assignment rises from 3.43 to 4.36 as the longer bandwidths emphasize persistent low-frequency residual variation. This small diagnostic supports the operational interpretation of the bandwidth rule: the cap is a finite-sample safeguard, and the substantive calibrated recommendation is not driven by a single truncation lag.

\begin{table}[t]
\centering
\caption{Bandwidth sensitivity of the calibrated cumulative-target selector.}
\label{tab:bandwidth_sensitivity}
\begin{tabular}{ccccc}
\toprule
Bandwidth multiplier & $b_T(0.90)$ & selected $r^\star$ & iid rel. risk & $r=0.90$ rel. risk \\
\midrule
0.25 & 58 & 0.90 & 3.428 & 1.000 \\
0.50 & 116 & 0.90 & 3.939 & 1.000 \\
1.00 & 231 & 0.90 & 4.174 & 1.000 \\
2.00 & 462 & 0.90 & 4.295 & 1.000 \\
4.00 & 924 & 0.90 & 4.356 & 1.000 \\
\bottomrule
\end{tabular}
\begin{minipage}{0.92\linewidth}
Notes: The table recomputes the residualized aggregate LCL cumulative-target selector under Bartlett bandwidths equal to one quarter, one half, one, two, and four times the capped candidate-specific default. The persistence grid is $[-0.90,0.90]$ in increments of 0.01. Relative risks are normalized by the oracle risk under the same bandwidth multiplier.
\end{minipage}
\end{table}

\subsection{Target-specific risk and persistence choice}
The design criterion for a contrast $c'g$ is $\mathcal R_c(r)=c'V_H(r)c$. In the homoskedastic benchmark,
\begin{equation}
    c'Q_H(r)^{-1}c
    =\frac{4}{1-r^2}\left[c_0^2+c_H^2+(1+r^2)\sum_{j=1}^{H-1}c_j^2-2r\sum_{j=1}^{H}c_{j-1}c_j\right].
    \label{eq:closed_form_risk}
\end{equation}
Define $a=\sum_{j=0}^{H}c_j^2$, $b=\sum_{j=1}^{H-1}c_j^2$, and $d=\sum_{j=1}^{H}c_{j-1}c_j$. Then
\begin{equation}
    \mathcal R_c(r)=4\sigma^2\frac{a+br^2-2dr}{1-r^2}.
    \label{eq:Rc}
\end{equation}

\begin{proposition}[Closed-form optimal persistence with feasible-set constraints]\label{prop:optimal}
Let $c\ne0$, let $m=a+b$, and consider $r_c^\star\in\arg\min_{r\in\mathcal R}\mathcal R_c(r)$ for a feasible interval $\mathcal R=[\underline r,r_{\max}]\subset(-1,1)$. If $d=0$, the unconstrained minimizer is $r=0$. If $0<|d|<m/2$, the unique stationary point in $(-1,1)$ is
\[
    r^{\mathrm{int}}_c=\frac{m-\sqrt{m^2-4d^2}}{2d},
\]
and the constrained optimum is the projection of $r^{\mathrm{int}}_c$ onto $\mathcal R$. If $|d|=m/2$, the risk has no finite minimizer in $(-1,1)$; its infimum is approached as $r\to \operatorname{sign}(d)$. In that boundary case, the constrained optimum is the feasible endpoint closest to $\operatorname{sign}(d)$. The case $|d|>m/2$ cannot occur for horizon-weight vectors because $|d|\le m/2$.
\end{proposition}

\begin{corollary}[Target dependence of optimal switchback]\label{cor:target_dependence}
The optimal Markov persistence depends on the horizon-weight vector. If $c=e_j$, so the target is a single horizon effect, then $d=0$ and $r_c^\star=0$. If $c=(1,\ldots,1)$, so the target is a cumulative effect, then $|d|=(a+b)/2$ and the variance-only criterion is minimized at the upper feasible endpoint. If $c$ is oscillating, then $d<0$ and alternating assignments can be optimal.
\end{corollary}

The boundary case in Proposition~\ref{prop:optimal} is substantive rather than pathological. For cumulative targets, the variance-only benchmark will always choose the largest admissible persistence, so the definition of $r_{\max}$ is part of the design. We therefore set $\mathcal R$ before evaluating the target rule by translating operational constraints into persistence bounds. For balanced designs, an upper bound on expected run length $L_{\max}$ gives $r_{\max}\le1-2/L_{\max}$. For unbalanced event designs, active and inactive caps give $r_{\max}\le1-\{(1-p_0)L_{1,\max}\}^{-1}$ and $r_{\max}\le1-\{p_0L_{0,\max}\}^{-1}$. A minimum expected number of active spells $N_{1,\min}$ over a planned length $T$ gives the additional diagnostic constraint
\[
    T p_0(1-p_0)(1-r)\ge N_{1,\min}.
\]
When a target lands on the boundary, we report the endpoint choice and treat sensitivity to a tighter endpoint as a design diagnostic, not as a new estimand.

\begin{lemma}[Active-spell lower bound]\label{lem:spell_lower_bound}
Let $N_{1T}=\sum_{t=2}^{T}\mathbf 1\{A_{t-1}=0,A_t=1\}$ be the number of active spells initiated in a stationary two-state Markov assignment with active share $p$ and persistence $r$. Then
\[
    T^{-1}N_{1T}\to p(1-p)(1-r)\quad a.s.
\]
For every fixed compact interior class with $p\in[\underline p,\bar p]$ and $r\le r_{\max}<1$, there is a constant $c>0$ such that $N_{1T}\ge cT$ with probability approaching one uniformly over the class. If $r=r_T$ approaches one, a transparent sufficient many-spell scaling is $T p(1-p)(1-r_T)\to\infty$; if this quantity is bounded, the design has only a bounded expected number of active episodes and normal approximations are not supported by a many-episode argument even when the information matrix is nonsingular.
\end{lemma}

\begin{corollary}[Linear reporting transforms]\label{cor:linear_transforms}
Let the object reported by the researcher be a lower-dimensional linear transform of the response curve,
\[
    \psi=L g,
\]
and let the reporting loss be represented by a positive semidefinite matrix $W_\psi$. Then the design risk for $\psi$ can be written in the HW-LP form
\[
    \operatorname{tr}\{W_\psi L V_H(r)L'\}
    =\operatorname{tr}\{L'W_\psi L V_H(r)\}.
\]
Thus cumulative effects, average effects over a delayed window, endpoint effects, policy-path summaries, and full-curve reporting can all be represented by choosing the effective horizon-weight matrix $W=L'W_\psi L$. The design problem therefore depends on the reported dynamic object, not on the particular parameterization used to describe it.
\end{corollary}

The scalar contrast criterion is the leading design object, but the same covariance logic applies to matrix-weighted curve losses, transformed response objects, and finite candidate menus. Appendix~\ref{app:design_extensions} develops these extensions in detail. It gives the frequency-domain interpretation, operational thresholds, numerical-stability triggers, higher-order and non-Markov benchmarks, finite-sequence coordinate-exchange designs, target-menu robust rules, and links to classical information-matrix criteria. The central implementation distinction is unchanged: the Toeplitz formula explains the target-shape margin, whereas field deployment should compare the covariance of the estimator and reporting object over the feasible design menu.

The appendix also quantifies when the first-order Markov benchmark is adequate. Finite-sequence diagnostics compare population and realized information matrices, target-menu calculations report the cost of using a common persistence across several reporting objects, and spectral diagnostics identify high-persistence, large-horizon configurations that require regularization or a narrower target menu. These diagnostics qualify the benchmark rather than alter the core closed-form result.

\section{Omitted-carryover misspecification}\label{sec:robust}
A persistent design can be variance-efficient for a smooth cumulative target, but persistence also correlates nearby treatment lags. This creates a second risk: if the analyst estimates a model with maximum horizon $H$ but the true carryover extends beyond $H$, omitted treatment lags can bias the estimated dynamic response. In geometric-tail benchmarks, this omitted-carryover bias is governed by the buffer horizon $K$, not by the assignment persistence: the cumulative omitted-tail bias is essentially invariant to $r$ over the persistence range, while persistence reallocates variance across horizons and the buffer horizon controls bias. The buffer controls how much unmodeled tail can enter the reported contrast.

Let $X_t=(\widetilde A_t,\ldots,\widetilde A_{t-H})'$ collect the included lags and let $Z_t=(\widetilde A_{t-H-1},\ldots,\widetilde A_{t-H-K})'$ collect $K$ omitted lags. Suppose
\[
    Y_t=X_t'g+Z_t'\gamma+\varepsilon_t,
\]
but the analyst estimates the truncated model using only $X_t$.

\begin{proposition}[Omitted-carryover bias under Markov switchback]\label{prop:omitted_bias}
Let $Q_{XX}(r)=E[X_tX_t']$ and $Q_{XZ}(r)=E[X_tZ_t']$. The population coefficient from the truncated regression of $Y_t$ on $X_t$ is
\[
    \bar g_H(r)=g+Q_{XX}(r)^{-1}Q_{XZ}(r)\gamma.
\]
If $\|\gamma\|_2\le M$, the worst-case squared bias for the contrast $c'g$ is
\[
    M^2\left\|Q_{ZX}(r)Q_{XX}(r)^{-1}c\right\|_2^2.
\]
Under the Markov Toeplitz benchmark,
\[
    Q_{ZX}(r)Q_{XX}(r)^{-1}c=c_H(r,r^2,\ldots,r^K)',
\]
and therefore
\begin{equation}
    \sup_{\|\gamma\|_2\le M}\left(c'Q_{XX}(r)^{-1}Q_{XZ}(r)\gamma\right)^2
    =M^2 c_H^2\sum_{m=1}^{K}r^{2m}.
    \label{eq:closed_form_tail_bias}
\end{equation}
\end{proposition}

The formula has two interpretations. If the omitted tail is fixed, the expression in Proposition~\ref{prop:omitted_bias} is a population bias diagnostic: it measures how sensitive the reported contrast is to carryover beyond the estimated horizon. If the omitted tail is local, $\gamma=\delta/\sqrt T$ with $\|\delta\|_2\le M_0$, then it yields a dimensionally consistent scaled-MSE design criterion,
\begin{equation}
    \mathcal S^{\mathrm{local}}_{c}(r;M_0)=c'V_H(r)c+M_0^2\left\|Q_{ZX}(r)Q_{XX}(r)^{-1}c\right\|_2^2.
    \label{eq:local_bias_risk}
\end{equation}
Equivalently, $T\,E\{c'(\hat g-g)\}^2$ is approximated by \eqref{eq:local_bias_risk} under the local-bias sequence. The local sequence is used because it puts variance and misspecification sensitivity on the same asymptotic scale; a fixed nonlocal tail would dominate the variance term and turn the design problem into pure worst-case bias avoidance. In applications $M_0$ should be pre-specified. We use three defaults: $M_0=0$ for the variance-only benchmark, a pilot-calibrated value based on the Euclidean norm of the last observed or fitted response-tail block, and a sensitivity grid such as $M_0\in\{0,0.5,1,2\}$ times that pilot norm. If no pilot tail is credible, the recommended report is the frontier over $M_0$ rather than a single robust design. For a fixed nonlocal tail, we report the second term as a bias sensitivity rather than as a finite-sample MSE component. The closed-form tail expression shows which targets are most vulnerable: in the Markov benchmark, omitted carryover beyond the estimated horizon affects the reported contrast only through the last included reporting weight $c_H$.

\begin{corollary}[Buffered path-adjusted local projections]\label{cor:buffer}
Suppose the researcher wants to report a target over horizons $0,\ldots,K_0$ but estimates the path-adjusted regression through a longer horizon $H>K_0$, assigning zero reporting weight to buffer horizons $K_0+1,\ldots,H$. Then $c_H=0$, and omitted carryover beyond $H$ has zero asymptotic bias for the reported target in the Markov Toeplitz benchmark.
\end{corollary}

Corollary~\ref{cor:buffer} gives a practical robustness rule: estimate beyond the horizons that will be reported. Buffer horizons can absorb residual carryover while leaving the reported contrast unchanged. They are not free: increasing the estimated horizon enlarges the nuisance part of the regression and can raise finite-sample variance through the inverse design matrix. The default procedure is to start with $H=K_0+2$, compute the calibrated target variance and the local-bias term, and then increase $H$ until the bias reduction is small relative to the induced variance increase. Under these empirical defaults, the search stops when the target variance rises by more than 10 percent or when the condition number of the residualized assignment information matrix exceeds the pre-specified cap. When $M_0=0$, \eqref{eq:local_bias_risk} reduces to the variance-only HW-LP criterion. When $M_0>0$, high persistence can become less attractive because it increases the correlation between included and omitted treatment lags.

\begin{figure}[t]
    \centering
    \includegraphics[width=0.92\textwidth]{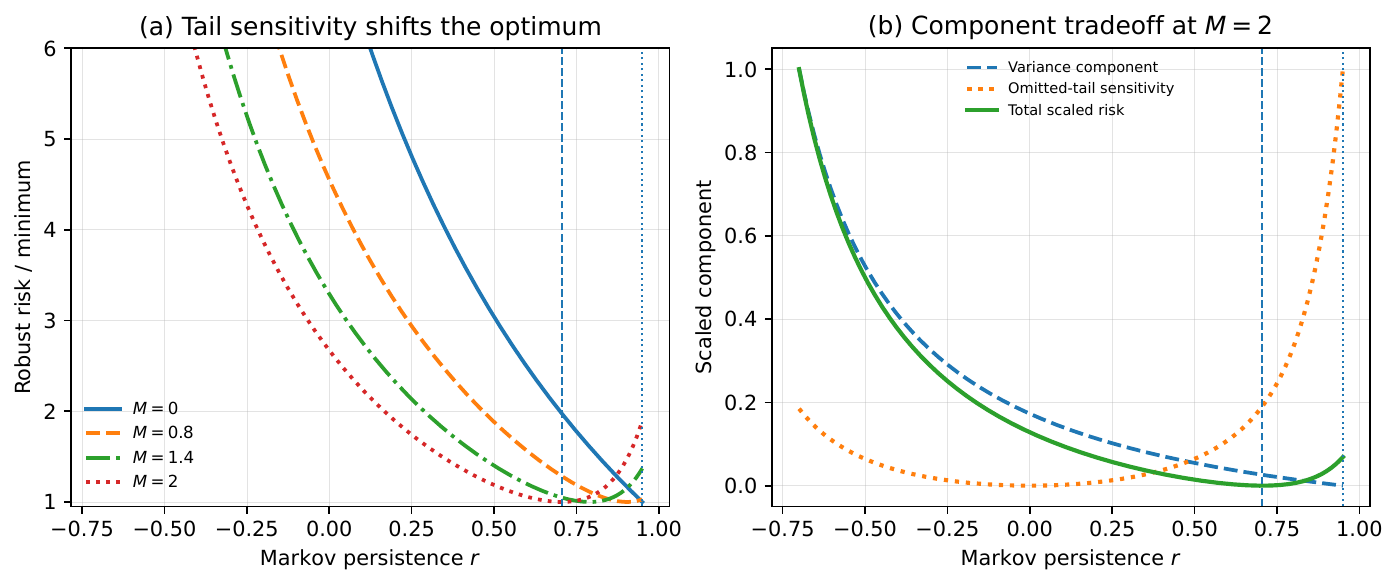}
    \caption{Unknown carryover creates a bias--variance design tradeoff. The left panel shows how increasing the omitted-tail radius shifts the cumulative-target optimum away from the variance-only boundary. The right panel decomposes the scaled components at the reported sensitivity radius; the vertical lines mark the variance-only and bias-augmented optima.}
    \label{fig:carryover-tradeoff}
\end{figure}

Figure~\ref{fig:carryover-tradeoff} reports the resulting bias--variance tradeoff.

\section{Estimation, design selection, and inference}\label{sec:simulations}
Implementation and inference require a reporting target, feasible design menu, and estimator class to be specified in advance. Pointwise inference, simultaneous full-curve reporting, post-selection inference, and estimator choice address distinct objects: fixed-target coverage, curve-level coverage, and design-selection uncertainty. The benchmark comparisons evaluate the design criterion without introducing a separate estimator. The benchmark DGP is
\[
    Y_t=\sum_{\ell=0}^{H}g_\ell A_{t-\ell}+u_t,
\]
with $H=8$. For scalar targets, the reported metric is the benchmark risk for $c'g$ under the relevant assignment persistence. This avoids comparing the same design under independent simulation draws: whenever the HW-LP rule selects $r=0$, it is displayed as the same iid design.

\subsection{Calibration and Monte Carlo design}
The analytical risk curves and benchmark tables are deterministic functions of the design covariance formulas. The Low Carbon London calibration uses 10,000 Monte Carlo replications per scenario-design pair and common random numbers across designs within each scenario. The replication count was fixed ex ante.

Table \ref{tab:mc-risk} reports relative target risk for five scalar response objects. The table focuses on scalar targets and excludes full-curve reporting, which is a matrix-weighted objective rather than a scalar contrast. Full-curve reporting is evaluated separately below with $W=I$.

\begin{table}[t]
    \centering
    \caption{Benchmark target risk by design.}
    \label{tab:mc-risk}
    \begin{tabular}{lllllll}
\toprule
Scenario & iid & alternating & moderate & persistent & HW-LP & HW-LP r \\
\midrule
Immediate & 1.000 & 1.961 & 1.099 & 1.961 & 1.000 & 0.000 \\
Cumulative & 7.468 & 38.447 & 4.404 & 2.001 & 1.000 & 0.950 \\
Delayed window & 1.342 & 6.375 & 1.017 & 1.464 & 1.000 & 0.382 \\
Rebound & 1.186 & 4.809 & 1.000 & 1.552 & 1.000 & 0.314 \\
Alternating & 7.468 & 2.001 & 13.158 & 38.447 & 1.000 & -0.950 \\
\bottomrule
\end{tabular}

    \begin{minipage}{0.92\linewidth}
    Notes: Entries are relative benchmark risks for the scalar target in each scenario. When the HW-LP rule selects $r=0$, it is identical to iid assignment and is displayed in the iid/HW-LP column.
    \end{minipage}
\end{table}

The benchmark results align with the theory. Isolated horizon effects favor balanced iid assignment. Smooth cumulative targets favor high persistence. Delayed and rebound-shaped targets favor intermediate or target-specific persistence. Oscillating contrasts favor negative persistence.

The closed-form Markov results use a homoskedastic benchmark to expose the assignment information matrix. In applications, serial dependence and calendar controls should be handled through the estimated covariance matrix of the residualized path-adjusted local-projection moments. The feasible implementation is therefore to compute $\hat V(r)$ from a pilot, HAC calculation, residual bootstrap, or calibrated baseline model and choose $r$ by minimizing $\operatorname{tr}\{W\hat V(r)\}$ over the feasible design class. The Low Carbon London evaluation below has two layers. The main scalar-target table reports the retained closed-form benchmark comparator and finite-candidate performance under realistic load dynamics. The delayed-reduction case, where the closed-form rule is not the retained best design, is treated as evidence that calibrated covariance selection is needed when residualized finite-sample covariance differs from the homoskedastic Toeplitz benchmark. The empirical evaluation retains the full response vector in each Monte Carlo draw, estimates design-specific covariance matrices, computes paired Monte Carlo standard errors, and evaluates observed tariff active indicators as fixed assignment schedules.

\subsection{Pointwise and near-unit inference}
The design criterion is an ex ante precision criterion, but it maps directly into standard inference for the reported target. Once the experiment is run under a pre-specified design, the contrast estimator satisfies
\[
    \sqrt T\{c'(\hat g-g)\}\Rightarrow N(0,c'V_\pi c),
\]
under the same score-covariance conditions used for the design calculation. A confidence interval for $c'g$ is therefore
\[
    c'\hat g \pm z_{1-\alpha/2}\sqrt{c'\widehat V_\pi c/T},
\]
with $\widehat V_\pi$ computed using the residualized HAC, bootstrap, or pilot-calibrated covariance appropriate for the realized assignment.

The variance interpretation should match the empirical design. The theoretical $V_\pi$ is a super-population long-run covariance for the stationary assignment-and-score process. A fixed-schedule diagnostic, by contrast, conditions on the realized assignment path and measures information in that finite design matrix. The Low Carbon London exercise is intermediate: the baseline load path is fixed from the observed data, while randomized assignments and injected responses are used to compare designs. In all cases the HW-LP rule optimizes the covariance of the estimator that will be used for inference, but the source of randomness should be stated in the pre-analysis plan and in the table notes. The design problem is thus separated from the inference problem: HW-LP chooses the assignment path to improve the covariance of estimators that are then reported with standard confidence intervals, bands, or design-based diagnostics.

Randomization-based inference is the primary finite-sample route in near-boundary regimes when the design produces few assignment episodes. Because the researcher controls the assignment law, a sharp-null Fisher test can be implemented by repeatedly drawing $A_{1:T}^{\ast}$ from the locked switchback design, recomputing the chosen statistic $\widehat{\theta}_c(A^{\ast},Y^{\mathrm{obs}})$, and comparing it with the observed statistic \citep{BojinovShephard2019,AtheyEcklesImbens2018}. Under the sharp null and the pre-specified exposure mapping, this test is exact conditional on the observed outcomes. Neymanian finite-population intervals or conservative randomization variances can also be reported when the target is an average over a fixed schedule of potential outcomes \citep{BojinovRambachanShephard2021,LinDing2025}. For $r>0.95$, a binding HAC cap, or an active-spell lower tail below the pre-specified threshold, the pre-analysis plan should make the Fisher or Neymanian randomization analysis primary and treat normal/HAC intervals as descriptive. Finite-sample covariate balance should also be logged for persistent designs. A simple diagnostic is the maximum standardized treated-control imbalance over pre-specified calendar covariates, such as hour-of-day, weekday, weekend, holiday, and season indicators; sequential rerandomization or rejection sampling can be added as a feasibility constraint when this imbalance exceeds a pre-specified threshold.

The inference route depends on the realized design diagnostics. For interior Markov designs with many active spells and a non-binding HAC cap, fixed-design HAC intervals or simultaneous bands are the default. When persistence exceeds 0.95, the HAC cap binds, or the active-spell lower tail falls below the pre-specified threshold, the primary finite-sample report should be a Fisher randomization test or a Neymanian randomization interval under the locked assignment law. Pilot-selected designs use fixed-design inference conditional on the selected design, with pilot uncertainty reported separately. If the same outcomes are used for tuning and inference, the protocol should use sample splitting, a simultaneous multiverse band, or selection-adjusted intervals. The pre-analysis plan should record the candidate grid, bandwidth rule, randomization draws, balance diagnostics, model menu, and family-wise error rule for the chosen route.

The fixed-$r$ theory does not provide local-to-unity inference guarantees. If the selected persistence is close to one, the effective number of independent assignment episodes can be small even when $T$ is large. This problem is analogous in spirit to persistent-regressor inference concerns in time-series LPs, although here the persistence comes from the assignment process rather than from the outcome dynamics \citep{Phillips1998,Mikusheva2007,PesaventoRossi2007}. Table~\ref{tab:near_unit_coverage} reports a finite-sample diagnostic under the homoskedastic Markov benchmark with $T=17{,}520$, $H=8$, cumulative target weights, and 2,000 Monte Carlo assignment draws. Under this ideal benchmark, conditional-assignment and Markov plug-in intervals remain close to nominal coverage, but the active-spell count falls sharply: at $r=0.99$ the fifth percentile is only 36 active spell starts. The operational implication is therefore not a new asymptotic correction, but a pre-analysis requirement: if $r>0.95$ or the capped bandwidth binds, the protocol should report a design-specific coverage, balance, and active-spell simulation before field deployment.

\begin{table}[t]
    \centering
    \caption{Near-unit-root finite-sample diagnostic for cumulative-target inference.}
    \label{tab:near_unit_coverage}
    \begin{tabular}{rrrrrrr}
\toprule
$r$ & $T_{\mathrm{eff}}$ & capped $b_T$ & mean spells & p5 spells & cond-SE cov. & Markov-SE cov. \\
\midrule
0.90 & 922 & 231 & 438.3 & 415 & 95.2\% & 95.1\% \\
0.95 & 449 & 112 & 219.3 & 203 & 94.9\% & 94.9\% \\
0.98 & 177 & 44 & 87.2 & 76 & 94.0\% & 94.0\% \\
0.99 & 88 & 22 & 43.7 & 36 & 94.5\% & 94.5\% \\
\bottomrule
\end{tabular}
    \begin{minipage}{0.92\linewidth}
    Notes: The table uses 2,000 Monte Carlo assignment draws under the balanced Markov benchmark, $T=17{,}520$, $H=8$, and cumulative target weights. The estimator draw follows the conditional Gaussian OLS law given the realized assignment. ``Cond-SE'' uses the realized assignment Gram matrix; ``Markov-SE'' uses the benchmark Markov information matrix. The table diagnoses many-spell behavior but does not establish local-to-unity validity.
    \end{minipage}
\end{table}

\subsection{Simultaneous and estimator-specific inference}
A scalar target can use the pointwise covariance formula above, but a reported response curve requires a simultaneous band. We use a Gaussian or residual-multiplier sup-$t$ critical value based on the estimated joint covariance, with Bonferroni as a conservative fallback when that covariance is unstable. Appendix~\ref{app:inference_details} gives the construction, finite-sample coverage diagnostics, power and minimum-detectable-effect calculations, sequential-monitoring cautions, and post-selection conditions.

The confirmatory design, target, horizon, covariance estimator, and tuning rules should be fixed using pilot information or a separated training sample. When the same confirmatory outcomes are used to select among designs or tuning parameters, valid reporting requires sample splitting, simultaneous inference over the pre-specified menu, or an explicit post-selection procedure. The appendix states these requirements and separates them from the estimator-efficiency comparison that follows.

The HW-LP criterion optimizes the design for a specified estimator and reporting object. The closed-form formulas use the path-adjusted least-squares LP because it is transparent, design-invariant under Assumption~\ref{ass:po_identification}, and robust to misspecifying a full dynamic outcome model. They are not a semiparametric Cram\'{e}r--Rao bound over all possible estimators.

If the researcher is willing to specify more structure, the same design logic applies with a different estimator-specific covariance matrix. With a known serial-error covariance $\Sigma$, for example, the feasible GLS distributed-lag estimator has covariance proportional to $(X'\Sigma^{-1}X)^{-1}$ rather than the OLS/HAC covariance $Q^{-1}\Omega Q^{-1}$. With a correctly specified VAR or transfer-function model, system estimation can dominate unrestricted LPs in variance, while unrestricted LPs can have lower specification bias when the model is wrong \citep{PlagborgMollerWolf2021,LiPlagborgMollerWolf2024}. Smooth, Bayesian, ridge, or other regularized LP estimators similarly replace $\widehat V(r)$ by the covariance or posterior risk of the regularized estimator \citep{BarnichonBrownlees2019,FerreiraMirandaAgrippinoRicco2025}. When a ridge-regularized estimator is used to stabilize an ill-conditioned design, the relevant object is the regularized target MSE, not variance alone: $\operatorname{tr}\{W\operatorname{MSE}(r,\lambda)\}=\operatorname{tr}\{WV(r,\lambda)\}+b(r,\lambda)'Wb(r,\lambda)$, where $b(r,\lambda)$ is the regularization bias for the reported object. The HW-LP selector should then minimize this pre-specified $(r,\lambda)$ MSE grid; the closed-form variance rule is the $\lambda=0$ benchmark. Thus a better estimator under iid assignment could dominate path-adjusted OLS under a persistent assignment for some data-generating processes. The contribution here is narrower: once the reporting object and estimator class are specified, HW-LP gives the assignment rule that minimizes that estimator's target risk within the candidate design class.

Table~\ref{tab:gls_ols_optimality} gives a compact numerical illustration. With known AR(1) residual correlation $\rho=0.5$, the OLS and GLS criteria agree for isolated and smooth cumulative targets, but they differ for the delayed-window contrast: the OLS benchmark selects $r=0.38$, while the feasible-GLS criterion selects $r=0.55$. Choosing the OLS persistence for the GLS estimator raises the GLS target risk by about nine percent. The stylized example makes the estimator-specific point concrete: persistence is not a property of the assignment class alone; it is a property of the reporting object and the planned estimator.

\begin{table}[t]
\centering
\caption{OLS and feasible-GLS persistence rules under AR(1) residual correlation.}
\label{tab:gls_ols_optimality}
\begin{tabular}{lccccc}
\toprule
Target & OLS $r^\star$ & GLS $r^\star$ & OLS loss at GLS $r$ & GLS loss at OLS $r$ & GLS iid rel. risk \\
\midrule
Immediate & 0.00 & 0.00 & 1.000 & 1.000 & 1.000 \\
Cumulative & 0.90 & 0.90 & 1.000 & 1.000 & 5.326 \\
Delayed & 0.38 & 0.55 & 1.095 & 1.090 & 1.831 \\
Rebound & 0.31 & 0.32 & 1.000 & 1.000 & 1.195 \\
\bottomrule
\end{tabular}
\begin{minipage}{0.92\linewidth}
Notes: The illustration uses $H=8$, the fine grid $[-0.90,0.90]$, and a stationary AR(1) error correlation $\rho=0.5$ with known covariance. OLS uses the homoskedastic Markov information matrix. GLS uses the large-sample precision matrix induced by the AR(1) inverse covariance. Ratios are normalized by the target-specific optimum for the same estimator.
\end{minipage}
\end{table}

Operationally, this means the design stage should not hide estimator choice. A protocol that plans to report GLS, VAR-based IRFs, smooth LPs, or Bayesian LPs should compute $\operatorname{tr}\{W\widehat V_{\mathcal E}(d)\}$ for that estimator $\mathcal E$, not reuse the homoskedastic OLS closed form. The OLS formulas remain useful as a benchmark because they isolate the assignment-spectrum channel and provide a diagnostic for when persistence itself is the dominant design margin.

The same estimator-specific interpretation covers GMM. A stacked path-adjusted LP can be written as a moment problem with moment vector $m_t(\theta)=(Z_t u_{t,h}(\theta))_{h=0}^H$, where $Z_t$ contains the residualized current and lagged assignment terms and any controls. Ordinary path-adjusted least squares corresponds to a particular identity or equation-by-equation weighting of these moments. Optimally weighted GMM, continuously updated GMM, or empirical likelihood would instead use the long-run covariance of the stacked moments, and the resulting design risk is the GMM sandwich risk rather than the OLS HW-LP risk \citep{Hansen1982,HansenHeatonYaron1996,Owen1988}. These estimators can dominate identity-weighted LP when the moment covariance is well specified, and their overidentifying-restriction diagnostics have design-dependent power. In such an analysis, the switchback persistence should be chosen by minimizing the planned GMM risk; the closed-form OLS persistence rule is the benchmark for identity-weighted path-adjusted LP, not an efficiency bound for all moment estimators.

A Bayesian decision-theoretic design is another valid specialization of the same principle. If the analyst has credible prior information about smoothness, sign restrictions, monotone decay, or cross-platform shrinkage, the design objective should be the prior or posterior-predictive expected loss, for example
\[
    d_\pi^\star\in\arg\min_{d\in\mathcal D} E_\pi\{L(\widehat{\theta}_d,\theta)\},
\]
where $L$ may be squared error for a reported contrast, a threshold loss for a regulatory decision, or a utility loss for a platform decision. Bayesian $D$- or $A$-optimal design criteria and empirical-Bayes switchback rules are natural counterparts \citep{ChalonerVerdinelli1995,AtkinsonDonevTobias2007,XiongChinTaylor2024}. Such priors can change the optimal persistence because they shrink or constrain the response curve across horizons. The main criterion remains frequentist and reporting-object based, but the implementation rule is modular: replace $\widehat V(d)$ or the loss function by the posterior risk implied by the planned Bayesian estimator.

Robust LP estimators fit into the same modular template. Half-hourly load, revenue, and clickstream outcomes can have heavy tails, outages, holidays, or equipment failures, so the $4+\delta$ moment condition in Assumption~\ref{ass:hac} is an assumption to be checked rather than a free empirical fact. If the pre-analysis plan specifies median, quantile, Huber, or trimmed LP estimation, the design criterion should use the corresponding influence-function or sandwich covariance instead of the OLS covariance \citep{Huber1964,KoenkerBassett1978,HampelEtAl1986}. Likewise, pilot preprocessing rules---outage flags, holiday flags, winsorization thresholds, or trimming fractions---should be fixed before the confirmatory experiment and repeated in sensitivity checks. The closed-form OLS rule then becomes a benchmark for clean finite-variance designs, while robust or distributionally conservative designs replace $\widehat V(d)$ by the risk of the robust estimator or by a worst-case covariance over a declared ambiguity set.

We study average response targets. If the analyst instead reports a quantile local projection at a fixed quantile $\tau$ and the standard sparsity condition holds, the leading sandwich variance has the form $\tau(1-\tau)f_\varepsilon\{F_\varepsilon^{-1}(\tau)\}^{-2} Q_H(r)^{-1}$ under the same lagged-assignment design. Thus, for a fixed horizon contrast $c$, this homoskedastic quantile-LP benchmark selects the same persistence as the OLS HW-LP rule; the scalar density factor changes scale, not the optimal $r$. This equivalence should not be overread. Peak-load, tail-risk, distributional-welfare, or conditional quantile treatment-effect reporting objects generally change the target vector, aggregation level, subgroup weights, or covariance estimator, and should be designed by replacing $V(d)$ with the covariance or loss for that distributional estimand rather than by reusing the mean-LP rule mechanically.

Applied researchers often report the full response curve rather than a single scalar contrast. This is a different design objective from the scalar cumulative target. If equal marginal precision across horizons is desired, the benchmark weight is $W=I$ and the criterion is $\operatorname{tr}\{V(r)\}$. Under the homoskedastic Markov benchmark this criterion favors iid assignment, not cumulative-target persistence. The numerical full-curve comparison uses the same candidate design menu, while Section~\ref{sec:simulations} gives the inference rule and Appendix~\ref{app:inference_details} gives the simultaneous-band construction.

\section{Empirical calibration: Low Carbon London}\label{sec:lcl}

We evaluate the design rules in a high-frequency residential electricity setting calibrated to the Low Carbon London smart-meter data \citep{StrbacEtAl2024}. The exercise is semi-synthetic: it is not an estimate of the causal effect of the observed dynamic time-of-use tariff. Instead, the observed data provide realistic half-hourly residential-load dynamics, and known dynamic response paths are injected into those dynamics so that different assignment designs can be compared against known targets.

The empirical analysis distinguishes three design questions. A balanced-persistence comparison isolates the persistence margin in the same finite design class as the analytical benchmark. Sparse active-share designs show how target-specific persistence changes under an event budget. Historical tariff schedules are evaluated as fixed design paths using the same information criterion. The distinction separates the persistence choice for a given target, its modification under a sparse budget, and the informativeness of a previously implemented schedule for that target. Only the first two blocks are directly comparable in semi-synthetic target-MSE units; the historical-schedule block is a fixed-schedule information diagnostic rather than a replay-based loss comparison.

\subsection{Data and calibration design}
The calibrated evaluation uses the full 2013 trial year and the normal-tariff households as the baseline source. The outcome is the group-average half-hourly electricity consumption series in kWh per half-hour. The aggregate panel contains 17,520 half-hour observations. The normal-tariff group has roughly 4,000--4,400 available household readings for most half-hours, and the dynamic-tariff group has roughly 1,000--1,100 available readings for most half-hours. Aggregation is by timestamp and tariff group; available readings are averaged within each group and timestamp. The calibrated normal-tariff baseline series has mean 0.214 kWh per half-hour and standard deviation 0.076.

The pre-specified Monte Carlo design uses 10,000 replications for each scenario-design pair and common random numbers across designs within each scenario. The horizon is $H=8$, corresponding to four hours at half-hour frequency. The four response scenarios are immediate reduction, persistent reduction, delayed reduction, and rebound. Scenario names describe the injected response path used for evaluation; design selection uses only the reporting weights $c$, the feasible persistence set, and the assignment covariance, not the injected magnitudes of $g$.

The evaluation fixes 10,000 replications per scenario-design pair, common random numbers across designs within each scenario, the horizon $H=8$, the balanced candidate set $r\in\{-0.7,0,0.3,0.7\}$ plus target-specific HW-LP candidates, sparse active shares $p_0\in\{0.05,0.10,0.14\}$, and the observed tariff indicators used only for fixed-schedule diagnostics. The generic-MSE selector summarized in Table~\ref{tab:selector-reconciliation} uses the same retained finite design menu as the target-specific selector and changes only the objective from $c'\widehat V_d c$ to $E\|\hat g-g\|^2$. The scenario menu was fixed before the full run to span four qualitatively distinct response shapes--front-loaded, smooth cumulative, delayed, and sign-changing rebound--rather than being selected from observed treatment effects. Common random numbers are implemented by using the same innovation stream within a scenario and changing only the design transformation; this is why paired comparisons have much smaller Monte Carlo noise than unpaired comparisons. The replication count was fixed ex ante.

\subsection{Balanced and calibrated design selection}
The primary comparison keeps the balanced Markov benchmark to isolate the persistence margin. We compare iid assignment, negatively persistent assignment with $r=-0.7$, a moderate Markov switchback with $r=0.3$, a persistent Markov switchback with $r=0.7$, and the application-feasible closed-form HW-LP rule. The main metric is target MSE normalized so that the best design in each scenario has value one. Because common random numbers are used, the table also reports paired Monte Carlo standard errors for the HW-LP design relative to the best design and a 95 percent tie flag.

\begin{table}[H]
\centering
\caption{Relative target MSE in the Low Carbon London empirical calibration.}
\label{tab:lcl-key-results}
\begin{tabular}{lrrrrrrrl}
\toprule
Scenario & iid & alt. & mod. & pers. & HW-LP & HW-LP $r$ & Paired SE & Tie \\
\midrule
Immediate & 1.000 & 1.953 & 1.088 & 1.941 & 1.000 & 0.000 & 0.000 & yes \\
Persistent & 2.646 & 12.38 & 1.657 & 1.000 & 1.000 & 0.700 & 0.000 & yes \\
Delayed & 4.784 & 27.17 & 2.722 & 1.000 & 2.224 & 0.382 & 3.1e-08 & no \\
Rebound & 1.084 & 5.397 & 1.000 & 2.311 & 1.025 & 0.314 & 7.0e-09 & yes \\
\bottomrule
\end{tabular}
\begin{minipage}{0.93\linewidth}
Notes: Entries are target MSE relative to the best-performing design within each scenario. The evaluation uses the full 2013 normal-tariff aggregate baseline dynamics and 10,000 replications for each scenario-design pair. Common random numbers are used across designs within each scenario. The paired standard error is for the loss difference between the HW-LP design and the best design in that scenario. The tie column reports whether this paired difference is within two paired standard errors.
\end{minipage}
\end{table}

Table~\ref{tab:lcl-key-results} evaluates each pre-specified design under known injected response paths; selection results are reported separately below. The delayed-reduction scenario is reported as a stress test for the closed-form rule. The homoskedastic Markov benchmark selects $r=0.382$, while the retained calibrated loss is minimized by the persistent design $r=0.7$. This is a substantive negative finding for the closed-form benchmark as a field rule: the analytical formula is useful for understanding the persistence mechanism, but it should not be used mechanically when residualized finite-sample covariance, serial dependence, or calendar dynamics differ from the stylized Toeplitz benchmark. The calibrated implementation in Proposition~\ref{prop:calibrated_selector} is designed for exactly this case: the selected persistence should be based on $c'\widehat V_d c$ over the retained feasible designs.

The remaining scenarios are closer to the analytical benchmark. For the immediate-reduction target, the HW-LP rule selects $r=0$, which coincides with iid assignment and delivers the lowest target MSE. For the persistent-reduction target, the application-feasible HW-LP rule selects the upper application endpoint, $r=0.7$, and this is also the best design among the retained comparisons. For the rebound scenario, the closed-form HW-LP rule selects moderate persistence and is statistically tied with the fixed moderate design under the paired Monte Carlo comparison.

The balanced comparison is the primary evidence for scalar targets because it most closely matches the closed-form theory. The sparse-budget and historical-schedule analyses extend the same design logic to additional constraints and fixed schedules rather than replacing the balanced benchmark.

The full-covariance evaluation retains the full response vector $\hat g$ in each Monte Carlo draw and computes the empirical covariance of the path-adjusted LP estimator for each retained design. This makes it possible to evaluate the finite-candidate calibrated selector
\[
    \widehat d_c\in\arg\min_{d\in\mathcal D} c'\widehat V_d c,
\]
where $\mathcal D$ is the finite set of retained designs. The selector uses the covariance of the estimator, not the realized squared error or the injected magnitudes of the response path. Table~\ref{tab:selector-reconciliation} reports the closed-form benchmark, the calibrated selector, and the best target-MSE design side by side. The best target-MSE design is observable only in this semi-synthetic exercise; it is included to show whether the feasible covariance selector points in the same direction as the known target loss.

\begin{table}[H]
\centering
\caption{Persistence selected by the benchmark, calibrated covariance, target loss, and generic curve loss.}
\label{tab:selector-reconciliation}
\begin{tabular}{lrrrr}
\toprule
Scenario & Benchmark $r$ & Calibrated $r$ & Target-best $r$ & Generic rel. MSE \\
\midrule
Immediate & 0.000 & 0.000 & 0.000 & 1.941 \\
Persistent & 0.700 & 0.700 & 0.700 & 1.000 \\
Delayed & 0.382 & 0.700 & 0.700 & 1.000 \\
Rebound & 0.314 & 0.300 & 0.300 & 2.311 \\
\bottomrule
\end{tabular}

\begin{minipage}{0.92\linewidth}
Notes: The benchmark column reports the analytical Markov persistence. The calibrated selector minimizes $c'\widehat V_d c$ over the retained finite design menu. The target-best persistence is observable only because the semi-synthetic exercise uses known injected response paths. The final column reports target MSE under the full-curve-MSE selector relative to the target-specific minimum.
\end{minipage}
\end{table}

The calibrated selector agrees with the closed-form benchmark for the immediate, persistent, and rebound targets, but it replaces the delayed-window benchmark with the persistent design. The generic full-curve-MSE selector is efficient for persistent and delayed targets, yet it over-persistently designs the immediate and rebound scalar targets. This divergence reflects the target-specific design criterion: persistence should be chosen for the reporting object, not for a generic curve loss.

The same target-MSE comparison yields a field decision. The benchmark and calibrated candidates coincide for the immediate and persistent targets. For the delayed target, the retained benchmark proxy has relative target MSE 2.224 and should be replaced by the calibrated persistent design. For the rebound target, the benchmark proxy has relative target MSE 1.025 and remains within the pre-specified 5 percent tolerance. The known semi-synthetic target loss is used only to evaluate the rule; field selection itself uses the covariance criterion.

The paired standard errors in Table~\ref{tab:lcl-key-results} are small in raw target-MSE units because common random numbers remove the shared baseline-load component. They are not evidence that the exercise is deterministic. The relative paired standard-error scale is economically interpretable: it is near zero when HW-LP and the best retained design are the same collapsed design, about 3.4 percent in the delayed stress test, and about 2.0 percent in the rebound comparison.

\subsection{Sparse-treatment and robustness diagnostics}
The balanced comparison isolates the persistence margin, but demand-response programs usually impose an active-share budget. Table~\ref{tab:sparse-hwlp} evaluates sparse unbalanced Markov HW-LP designs with active shares $p_0\in\{0.05,0.10,0.14\}$. These designs use the same target-specific persistence logic but operate under smaller information scale $p_0(1-p_0)$. The sparse block is both a limitation check and an implementation exercise. When $p_0=0.05$, the delayed-target relative MSE is 12.09: changing persistence cannot compensate for an event budget that leaves too few active half-hours. In such settings the practical conclusion is not that one should search harder within the first-order Markov class, but that the experiment needs a larger active share, a longer trial, a stronger prior/model-based estimator, or a more modest reporting target. The HW-LP rule allocates the available information across horizons; it does not create information that the active-share budget removes.

\begin{table}[H]
\centering
\caption{Sparse active-share HW-LP designs in the Low Carbon London evaluation.}
\label{tab:sparse-hwlp}
\begin{tabular}{lrrr}
\toprule
Scenario & $p_0$ & $r$ & Rel. MSE \\
\midrule
Immediate & 0.05 & 0.000 & 5.339 \\
Immediate & 0.10 & 0.000 & 2.858 \\
Immediate & 0.14 & 0.000 & 2.136 \\
Persistent & 0.05 & 0.700 & 5.347 \\
Persistent & 0.10 & 0.700 & 2.856 \\
Persistent & 0.14 & 0.700 & 2.118 \\
Delayed & 0.05 & 0.382 & 12.093 \\
Delayed & 0.10 & 0.382 & 6.221 \\
Delayed & 0.14 & 0.382 & 4.720 \\
Rebound & 0.05 & 0.314 & 5.314 \\
Rebound & 0.10 & 0.314 & 2.754 \\
Rebound & 0.14 & 0.314 & 2.098 \\
\bottomrule
\end{tabular}
\begin{minipage}{0.78\linewidth}
Notes: Entries are target MSE relative to the best design among all retained designs for the scenario. The sparse HW-LP design fixes active share $p_0$ and chooses target-specific persistence within the unbalanced Markov class.
\end{minipage}
\end{table}

Appendix~\ref{app:lcl-details} reports five diagnostics that qualify the main comparison without changing its target-MSE interpretation: calendar-state sensitivity, historical tariff schedules evaluated as fixed assignment paths, finite design-matrix conditioning, injected-path scale, and pilot-window stability. The calendar-state exercise shows that a common persistence recommendation can range from iid assignment to about $r=0.69$ when reporting weights change across weekday and weekend states. The historical schedule diagnostic confirms that the observed tariff signals are operationally realistic but sparse and persistent, which makes some multi-horizon contrasts information-poor. These are information diagnostics rather than causal estimates of the historical tariff.

The pilot-window checks support the calibrated-selector implementation. Three-month or longer aggregate baseline windows produce stable recommendations for all four targets, whereas one-month windows can under-select the cumulative target's boundary persistence. Appendix Tables~\ref{tab:pilot_length_sensitivity}, \ref{tab:loadstate_calibration_sensitivity}, and \ref{tab:synthetic_transportability} report the numerical evidence. The exercise remains a semi-synthetic covariance stress test: it evaluates known injected response paths under realistic half-hourly load dynamics and does not estimate household-level elasticities or forecast contemporary demand-response magnitudes.

\section{Interpretation and external validity}\label{sec:limitations}
The results are benchmark and implementation tools for a pre-specified dynamic reporting object. The scope is narrower than the set of all dynamic-design problems. Three boundaries matter.

\subsection{Estimands, estimators, and assignment mechanisms}
The main target is an assigned-path, reduced-form LP response for a discrete-time treatment path. Structural price elasticities, welfare primitives, market-clearing counterfactuals, mechanism decompositions, distributional welfare, household-level heterogeneous effects, individual treatment effects, and continuous-horizon response curves require different estimands or additional assumptions. Similarly, GLS, VAR or transfer-function estimation, GMM, Bayesian or smooth LPs, Huber or quantile LPs, and regularized large-horizon LPs should use their own estimator-specific covariance or posterior risk rather than mechanically reusing the OLS/HAC benchmark. The finite-memory benchmark treats the lag coefficients $g_\ell$ as time-invariant over the experiment. If $g_\ell$ varies smoothly with season, temperature, or adoption trends, the path-adjusted LP targets a time-averaged response under the chosen weights; abrupt regime changes instead require stratified analysis or time-varying-parameter estimators before applying the calibrated selector. The HW-LP logic is modular in this respect: specify the reporting object, estimator, and loss weights first, then evaluate the corresponding design risk.

First-order Markov switchbacks are attractive because they are operationally transparent, easy to implement, and directly control expected run length and active-spell counts. They are not globally optimal over higher-order Markov, multisine, adaptive, finite-sequence, business-loss-aware, or continuous-time designs. Multi-arm or dose-response switchbacks also fall outside the binary closed form because their transition structure is higher-dimensional and the information matrix is generally block-Toeplitz. Appendix Figure~\ref{fig:multiarm-information} shows the distinction in a three-arm diagnostic: symmetric switching retains a separable scalar structure after contrast coding, while directional cycling produces matrix-valued lag blocks. The calibrated covariance selector remains applicable once the realized or model-implied block information matrix is computed. Richer finite-sequence algorithms can improve risk when the exact calendar grid and balance constraints are fixed, as Table~\ref{tab:finite_sequence_benchmark} illustrates. If the assignment follows a continuous-time switching process and is sampled at interval $\Delta$, the discrete persistence should be interpreted through $r\simeq \exp(-\lambda\Delta)$; changing the analysis frequency changes the persistence scale and requires a new design calculation. Irregularly sampled outcomes--event-driven trades, jittered impressions, or sparse session-level logs--require preprocessing to a regular analysis grid or a continuous-time HW-LP formulation; the discrete Toeplitz benchmark assumes fixed-interval observations. The reporting horizon $H$ is a pre-specified reporting resolution, not an estimate of the true biological or economic memory. If a deployment requires 12--48 hour rebound targets or smooth response curves, the analysis should jointly pre-specify the target horizon, regularization, and inference adjustment rather than treating the $H=8$ benchmark as universal. Cross-sectional interference, two-stage clustered switchbacks, state-dependent persistence, online learning, group-sequential stopping, and platform bandits are compatible with the broad idea of target-specific design, but they replace the stationary Toeplitz matrix with a realized, time-inhomogeneous, or policy-dependent information matrix.

\subsection{External validity and field constraints}
Field use also imposes constraints that are not statistical optimality conditions: participant burden, business loss, settlement windows, capacity obligations, maximum bill-risk protections, consent and withdrawal rules, privacy budgets, secure aggregation, audit logs, and regulator-facing disclosure. These constraints should enter the feasible design set or the loss function before the experiment begins. The Low Carbon London analysis is a semi-synthetic aggregate covariance stress test based on one public demand-response setting, not a transportability claim for other populations, climates, technologies, or household subgroups. For new deployments, the recommended procedure is to choose the reporting object and weight matrix, select a pilot-free or pilot-calibrated assignment rule, use randomization-first inference in high-persistence regimes, and document the realized assignment, exclusions, numerical diagnostics, and finite-sample balance reports.

\section{Conclusion}
Dynamic experiments should be designed for the dynamic objects researchers plan to report, and the field recommendation should be calibrated to the estimator and covariance environment actually used. We develop a horizon-weighted local projection criterion that maps a target response curve or contrast into an assignment design problem. In balanced Markov switchback designs, the criterion is analytically tractable: the information matrix is AR(1)-Toeplitz, its inverse is tridiagonal, and the optimal assignment persistence has a closed form. That formula is the benchmark, not the universal prescription.

The benchmark results show that there is no universally optimal assignment persistence. Within the benchmark Markov class, iid assignment is appropriate for isolated horizon effects. Persistent assignment can be efficient for smooth cumulative effects. Moderate or alternating assignment can be useful for rebound and other sign-changing contrasts. The omitted-carryover analysis further cautions against choosing extreme persistence solely because it lowers variance for a smooth target: persistence is the variance-allocation margin, while buffer horizons are the main analysis-design tool for protecting reported contrasts from omitted tails.

The results should be interpreted as design guidance for experiments that can choose the temporal pattern of assignment after pre-specifying a dynamic reporting object. The empirical results separate three design tasks: use the analytical benchmark when the main issue is persistence itself, use the sparse-budget comparison when the design is constrained by event frequency or fatigue, and use the historical-schedule diagnostic when the relevant question is whether an already deployed tariff path was informative for the target of interest. These three evidence blocks answer related but distinct design questions and should not be collapsed into a single ranking. The closed-form rules are benchmark formulas; in field settings with heterogeneous residual variance, calendar controls, or serial dependence, the same objective is implemented using pilot, HAC, residual-bootstrap, or calibrated-baseline covariance estimates. The semi-synthetic Low Carbon London exercise therefore evaluates designs under realistic load dynamics with known injected response paths, rather than estimating the causal effect of the observed tariff schedule. The LCL results calibrate design risk under realistic half-hourly covariance; they are not a forecast of contemporary demand-response magnitudes.

The benchmark comparisons and Low Carbon London calibrated evaluation illustrate that these differences are empirically meaningful. Mismatching the assignment design to the reported dynamic object can substantially inflate target MSE. In sparse demand-response settings, the active-share budget should be separated from the persistence choice: the budget determines information scale, while persistence determines how information is allocated across horizons. Extending the closed-form rule to multi-arm or dose-response switchbacks is a natural next step. The three-arm diagnostic in Appendix Figure~\ref{fig:multiarm-information} shows why the extension is nontrivial: the information matrix is block-Toeplitz and directional switching can invalidate a scalar persistence summary, although the calibrated selector still applies once that block matrix is formed. The implication is that the experimental design should be chosen after specifying the local-projection target, not before.

\appendix

\section{Design extensions and numerical benchmarks}\label{app:design_extensions}

\subsection{Weighting matrices and spectral characterization}
Corollary~\ref{cor:linear_transforms} is useful in practice. A researcher may first estimate a response curve and then report a cumulative sum, a delayed-window average, or a policy-path summary. The design criterion is obtained by mapping that reporting object back to an effective weight matrix on the response curve.

The weight matrix $W$ is a decision-theoretic input, not a free tuning parameter to be chosen after seeing estimates. A scalar report fixes $W=cc'$. A full-curve report with $W=I$ encodes the value judgment that each horizon is equally important on the chosen scale. If the policy decision weights early, peak, or late horizons differently, those weights should be declared before the design is selected, and the sensitivity should be reported. Table~\ref{tab:weight_sensitivity} gives a benchmark sensitivity calculation for $H=8$: cumulative reporting favors high persistence, whereas full-curve, early-horizon, and front-loaded reporting favor iid assignment in the homoskedastic Markov benchmark.

\begin{table}[H]
\centering
\caption{Sensitivity of the benchmark persistence rule to the reporting weights.}
\label{tab:weight_sensitivity}
\begin{tabular}{p{0.24\linewidth}p{0.20\linewidth}ccc}
\toprule
Reporting object & Weight or contrast & Selected $r^*$ & iid rel. risk & $r=0.90$ rel. risk\\
\midrule
Immediate scalar & $e_0$ & 0.000 & 1.000 & 5.263\\
Cumulative scalar & $\mathbf 1$ & 0.950 & 7.468 & 1.179\\
Alternative delayed window & avg $h=4,5,6$ & 0.382 & 1.342 & 4.307\\
Alternative rebound contrast & early $-$ late & 0.431 & 1.404 & 3.509\\
Full curve uniform & diag uniform & 0.000 & 1.000 & 8.579\\
Front-loaded curve & diag $e^{-h/2}$ & 0.000 & 1.000 & 7.799\\
Late-loaded curve & diag $e^{-(H-h)/2}$ & 0.000 & 1.000 & 7.799\\
\bottomrule
\end{tabular}
\begin{minipage}{0.94\linewidth}
\emph{Notes:} The table evaluates $\operatorname{tr}\{WQ_H(r)^{-1}\}$ or $c'Q_H(r)^{-1}c$ on a fine grid for $H=8$ and $r\in[-0.95,0.95]$. Relative risks are normalized by the row-specific optimum. The entries show why $W$ is a pre-analysis-plan choice: a cumulative scalar target favors high persistence, while full-curve or early-horizon reporting favors iid assignment in the homoskedastic benchmark. The alternative delayed and rebound rows are sensitivity contrasts and are distinct from the empirical response scenarios in Section~\ref{sec:lcl}.
\end{minipage}
\end{table}

For a general symmetric matrix $W\succeq0$, $\mathcal R_W(r)=\operatorname{tr}\{W\sigma^2Q_H(r)^{-1}\}$. Because $Q_H(r)^{-1}$ is tridiagonal, the risk depends only on the diagonal and adjacent off-diagonal elements of $W$:
In practice, this means that the design problem is driven by which adjacent horizons matter jointly in the reported object. Smooth reporting weights reward persistence; oscillating weights reward separation across nearby horizons. When the experiment also faces an active-share budget, the same persistence logic applies after conditioning on the budgeted share.
\[
    \operatorname{tr}\{WQ_H(r)^{-1}\}=\frac{4}{1-r^2}\left[a_W+b_Wr^2-2d_Wr\right],
\]
where $a_W=\operatorname{tr}(W)$, $b_W=\sum_{j=1}^{H-1}W_{jj}$, and $d_W=\sum_{j=1}^{H}W_{j-1,j}$. Thus the same closed-form characterization applies: isolated horizons and full-curve reporting favor iid assignment in the benchmark, smooth cumulative reporting favors persistence, and oscillating rebound contrasts favor moderate persistence or alternation depending on the sign of adjacent-horizon weights. This guidance is kept in prose to avoid duplicating the headline risk figure and the decision tables below.

The same target-dependence can be read in frequency-domain language. Let $C(z)=\sum_{j=0}^{H}c_jz^j$ be the lag polynomial associated with the reported contrast. The centered Markov assignment has spectral density
\[
    f_A(\omega;r)=\frac{1}{4}\frac{1-r^2}{|1-re^{-i\omega}|^2}.
\]
For the finite contrast vector, the following identity links the spectral expression to the finite Toeplitz inverse.
\begin{lemma}[Spectral endpoint identity]\label{lem:spectral_endpoint}
Let $C(z)=\sum_{j=0}^{H}c_jz^j$ and let
\[
    f_A(\omega;r)=\frac{1}{4}\frac{1-r^2}{|1-re^{-i\omega}|^2},\qquad |r|<1.
\]
Then
\[
    c'Q_H(r)^{-1}c
    =\frac{1}{2\pi}\int_{-\pi}^{\pi}\frac{|C(e^{-i\omega})|^2}{f_A(\omega;r)}\,d\omega
    -\frac{4r^2}{1-r^2}(c_0^2+c_H^2).
\]
Equivalently, the integral equals the exact finite-$H$ quadratic form plus the endpoint term $4r^2(c_0^2+c_H^2)/(1-r^2)$.
\end{lemma}

Thus the frequency-domain expression is not a competing approximation to the time-domain formula; it is the same quadratic form plus a transparent finite-endpoint correction. For diffuse targets whose endpoint weights do not dominate $\sum_jc_j^2$, the relative correction is $O(1/H)$ as the reported horizon grows. The design intuition is classical: smooth cumulative contrasts emphasize low frequencies and favor persistent assignment, while alternating or rebound-shaped contrasts have higher-frequency content and can favor alternation or moderate persistence. The endpoint correction also has a useful Toeplitz-spectral reading. As $H$ grows, the eigenvalue distribution of the AR(1)-Toeplitz information matrix is governed by the Szego spectral limit, while the finite-$H$ endpoints account for the residual boundary contribution in Lemma~\ref{lem:spectral_endpoint}. In that sense, the correction captures the finite-horizon edge contribution that disappears for diffuse contrasts but can matter for endpoint-heavy targets. Appendix Figure~\ref{fig:spectral-sampling} visualizes the corresponding finite-Toeplitz eigenvalues and sampling-frequency mapping.

The calibrated selector has the same spectral reading after replacing the homoskedastic Markov covariance by the estimator covariance used in the field design. A pilot HAC, residualized covariance, residual bootstrap, or realized-schedule calculation estimates how the score spectrum differs from the benchmark $f_A(\omega;r)$. The calibrated risk can then be read as weighting the target filter $|C(e^{-i\omega})|^2$ by this residualized spectrum rather than by the homoskedastic assignment spectrum alone. When residual load dynamics or calendar residualization reshape the low- and high-frequency parts of the score spectrum, the calibrated minimizer can move away from the closed-form $r_c^{\mathrm{int}}$. This is why the spectral identity is used as a benchmark interpretation, while the field recommendation is made from the covariance of the planned estimator.

The benchmark is a discrete-time design problem on a pre-specified analysis grid. If assignment follows a continuous-time two-state switching process with rate $\lambda$ and is sampled every $\Delta$ units of time, the discrete persistence is approximately $r=\exp(-\lambda\Delta)$ under the convention that $\lambda$ indexes the autocorrelation decay. Consequently, changing from half-hourly to five-minute or one-minute data without changing the underlying physical switching rate changes the numerical value of $r$. The HW-LP formula is therefore not sampling-frequency invariant by itself: the sampling grid, feasible minimum run length, and reporting horizons must be fixed before the Markov persistence rule is interpreted. Continuous-time switchback designs based on diffusion or controlled Markov-process likelihoods would replace the finite Toeplitz matrix by a continuous-time information operator and are a separate design problem \citep{AitSahalia2002}. The discrete rule here is the finite-grid benchmark used for the planned analysis frequency, not an Ito-calculus limit theorem.

\subsection{Implementation and numerical diagnostics}
The analytical rule is used as a benchmark, not as an automatic field prescription. The implementation uses a small set of pre-specified tolerances: switch from the closed-form benchmark to the calibrated selector when the calibrated relative-risk gain exceeds 20 percent; recompute the selector under one-half and twice the capped HAC bandwidth; start buffer-horizon checks at $H=K_0+2$ and stop when target variance rises by more than 10 percent or numerical conditioning exceeds the pre-specified cap; and record a tighter feasible-endpoint check when the selected target lies on a persistence boundary. These tolerances are design-protocol inputs rather than new optimality claims. With pilot covariance errors of order $T_p^{-1/2}$ and grid error controlled at the same order as in Corollary~\ref{cor:pilot_length}, relative risk gaps below roughly two pilot-standard-error units should be treated as non-decisive.

The appropriate design route depends on the available calibration information. Without a local pilot, the closed form is a conservative benchmark rather than a field guarantee. With historical baseline or pilot information, calibrated covariance selection is preferable. A separated pilot and confirmatory stage gives the cleanest inference but is most expensive. Rolling deployments require adaptive-inference tools rather than the simple two-stage coverage proposition.

The reporting horizon is a reporting choice, not a hidden true-memory estimate. When the target itself changes with the horizon--for example, an all-reported-horizons cumulative target--the benchmark persistence recommendation can move to the feasible boundary and the condition number rises. When the scientific target is a fixed short delayed or rebound window and the analyst only pads additional nuisance horizons, the recommendation is more stable but numerical conditioning still deteriorates as the reported nuisance horizon grows. The joint $(H,r)$ sensitivity grid provides the corresponding numerical diagnostic.

The pre-specified switch threshold is reported as a sensitivity parameter rather than as a universal constant. At thresholds of 2, 5, 10, and 20 percent, the delayed target always replaces the benchmark with the calibrated selector; the immediate and cumulative targets never do. The rebound target switches only at the 2 percent threshold because its benchmark loss is 2.5 percent above the calibrated candidate. Under the descriptive approximation $\mathrm{SE}\simeq T_p^{-1/2}$ and a two-pilot-standard-error separation rule, thresholds of 2, 5, 10, and 20 percent correspond to effective pilot-information scales of roughly 10,000, 1,600, 400, and 100. These calculations are diagnostics for pre-analysis planning, not a formal optimal-threshold result.

The implementation defaults are fixed before the confirmatory experiment. When the pilot sample is large enough, the same protocol can replace fixed choices of $H$, $K_0$, and $\lambda$ with pre-specified selectors: AIC or BIC for lag and buffer length, blocked or rolling time-series cross-validation for prediction-oriented horizon and ridge choices, and sensitivity grids for local-bias radii \citep{Akaike1974,Schwarz1978,HansenLunde2005,BergmeirHyndmanKoo2018}. Those selectors should be run on pilot data or on a separated training fold. If the confirmatory outcomes are used to choose $\widehat H$, $\widehat K_0$, $\widehat\lambda$, or the target menu, the fixed-design intervals in Section~\ref{sec:simulations} no longer have their nominal interpretation without sample splitting, simultaneous multiverse bands, or post-selection intervals. The HW-LP rule therefore treats model selection as part of the design protocol, not as an informal robustness step after the reported effect has been inspected.

The closed-form expressions should be evaluated as quadratic forms, not by forming an explicit inverse in finite precision. The numerical implementation therefore constructs $Q_H(r)$, solves linear systems using symmetric positive-definite routines when the matrix is well conditioned, and reports the relative residual $\|Qx-c\|_2/(\|Q\|_2\|x\|_2+\|c\|_2)$ for any solve used in a design diagnostic. The default pre-analysis rule is to restrict the feasible set further if the candidate design has $\kappa_2(Q_H(r))>10^{10}$ and to treat $\kappa_2(Q_H(r))>10^{12}$ as a hard numerical exclusion unless a regularized estimator has been pre-specified. These thresholds are looser than the operational run-length constraints used in the LCL analysis, and they make explicit that persistence close to one is a numerical as well as statistical boundary. Table~\ref{tab:numerical_stability} recomputes the relevant scales for $H=8$ and $T=17{,}520$; the condition numbers remain modest at the LCL cap, while the raw HAC bandwidth would be impractically long without the effective-sample-size cap.

\begin{table}[htbp]
\centering
\caption{Numerical conditioning and capped HAC-bandwidth diagnostics.}
\label{tab:numerical_stability}
\begin{tabular}{rrrrrrrr}
\toprule
$r$ & $\kappa_2(Q_8)$ & $\log_{10}\kappa$ & Raw $b_0$ & $T_{\mathrm{eff}}$ & Capped $b_T$ & $b_T/T_{\mathrm{eff}}$ & Cap binds? \\ 
\midrule
0.00 & 1.0 & 0.00 & 34 & 17520 & 34 & 0.002 & no \\ 
0.70 & 22.5 & 1.35 & 192 & 3092 & 192 & 0.062 & no \\ 
0.90 & 124.8 & 2.10 & 642 & 922 & 231 & 0.251 & yes \\ 
0.95 & 294.5 & 2.47 & 1317 & 449 & 112 & 0.249 & yes \\ 
0.98 & 814.9 & 2.91 & 3343 & 177 & 44 & 0.249 & yes \\ 
0.99 & 1686.5 & 3.23 & 6720 & 88 & 22 & 0.250 & yes \\ 
\bottomrule
\end{tabular}

\begin{minipage}{0.92\linewidth}
Notes: $Q_8(r)=\frac14(r^{|i-j|})_{i,j=0}^8$. The digits-lost column is the first-order double-precision diagnostic $\log_{10}\kappa_2(Q_8)$. Raw $b_0$ is the persistence-adjusted starting bandwidth; capped $b_T$ applies the nearest-integer version of the finite-sample cap $0.25T_{\mathrm{eff}}$. Cap binding is interpreted as a warning to use calibrated covariance, bootstrap, or a smaller feasible persistence interval.
\end{minipage}
\end{table}

The cap-binding column makes the inference trigger explicit. Around the high-persistence region used by smooth cumulative targets, the cap begins to bind. Once it binds, the capped bandwidth is no longer an estimator of an asymptotically optimal truncation lag; it is a warning that the assignment path contains too few effective episodes for HAC normal approximations to be the primary inferential route. This is why cap binding, high persistence, and low active-spell counts are treated as the same randomization-first regime.

\subsection{Alternative assignment mechanisms and finite-sequence designs}
The closed-form Markov rule is an analytically transparent benchmark rather than an exhaustive search over feasible assignment processes. This part compares it with richer stochastic inputs, fixed-block alternatives, finite-sequence search, boundary diagnostics, and calibrated selection. The scope distinction is operational: richer design classes may improve target risk, but they should be evaluated with the same pre-specified HW-LP loss and with their own conditioning and feasibility diagnostics.

The closed-form rule is a conditional optimum within an analytically restricted assignment class. It is not claimed to dominate higher-order Markov, periodic, multisine, adaptive, multi-arm, or fully optimized input designs. The role of the Markov benchmark is to provide an interpretable mapping from a dynamic reporting target to assignment persistence. More elaborate design classes can improve on this benchmark when additional structure is specified. The same HW-LP objective can be evaluated for such designs by replacing $Q_H(r)$ with the appropriate empirical or model-implied information matrix. Appendix Figure~\ref{fig:multiarm-information} gives a scoped three-arm example: uniform switching yields separable scalar lag blocks after orthonormal arm coding, while directional cycling creates cross-arm blocks and can make the scalar persistence rule materially suboptimal for sign-changing targets.

Table~\ref{tab:nonmarkov_comparison} gives a focused comparison with deterministic fixed-block candidates. It is not a search over all higher-order, multisine, or adaptive designs, but it quantifies the scope statement. Long fixed blocks improve on the application-capped first-order Markov benchmark for the smooth cumulative target, while they are substantially worse for immediate, delayed, and rebound contrasts. The cumulative improvement also comes with a high design-matrix condition number, so lower benchmark risk need not imply easier finite-sample inference. The comparison reinforces the central claim: Markov persistence is a transparent benchmark, and richer design classes should be evaluated against the same pre-specified HW-LP loss rather than presumed uniformly better. Section~\ref{sec:lcl} adds a lightweight generic-MSE selector comparison, motivated by adaptive A/B-testing objectives, to separate target-specific design from target-agnostic curve-MSE optimization.

\begin{table}[htbp]
\centering
\caption{First-order Markov benchmark versus fixed-block non-Markov candidates.}
\label{tab:nonmarkov_comparison}
\begin{tabular}{lrrrr}
\toprule
Target & Markov $r^*$ & Best block $L$ & Rel. risk & Cond. no. \\
\midrule
Immediate & 0.000 & 12 & 3.750 & 55.8 \\
Cumulative & 0.700 & 24 & 0.622 & 158.7 \\
Delayed & 0.382 & 12 & 2.683 & 55.8 \\
Rebound & 0.314 & 12 & 2.891 & 55.8 \\
\bottomrule
\end{tabular}

\begin{minipage}{0.88\linewidth}
Notes: The Markov column minimizes the benchmark risk over $r\in[-0.7,0.7]$. Fixed-block candidates alternate $L$ treated half-hours and $L$ control half-hours, with $L\in\{1,2,3,4,6,8,12,16,24\}$; the table reports the best block length for each target. Relative risk is fixed-block risk divided by the Markov optimum for the same target. The condition number is for the corresponding finite lag-information matrix; large values, such as the cumulative block-$24$ row, warn that a low scalar risk can coexist with fragile finite-sample inversion.
\end{minipage}
\end{table}

Finite-sequence optimization asks a different question from the stationary Markov benchmark: it conditions on the exact calendar length and searches over realized assignment paths. The comparison below therefore reports both target risk and finite-design diagnostics rather than treating a locally optimized sequence as a universal assignment law.

Table~\ref{tab:finite_sequence_benchmark} adds finite-sequence benchmarks at two sample lengths and compares the theoretical Markov risk with realized Markov paths, fixed blocks, and coordinate-exchange finite sequences. The exercise clarifies the trade-off within the Markov class: first-order Markov designs are transparent, verifiable, and easy to constrain by expected run length, while finite-sequence algorithms can improve scalar target risk when the exact horizon, calendar grid, and balance constraints are fixed before the experiment. The realized-path distribution also shows that a Markov rule is a distribution over finite design matrices, not a guarantee that every realized path matches the expected Toeplitz benchmark. Appendix Table~\ref{tab:target_finite_path_dispersion} extends this check to the four reporting targets and to the sign-changing alternating contrast at the planning and LCL sample lengths.

\begin{table}[htbp]
\centering
\caption{Finite-sequence benchmarks for the cumulative target.}
\label{tab:finite_sequence_benchmark}
\begin{tabularx}{\textwidth}{@{}>{\raggedright\arraybackslash}Xrrrrrr@{}}
\toprule
\multicolumn{7}{@{}l}{\textit{Panel A: $T=160$}} \\
Design & Risk & \shortstack{Rel. to\\Markov} & \shortstack{Cond.\\no.} & \shortstack{Markov\\p10} & \shortstack{Markov\\median} & \shortstack{Markov\\p90} \\
\midrule
iid expectation & 36.00 & 3.732 & 1.0 & -- & -- & -- \\
HW-LP Markov expectation ($r=0.70$) & 9.65 & 1.000 & 22.5 & -- & -- & -- \\
Realized Markov $r=0.70$ paths & -- & -- & -- & 0.963 & 1.186 & 1.523 \\
Best well-conditioned fixed block ($L=20$) & 6.33 & 0.657 & 138.9 & -- & -- & -- \\
Coordinate-exchange finite sequence & 6.44 & 0.667 & 85.8 & -- & -- & -- \\
Unconstrained coordinate exchange (20 starts, best) & 4.22 & 0.438 & 92.0 & -- & -- & -- \\
\addlinespace
\multicolumn{7}{@{}l}{\textit{Panel B: $T=500$}} \\
\midrule
iid expectation & 36.00 & 3.732 & 1.0 & -- & -- & -- \\
HW-LP Markov expectation ($r=0.70$) & 9.65 & 1.000 & 22.5 & -- & -- & -- \\
Realized Markov $r=0.70$ paths & -- & -- & -- & 1.066 & 1.198 & 1.359 \\
Best well-conditioned fixed block ($L=20$) & 6.56 & 0.680 & 128.4 & -- & -- & -- \\
Coordinate-exchange finite sequence & 6.61 & 0.685 & 99.9 & -- & -- & -- \\
Unconstrained coordinate exchange (20 starts, best) & 4.12 & 0.427 & 477.3 & -- & -- & -- \\
\bottomrule
\end{tabularx}

\begin{minipage}{0.94\linewidth}\emph{Notes:} The finite-sequence exercise uses centered binary sequences, $H=8$, the cumulative contrast, exact finite design matrices, balance-preserving coordinate exchanges, and a maximum run length of 20 unless noted otherwise. Realized Markov rows report finite-sequence risk relative to the theoretical Markov $r=0.70$ benchmark after filtering for treatment share between 0.45 and 0.55 and maximum run length 20. The unconstrained row removes the run-length cap and is included only as a global-search warning. The unconstrained coordinate-exchange entries use 20 random starts; the corresponding best/median/worst risks are 4.222/4.458/4.750 at $T=160$ and 4.117/4.286/4.690 at $T=500$.
\end{minipage}
\end{table}

The unconstrained coordinate-exchange rows provide a limiting finite-sequence comparison. Removing the run-length cap improves scalar cumulative-target risk at both $T=160$ and $T=500$, but the resulting paths create long episodes, higher condition numbers, and less transparent assignment rules than the Markov class. The multistart line in the notes is included because coordinate exchange is a local optimizer: a single local optimum should not be interpreted as a global finite-sequence bound. The Markov rule is therefore justified as a transparent, implementable, and verifiable design class, not as a global finite-sequence optimum.

Smooth cumulative targets also illustrate the difference between an interior statistical optimum and an operational cap. In the benchmark formula the uncapped cumulative optimum approaches the positive boundary; therefore, when the field feasible set is $[-0.7,0.7]$, the reported cumulative recommendation is $r_{\max}$ rather than an interior optimum. Table~\ref{tab:cumulative_corner} reports the resulting sensitivity to the chosen value of $r_{\max}$. The diagnostic is most important when the target hits a feasibility boundary: the design recommendation is then partly an operational statement about maximum allowable episode length.

\begin{table}[htbp]
\centering
\caption{Cumulative-target corner sensitivity to $r_{\max}$.}
\label{tab:cumulative_corner}
\begin{tabular}{rrrrr}
\toprule
$r_{\max}$ & Selected $r^\star$ & Rel. risk vs $r_{\max}=0.70$ & Expected active spells & $\kappa_2(Q_8)$ \\ 
\midrule
0.50 & 0.50 & 1.520 & 2190 & 7.7 \\ 
0.60 & 0.60 & 1.244 & 1752 & 12.7 \\ 
0.70 & 0.70 & 1.000 & 1314 & 22.5 \\ 
0.80 & 0.80 & 0.783 & 876 & 45.5 \\ 
0.90 & 0.90 & 0.589 & 438 & 124.8 \\ 
\bottomrule
\end{tabular}

\begin{minipage}{0.88\linewidth}
Notes: The cumulative contrast uses $H=8$ in the balanced homoskedastic Markov benchmark. Relative risk is normalized to the $r_{\max}=0.70$ recommendation. Expected active spells use $Tp(1-p)(1-r)$ with $T=17{,}520$ and $p=1/2$.
\end{minipage}
\end{table}

\begin{proposition}[Risk consistency of calibrated HW-LP selectors]\label{prop:calibrated_selector}
Let $R_W(r)=\operatorname{tr}\{WV(r)\}$ and let $\widehat R_W(r)=\operatorname{tr}\{W\widehat V(r)\}$ on a compact feasible set $\mathcal R$. Suppose
\[
    \sup_{r\in\mathcal R}|\widehat R_W(r)-R_W(r)|\to_p0.
\]
If $\widehat r_W$ is an $o_p(1)$ approximate minimizer of $\widehat R_W$ over $\mathcal R$, then
\[
    R_W(\widehat r_W)-\inf_{r\in\mathcal R}R_W(r)\to_p0.
\]
If the oracle minimizer $r_W^\star$ is unique and separated, then $\widehat r_W\to_p r_W^\star$.
\end{proposition}

\begin{corollary}[Pilot length and approximate oracle selection]\label{cor:pilot_length}
Suppose the feasible set is implemented as a finite design grid $\mathcal D_T$ with $|\mathcal D_T|$ bounded or growing polynomially in $T_p$, and suppose a pilot sample of length $T_p$ produces covariance estimates satisfying
\[
    \max_{r\in\mathcal D_T}\|\widehat V(r)-V(r)\|=O_p(T_p^{-1/2}).
\]
Then $\max_{r\in\mathcal D_T}|\widehat R_W(r)-R_W(r)|=O_p(T_p^{-1/2})$. Consequently, any plug-in minimizer has oracle regret $O_p(T_p^{-1/2})$ over the grid. If the continuous risk surface is Lipschitz with constant $L_R$ and the grid mesh is $\Delta_T=\sup_{r\in\mathcal R}\min_{u\in\mathcal D_T}|r-u|$, the continuous-interval regret bound becomes
\[
    O_p(T_p^{-1/2})+2L_R\Delta_T.
\]
Choosing $\Delta_T=O(T_p^{-1/2})$ keeps the grid approximation error at the pilot-estimation rate. Thus pilot uncertainty enters the design rule through the accuracy of the covariance surface; when $T_p$ is small relative to the planned experiment, the calibrated selector should be interpreted as a noisy pilot recommendation rather than as an oracle design.
\end{corollary}

Proposition~\ref{prop:calibrated_selector} is the formal bridge from the closed-form benchmark to field implementation. The researcher can estimate $\widehat V(r)$ using a pilot experiment, a residual bootstrap, a HAC covariance calculation, or a calibrated baseline model, and then minimize the same HW-LP objective. For cumulative and rebound contrasts, the full covariance of the estimated response curve matters; diagonal marginal standard deviations are only a diagnostic approximation.

\subsection{Target robustness and information criteria}
In some experiments the exact reporting object may not be known before the experiment, even though a small menu of plausible dynamic objects is known. For example, a demand-response experiment may be designed to learn an immediate reduction, a cumulative reduction, or a rebound contrast. The HW-LP criterion can be used in this setting without changing the estimating equations.

Let $\mathcal C=\{c^{(1)},\ldots,c^{(K)}\}$ be a finite menu of horizon-weight vectors. Define the normalized target-menu loss
\[
    L_{\mathcal C}(r)=\max_{c\in\mathcal C}\frac{\mathcal R_c(r)}{\inf_{u\in\mathcal R}\mathcal R_c(u)}.
\]
The menu-robust design chooses
\[
    r_{\mathcal C}^{\star}\in\arg\min_{r\in\mathcal R}L_{\mathcal C}(r).
\]
This criterion is not meant to replace pre-specification of a primary estimand. Rather, it is useful when a researcher wants one assignment design that is not badly mismatched to a small set of pre-specified dynamic summaries. The normalization is a competitive-ratio choice: each target is scaled by the best risk attainable for that target within the same feasible set. This avoids letting a high-variance target dominate solely because it is measured on a larger scale. A Bayesian average risk $E_\pi[\mathcal R_c(r)]$, an absolute-regret criterion, or a Pareto-frontier report may be preferable when the researcher has credible prior weights over targets; the max-relative criterion is used here as a conservative default for target menus without such prior weights.

\begin{proposition}[Finite-menu robust HW-LP design]\label{prop:menu_robust}
Suppose $\mathcal R$ is compact and $\mathcal C$ is finite. If $\mathcal R_c(r)$ is continuous in $r$ for each $c\in\mathcal C$ and $\inf_{u\in\mathcal R}\mathcal R_c(u)>0$, then a menu-robust design $r_{\mathcal C}^{\star}$ exists. If $\widehat{\mathcal R}_c(r)$ satisfies
\[
    \max_{c\in\mathcal C}\sup_{r\in\mathcal R}|\widehat{\mathcal R}_c(r)-\mathcal R_c(r)|\to_p0,
\]
and $\widehat r_{\mathcal C}$ is an $o_p(1)$ approximate minimizer of the plug-in menu loss, then
\[
    L_{\mathcal C}(\widehat r_{\mathcal C})-\inf_{r\in\mathcal R}L_{\mathcal C}(r)\to_p0.
\]
\end{proposition}

Table~\ref{tab:menu_robust} illustrates the idea for the application-feasible interval $[-0.7,0.7]$ and a menu containing immediate, cumulative, delayed-window, and rebound targets. The menu-robust design is not optimal for any single target, but it limits the worst relative loss across the menu. This is a useful diagnostic when an experiment will report several dynamic summaries or when the primary target is selected from a pre-analysis menu.

\begin{table}[t]
    \centering
    \caption{Finite-menu robust design over common dynamic targets.}
    \label{tab:menu_robust}
    \begin{tabular}{lrrrrrr}
\toprule
Design & $r$ & immediate & cumulative & delayed & rebound & max loss \\
\midrule
iid & 0.000 & 1.00 & 3.73 & 1.34 & 1.19 & 3.73 \\
alternating & -0.700 & 1.96 & 19.21 & 6.37 & 4.81 & 19.21 \\
moderate & 0.300 & 1.10 & 2.20 & 1.02 & 1.00 & 2.20 \\
persistent & 0.700 & 1.96 & 1.00 & 1.46 & 1.55 & 1.96 \\
menu-robust & 0.539 & 1.41 & 1.41 & 1.08 & 1.14 & 1.41 \\
\bottomrule
\end{tabular}

    \begin{minipage}{0.92\linewidth}
    Notes: Entries are relative benchmark risks, normalized by the target-specific optimum over $r\in[-0.7,0.7]$. The menu is immediate, cumulative, delayed-window, and rebound targets. The menu-robust design minimizes the maximum relative loss across the menu.
    \end{minipage}
\end{table}

Table~\ref{tab:compromise_cum_delayed} reports a pairwise formulation of the same idea for a single experiment that must serve both cumulative and delayed-window summaries. Optimizing for one target alone can be costly for the other. In the calibrated pairwise comparison, the compromise design $r=0.66$ raises cumulative risk by 9.4 percent and delayed risk by 19.4 percent, while avoiding the larger single-target mismatch costs. The table is a planning diagnostic: if both targets are truly primary and neither compromise cost is acceptable, the researcher should pre-specify separate experiments or a minimax design over the target menu.

\begin{table}[t]
    \centering
    \caption{Cost of a single design serving cumulative and delayed targets.}
    \label{tab:compromise_cum_delayed}
    \begin{tabularx}{\textwidth}{@{}>{\raggedright\arraybackslash}Xrrrrr@{}}
\toprule
Design & $r$ & \shortstack{Cumulative\\rel. risk} & \shortstack{Delayed\\rel. risk} & Max loss & Sum loss \\
\midrule
Cumulative-optimal design & 0.70 & 1.000 & 1.450 & 1.450 & 2.450 \\
Delayed-optimal design & 0.45 & 1.310 & 1.000 & 1.310 & 2.310 \\
Pairwise compromise & 0.66 & 1.094 & 1.194 & 1.194 & 2.288 \\
\bottomrule
\end{tabularx}

    \begin{minipage}{0.92\linewidth}
    Notes: Relative risks are normalized by each target's target-specific calibrated optimum in the same candidate design menu. The compromise row minimizes the sum of the two normalized risks among the evaluated candidates. The diagnostic quantifies the cost of reporting both targets from one assignment design.
    \end{minipage}
\end{table}

The scalar HW-LP loss is the right object only after the decision maker has fixed the reported statistical target and any non-statistical costs. Field designs often also trade off statistical precision, participant fatigue, business or revenue loss, operational simplicity, equity, and welfare. The HW-LP objective can be embedded in standard multi-objective design tools: weighted-sum scalarization, lexicographic rules, $\epsilon$-constraints, or goal programming \citep{CharnesCooper1961,Miettinen1999}. Figure~\ref{fig:pareto-frontier} illustrates the idea for the cumulative target. The horizontal axis is normalized statistical risk, the vertical axis is an illustrative persistent-exposure loss proxy based on expected episode length, and the point labels show Markov persistence. The figure is not a welfare model; it is a diagnostic that makes clear that a variance-optimal persistence is one point on a frontier rather than a universal field recommendation.

The illustrative Pareto frontier is displayed in Appendix Figure~\ref{fig:pareto-frontier} to keep the main design section focused on the scalar HW-LP criterion.

The scalar HW-LP criterion is a $C$-optimality criterion for the pre-specified LP contrast $c'g$. Classical $D$-, $A$-, and $E$-optimality answer different questions: they optimize the volume, average marginal variance, or worst principal-axis variance of the full response-vector confidence ellipsoid. In the balanced Markov benchmark, those full-vector criteria favor iid assignment over the feasible interval used in the application, while target-specific $C$-optimality can favor persistent or moderate persistence. Appendix Table~\ref{tab:classical_optimality} gives the deterministic comparison. The point is not that $D$-, $A$-, or $E$-optimality is wrong; it is that they target full-vector information rather than the scalar dynamic object that the pre-analysis plan says will be reported. A Bayesian or KL-information design would similarly be appropriate when the scientific object is expected information gain or posterior contraction rather than frequentist precision for a named LP contrast \citep{Lindley1956,Pukelsheim2006,Amari2016}.

A Bayesian counterpart minimizes prior-averaged posterior risk rather than the frequentist variance of a named contrast. With a diffuse Gaussian prior and squared-error loss for the same LP contrast, the leading quadratic term reduces to the HW-LP criterion. With informative smoothness, sign-restriction, monotonicity, or shrinkage priors, the prior-averaged loss effectively replaces $W$ by a prior-weighted loss matrix and can move the Bayes-optimal persistence toward designs that emphasize the favored response shapes. Thus the closed-form rule is the frequentist benchmark; Bayesian designers should compute the same design menu using their posterior or prior-predictive risk.

\section{Inference details}\label{app:inference_details}

\subsection{Simultaneous confidence bands and full-curve inference}
Pointwise intervals are the starting point, but dynamic reporting often asks for a curve or a target menu. These objects require joint critical values rather than a collection of marginal horizon-by-horizon intervals.

For a response curve with $q$ reported horizons, let $D=\operatorname{diag}(\widehat V_\pi)$ and let $\widehat R=D^{-1/2}\widehat V_\pi D^{-1/2}$ be the estimated correlation matrix. A Bonferroni band uses
\[
    \hat g_j \pm z_{1-\alpha/(2q)}\sqrt{\widehat V_{\pi,jj}/T},\qquad j=0,\ldots,q-1.
\]
A maximum-$t$ band instead draws $Z\sim N(0,\widehat R)$, takes $k_{1-\alpha}$ as the $(1-\alpha)$ quantile of $\max_j|Z_j|$, and reports
\[
    \hat g_j \pm k_{1-\alpha}\sqrt{\widehat V_{\pi,jj}/T}.
\]
The same construction applies after premultiplying by any reporting matrix $L$ for cumulative windows or target menus.

For full-curve reporting, marginal horizon-by-horizon intervals should not be presented as a simultaneous band. A pre-analysis plan that reports $W=I$ should also specify a sup-$t$ band. Let $D=\operatorname{diag}(\widehat V)$ and define the studentized process
\[
    T_h=\frac{\widehat g_h-g_h}{\sqrt{\widehat V_{hh}/T}},\qquad h=0,\ldots,H.
\]
The simultaneous $1-\alpha$ band uses a critical value $k_{1-\alpha}$ satisfying
\[
    \Pr\left(\max_{0\le h\le H}|T_h|\le k_{1-\alpha}\mid \widehat V\right)\approx 1-\alpha,
\]
computed either by Gaussian simulation from $N(0,D^{-1/2}\widehat V D^{-1/2})$ or by design-resampling the realized switchback assignment. The band is
\[
    \widehat g_h \pm k_{1-\alpha}\sqrt{\widehat V_{hh}/T},\qquad h=0,\ldots,H.
\]
This is the local-projection analogue of simultaneous IRF bands such as those advocated by \citet{MontielOleaPlagborgMoller2021}; Bonferroni is a conservative fallback when the joint covariance is unstable. A deterministic benchmark diagnostic shows why this matters: for $H=8$, marginal 95 percent bands have only about 63--68 percent simultaneous coverage over the full curve under common Markov covariance benchmarks, while simulated sup-$t$ critical values around 2.74--2.77 restore 95 percent simultaneous coverage. Appendix Table~\ref{tab:supt_finite_sample} adds a finite-sample switchback diagnostic: at $T=2000$, pointwise coverage remains near 95 percent, marginal bands cover the full curve only about 62--69 percent of the time, and Gaussian or residual-multiplier sup-$t$ bands restore simultaneous coverage to about 94--95 percent. If a paper reports several scalar targets such as immediate, cumulative, delayed, and rebound effects and treats any rejection as confirmatory, the pre-analysis plan should either designate one primary target or apply a target-level family-wise adjustment, such as Bonferroni or a joint max-$t$ calculation over the target contrasts.

\subsection{Power, sequential testing, and design selection}\label{subsec:power_mde}
MSE and power are the same design comparison expressed on different scales. For a two-sided level-$\alpha$ test with desired power $1-\beta$, the approximate minimum detectable effect for $\theta_c=c'g$ is
\[
    \operatorname{MDE}_{1-\beta}(r)=\{z_{1-\alpha/2}+z_{1-\beta}\}\sqrt{c'V(r)c/T}.
\]
Thus minimizing $c'V(r)c$ also minimizes the squared MDE for the same target, but reporting the MDE is more useful for grant, IRB, and regulator-facing design documents. Table~\ref{tab:power_mde} and Figure~\ref{fig:power-mde} translate the benchmark risk into 80 percent power MDEs. The table reports MDEs in residual-standard-deviation units and, for scale only, the percent of mean load when the pilot residual standard deviation is 20 percent of mean load. Users should replace that ratio by their own pilot estimate.

Randomization-first inference solves the validity problem in high-persistence regimes, but it can be less powerful when the design generates very few active spells. Table~\ref{tab:fisher_power_cost} reports a finite-spell power diagnostic for a cumulative effect of 0.15 in a $T=2000$ planning exercise. This diagnostic addresses a different question from the ideal-coverage table: it asks what is lost in power when the exact randomization route replaces the normal approximation. At $r=0.99$, the spell count is only five and the Fisher power falls to 65 percent in the diagnostic. Thus high persistence should be justified by both target-risk gains and a power/MDE calculation, not by variance minimization alone.

\begin{table}[t]
\centering
\caption{Power cost of randomization-first inference in high-persistence regimes.}
\label{tab:fisher_power_cost}
\begin{tabular}{rrrrr}
\toprule
$r$ & Active spells & Normal power & Fisher power & MDE inflation \\ 
\midrule
0.00 & 500 & 15.0\% & 12.5\% & 3.46 \\ 
0.70 & 150 & 52.5\% & 45.8\% & 1.51 \\ 
0.95 & 25 & 84.2\% & 78.3\% & 1.02 \\ 
0.99 & 5 & 88.3\% & 65.0\% & 1.19 \\ 
\bottomrule
\end{tabular}

\begin{minipage}{0.88\linewidth}
Notes: The diagnostic uses the finite-spell planning calculation: $T=2000$, a cumulative effect of 0.15, and a two-sided 5 percent test. MDE inflation is the approximate factor needed to reach 80 percent power under the Fisher-power curve, using the same normal-score translation as Table~\ref{tab:power_mde}. The diagnostic provides a planning warning rather than a replacement for design-specific randomization power simulation.
\end{minipage}
\end{table}

\begin{table}[t]
    \centering
    \caption{Power translation: 80 percent MDEs for common LP targets.}
    \label{tab:power_mde}
    \begin{tabular}{llrrrr}
\toprule
Target & Design & $r$ & MDE/$\sigma$ & MDE if $\sigma/\mu=.20$ & rel. MDE \\
\midrule
Immediate & HW-LP & 0.000 & 0.042 & 0.85\% & 1.00 \\
Immediate & iid & 0.000 & 0.042 & 0.85\% & 1.00 \\
Immediate & persistent & 0.700 & 0.059 & 1.19\% & 1.40 \\
Cumulative & HW-LP & 0.700 & 0.066 & 1.31\% & 1.00 \\
Cumulative & iid & 0.000 & 0.127 & 2.54\% & 1.93 \\
Cumulative & persistent & 0.700 & 0.066 & 1.31\% & 1.00 \\
Delayed & HW-LP & 0.382 & 0.063 & 1.27\% & 1.00 \\
Delayed & iid & 0.000 & 0.073 & 1.47\% & 1.16 \\
Delayed & persistent & 0.700 & 0.077 & 1.53\% & 1.21 \\
Rebound & HW-LP & 0.314 & 0.078 & 1.55\% & 1.00 \\
Rebound & iid & 0.000 & 0.085 & 1.69\% & 1.09 \\
Rebound & persistent & 0.700 & 0.097 & 1.94\% & 1.25 \\
\bottomrule
\end{tabular}
    \begin{minipage}{0.92\linewidth}
    Notes: MDEs are computed from the homoskedastic Markov benchmark over $r\in[-0.7,0.7]$ with $H=8$, $T=17{,}520$, $\alpha=0.05$, and power 80 percent. The percent column is illustrative and equals MDE/$\sigma$ multiplied by $100\times0.20$; it should be rescaled by the pilot residual-SD-to-mean-load ratio in a field application.
    \end{minipage}
\end{table}

\begin{figure}[t]
    \centering
    \includegraphics[width=0.78\textwidth]{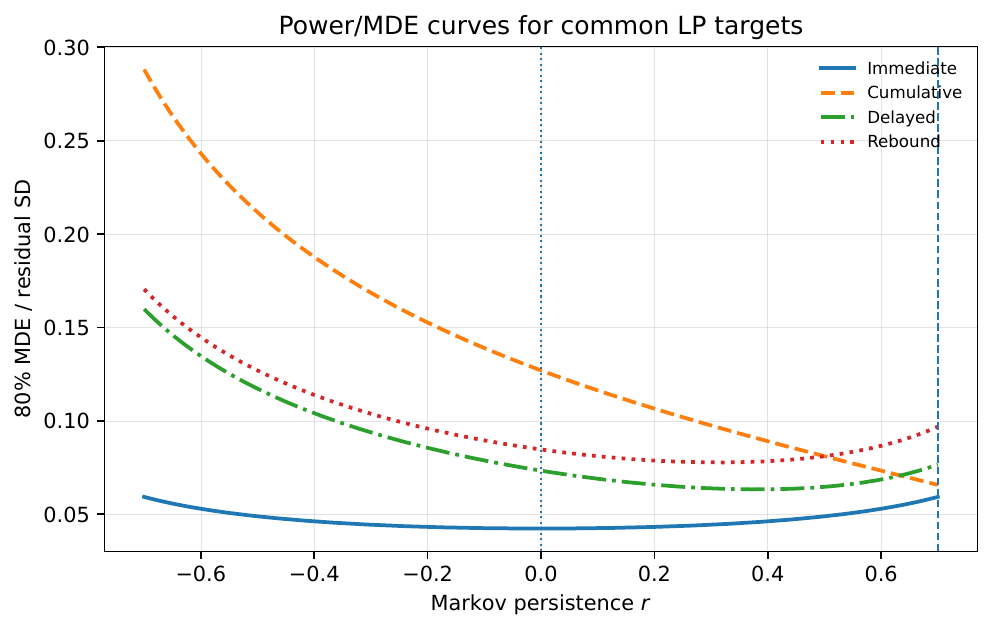}
    \caption{Power curves expressed as 80 percent MDEs in residual-standard-deviation units. The MDE scale is monotone in the square root of the HW-LP risk, so the variance-optimal design is also MDE-optimal for a fixed target and test.}
    \label{fig:power-mde}
\end{figure}

Sequential monitoring changes the inferential layer but not the role of the HW-LP information calculation. If interim analyses are planned, the fixed-$T$ critical value in the MDE formula should be replaced by the boundary implied by a pre-specified group-sequential, alpha-spending, or anytime-valid procedure \citep{Wald1945,Pocock1977,OBrienFleming1979,HowardEtAl2021}. At interim time $t_j$, the HW-LP variance calculation supplies the information fraction, approximately $I_j/I_J \simeq t_j\{c'V(r)c\}^{-1}/T\{c'V(r)c\}^{-1}=t_j/T$ for a fixed design and stable covariance, while calibrated implementations can use the realized covariance at each look. For a pre-specified effect size $\Delta_0$, the fixed-design sample-size analogue of the MDE formula is
\[
    T^*(r)=\frac{(z_{1-\alpha/2}+z_{1-\beta})^2 c'V(r)c}{\Delta_0^2}.
\]
Early stopping for efficacy, futility, or safety should therefore be specified as a monitoring rule layered on the chosen design, with the alpha-spending or e-process boundary replacing the one-shot normal critical value. Optimizing the interim schedule itself is a distinct sequential-design problem.

The recommended procedure is to evaluate candidate designs on a training or pilot sample, fix the assignment rule, and run the confirmatory experiment under that rule. The calibrated selector is data-dependent, so it must be part of the pre-specified design procedure rather than an informal after-the-fact design choice. A pre-analysis plan should record the map $\phi:\mathcal P\to\mathcal D$ from pilot information to the assignment design, including the pilot window, grid, covariance estimator, thresholds, and fallback rules. Conditional on that logged procedure and a separate confirmatory experiment, ordinary fixed-design coverage remains the relevant target.

Let $\mathcal D$ be a finite set of candidate designs and let a pilot sample $\mathcal P$ select
\[
    \widehat d=\phi(\mathcal P)\in\mathcal D,
\]
where $\phi$ includes the target weights, feasible set, bandwidth rule, buffer horizon, ridge constant, and closed-form-versus-calibrated switching threshold. If the confirmatory experiment is then run under $\widehat d$ and the pilot outcomes are not reused to estimate $g$, ordinary design-specific inference remains valid after selection.

\begin{proposition}[Two-stage design selection and coverage]\label{prop:post_selection_coverage}
Suppose $\mathcal D$ is finite and the pilot-selected design $\widehat d=\phi(\mathcal P)$ is chosen before the confirmatory experiment begins. For each fixed $d\in\mathcal D$, let $C_d$ be a confidence set for $\theta_c=c'g$ computed from the confirmatory sample under design $d$. If
\[
    \inf_{d\in\mathcal D}\Pr\{\theta_c\in C_d\mid \widehat d=d\}=1-\alpha+o(1),
\]
with the probability taken over the confirmatory assignment and outcome process, then
\[
    \Pr\{\theta_c\in C_{\widehat d}\}=1-\alpha+o(1).
\]
\end{proposition}

The proof is the law of iterated expectations: condition on the pilot event $\{\widehat d=d\}$ and use the uniform fixed-design coverage over the finite candidate set. Thus pilot uncertainty affects regret and power, as in Corollary~\ref{cor:pilot_length}, but it need not distort size when the pilot and confirmatory samples are separated and the selection rule is logged. This is the recommended protocol for two-stage experiments.

The result does not cover using the same realized experimental outcomes both to choose among multiple designs, horizons, bandwidths, ridge constants, or local-bias radii and to report an unadjusted confidence interval. If $H$, $K_0$, $M_0$, $\lambda$, $b_T$, or the target menu is tuned after inspecting the confirmatory estimates, the analysis should either use sample splitting, report simultaneous bands over the pre-specified multiverse, or use selection-adjusted inference in the spirit of post-selection simultaneous intervals and conditional selection intervals \citep{BerkEtAl2013,LeeEtAl2016}. A pre-analysis plan should therefore record: the target or menu $W$, the candidate design set $\mathcal D$, the pilot window, the feasible persistence interval, the HAC bandwidth rule and sensitivity grid, the buffer-horizon stopping rule, the ridge grid, and whether the final report is design-invariant path-adjusted LP or a design-specific policy-path LP.

\setcounter{table}{0}
\setcounter{figure}{0}
\setcounter{equation}{0}
\renewcommand{\thetable}{A.\arabic{table}}
\renewcommand{\thefigure}{A.\arabic{figure}}
\makeatletter
\renewcommand{\theHtable}{A.\arabic{table}}
\renewcommand{\theHfigure}{A.\arabic{figure}}
\renewcommand{\theHequation}{A.\arabic{equation}}
\makeatother
\renewcommand{\theequation}{A.\arabic{equation}}
\section{Additional Monte Carlo evidence}\label{app:additional_diagnostics}
\subsection{Finite-path, calibration, and inference diagnostics}

\begin{table}[htbp]
\centering
\caption{Target-specific finite randomization-path risk dispersion.}
\label{tab:target_finite_path_dispersion}
\begin{tabular}{llrrrrr}
\toprule
Target & $T$ & $r^\star$ & p10 & median & p90 & p90 cond. no. \\
\midrule
Immediate & 2,000 & 0.000 & 1.002 & 1.004 & 1.007 & 1.3 \\
Immediate & 17,520 & 0.000 & 1.000 & 1.000 & 1.001 & 1.1 \\
Cumulative & 2,000 & 0.700 & 0.937 & 1.003 & 1.075 & 26.4 \\
Cumulative & 17,520 & 0.700 & 0.978 & 0.999 & 1.026 & 23.7 \\
Delayed & 2,000 & 0.382 & 0.989 & 1.003 & 1.017 & 5.2 \\
Delayed & 17,520 & 0.382 & 0.996 & 1.001 & 1.006 & 4.8 \\
Rebound & 2,000 & 0.314 & 0.964 & 1.001 & 1.046 & 4.0 \\
Rebound & 17,520 & 0.314 & 0.987 & 1.000 & 1.014 & 3.6 \\
Alternating & 2,000 & -0.700 & 0.935 & 1.001 & 1.074 & 26.4 \\
Alternating & 17,520 & -0.700 & 0.977 & 1.000 & 1.024 & 23.6 \\
\bottomrule
\end{tabular}

\begin{minipage}{0.92\linewidth}
Notes: The diagnostic simulates finite stationary Markov assignment paths and recomputes the exact lagged-assignment Gram matrix for each path. Reported values are quantiles of finite-path target risk divided by the corresponding population Markov risk for the same target and persistence. The calculation uses $H=8$ and 500 simulated paths per target and sample length. The Alternating row uses the sign-changing weights $c_h=(-1)^h$ and the feasible endpoint $r^\star=-0.70$.
\end{minipage}
\end{table}

\begin{table}[htbp]
    \centering
    \caption{Target-specific $C$-optimality versus full-vector information criteria.}
    \label{tab:classical_optimality}
    \begin{tabular}{lrrrrrr}
\toprule
Target & C-opt $r$ & D-opt $r$ & A-opt $r$ & E-opt $r$ & iid/C & D/C \\
\midrule
Immediate & 0.000 & 0.000 & 0.000 & 0.000 & 1.00 & 1.00 \\
Cumulative & 0.700 & 0.000 & 0.000 & 0.000 & 3.73 & 3.73 \\
Delayed & 0.382 & 0.000 & 0.000 & 0.000 & 1.34 & 1.34 \\
Rebound & 0.314 & 0.000 & 0.000 & 0.000 & 1.19 & 1.19 \\
\bottomrule
\end{tabular}
    \begin{minipage}{0.92\linewidth}
    Notes: The feasible interval is $[-0.7,0.7]$ and $H=8$. $C$-optimality minimizes $c'Q_H(r)^{-1}c$ for the row target. $D$-, $A$-, and $E$-optimality minimize determinant, trace, and maximum eigenvalue criteria for the full response vector. The last columns report scalar $C$-risk relative to the target-specific optimum.
    \end{minipage}
\end{table}

\begin{table}[!htbp]
\centering
\caption{Residualized aggregate calibration stability across pilot windows.}
\label{tab:residualized_split_stability}
\begin{tabular}{lcccc}
\toprule
Pilot window & Target & Pilot $r^\star$ & Evaluation oracle $r^\star$ & Evaluation relative risk \\
\midrule
Jan--Jun pilot & Immediate & 0.00 & 0.00 & 1.000 \\
Jan--Jun pilot & Cumulative & 0.90 & 0.90 & 1.000 \\
Jan--Jun pilot & Delayed & 0.79 & 0.77 & 1.003 \\
Jan--Jun pilot & Rebound & 0.27 & 0.29 & 1.001 \\
Jul--Dec pilot & Immediate & 0.00 & 0.00 & 1.000 \\
Jul--Dec pilot & Cumulative & 0.90 & 0.90 & 1.000 \\
Jul--Dec pilot & Delayed & 0.77 & 0.79 & 1.008 \\
Jul--Dec pilot & Rebound & 0.29 & 0.27 & 1.001 \\
\bottomrule
\end{tabular}
\begin{minipage}{0.92\linewidth}\emph{Notes:} The aggregate baseline load is residualized on hour-of-day, day-of-week, and month indicators. The persistence grid is $[-0.90,0.90]$ in increments of 0.01. The table estimates residual autocovariances on one half of 2013, selects the calibrated persistence on that pilot half, and evaluates the resulting target risk on the other half. Relative risk is normalized by the evaluation-half oracle over the same first-order Markov grid.
\end{minipage}
\end{table}

\begin{table}[!htbp]
\centering
\caption{Aggregate load-state calibration sensitivity.}
\label{tab:loadstate_calibration_sensitivity}
\begin{tabularx}{\textwidth}{@{}>{\raggedright\arraybackslash}Xlccc@{}}
\toprule
Calibration sample & Target & Fine-grid $r^\star$ & Gain vs iid & \shortstack{Relative risk at\\pooled $r^\star$} \\
\midrule
Pooled & Immediate & 0.00 & 1.00 & 1.000 \\
Pooled & Cumulative & 0.90 & 4.47 & 1.000 \\
Pooled & Delayed & 0.79 & 4.22 & 1.000 \\
Pooled & Rebound & 0.28 & 1.16 & 1.000 \\
Low-load periods & Immediate & 0.00 & 1.00 & 1.000 \\
Low-load periods & Cumulative & 0.90 & 4.14 & 1.000 \\
Low-load periods & Delayed & 0.81 & 4.71 & 1.004 \\
Low-load periods & Rebound & 0.25 & 1.14 & 1.002 \\
Middle-load periods & Immediate & 0.00 & 1.00 & 1.000 \\
Middle-load periods & Cumulative & 0.90 & 4.87 & 1.000 \\
Middle-load periods & Delayed & 0.77 & 3.99 & 1.003 \\
Middle-load periods & Rebound & 0.28 & 1.17 & 1.000 \\
High-load periods & Immediate & 0.00 & 1.00 & 1.000 \\
High-load periods & Cumulative & 0.90 & 3.66 & 1.000 \\
High-load periods & Delayed & 0.78 & 4.12 & 1.001 \\
High-load periods & Rebound & 0.27 & 1.16 & 1.000 \\
\bottomrule
\end{tabularx}%

\begin{minipage}{0.92\linewidth}\emph{Notes:} This diagnostic uses only the aggregate baseline series described in Section~\ref{sec:lcl}. Low-, middle-, and high-load periods are defined by terciles of the aggregate baseline load, after residualization on hour, weekday, and month indicators. The exercise tests whether calibrated residual autocovariances change the persistence recommendation across load states; it is not a household-level heterogeneous-treatment-effect analysis.
\end{minipage}
\end{table}

\begin{table}[htbp]
\centering
\caption{Synthetic transportability diagnostic across residual dynamics.}
\label{tab:synthetic_transportability}
\begin{tabularx}{\textwidth}{@{}>{\raggedright\arraybackslash}Xccc>{\raggedright\arraybackslash}X@{}}
\toprule
Residual DGP & Mean ratio & Max ratio & Changed targets & Calibrated $r^*$: I/C/D/R \\
\midrule
iid & 1.00 & 1.00 & 0/4 & 0.00/0.90/0.38/0.31 \\
AR(1), $\rho=0.5$ & 1.02 & 1.09 & 1/4 & 0.00/0.90/0.55/0.33 \\
AR(1), $\rho=0.8$ & 1.11 & 1.43 & 1/4 & 0.00/0.90/0.69/0.31 \\
AR(2), $(0.5,0.2)$ & 1.06 & 1.26 & 1/4 & 0.00/0.90/0.64/0.32 \\
seasonal 48-lag proxy & 1.07 & 1.29 & 1/4 & 0.00/0.90/0.65/0.32 \\
volatility-cluster proxy & 1.05 & 1.18 & 1/4 & 0.00/0.90/0.60/0.33 \\
Student-$t$(df=5) heavy-tail proxy & 1.04 & 1.15 & 1/4 & 0.00/0.90/0.58/0.33 \\
two-regime switching variance proxy & 1.08 & 1.34 & 1/4 & 0.00/0.90/0.67/0.32 \\
compound heavy-tail/regime proxy & 1.07 & 1.27 & 1/4 & 0.00/0.90/0.64/0.32 \\
\bottomrule
\end{tabularx}

\begin{minipage}{0.92\linewidth}
Notes: Risk ratios compare using the closed-form benchmark design with using the calibrated design under the listed residual autocovariance. Values are averaged over immediate, cumulative, delayed, and rebound contrasts. I/C/D/R reports calibrated persistence for those four targets on the fine grid $[-0.90,0.90]$. The Student-$t$(df=5), two-regime, and compound rows are synthetic robustness probes for tail, volatility-state, and joint tail-regime sensitivity; they are not a second field-data validation.
\end{minipage}
\end{table}

\begin{table}[htbp]
\centering
\caption{Pilot-length sensitivity for calibrated aggregate covariance selection.}
\label{tab:pilot_length_sensitivity}
\begin{tabular}{llcccc}
\toprule
Pilot length & Target & Windows & $r^\star$ range & Median $r^\star$ & Max eval. rel. risk \\
\midrule
1 month & Immediate & 12 & 0.00 & 0.00 & 1.000 \\
1 month & Cumulative & 12 & 0.84--0.90 & 0.90 & 1.175 \\
1 month & Delayed & 12 & 0.75--0.82 & 0.78 & 1.022 \\
1 month & Rebound & 12 & 0.24--0.30 & 0.28 & 1.004 \\
3 months & Immediate & 4 & 0.00 & 0.00 & 1.000 \\
3 months & Cumulative & 4 & 0.90 & 0.90 & 1.000 \\
3 months & Delayed & 4 & 0.77--0.80 & 0.79 & 1.006 \\
3 months & Rebound & 4 & 0.27--0.29 & 0.28 & 1.001 \\
6 months & Immediate & 2 & 0.00 & 0.00 & 1.000 \\
6 months & Cumulative & 2 & 0.90 & 0.90 & 1.000 \\
6 months & Delayed & 2 & 0.77--0.79 & 0.78 & 1.008 \\
6 months & Rebound & 2 & 0.27--0.29 & 0.28 & 1.001 \\
12 months & Immediate & 1 & 0.00 & 0.00 & 1.000 \\
12 months & Cumulative & 1 & 0.90 & 0.90 & 1.000 \\
12 months & Delayed & 1 & 0.79 & 0.79 & 1.000 \\
12 months & Rebound & 1 & 0.28 & 0.28 & 1.000 \\
\bottomrule
\end{tabular}

\begin{minipage}{0.92\linewidth}
Notes: The diagnostic residualizes the aggregate baseline series on hour-of-day, day-of-week, and month indicators, selects persistence on each pilot window, and evaluates the resulting relative risk on the complementary window; the full-year row is evaluated on the full sample. The one-month rows summarize 12 monthly pilots, the three-month rows summarize four calendar quarters, and the six-month rows summarize the two half-year splits. 
\end{minipage}
\end{table}

\begin{table}[htbp]
\centering
\caption{Finite-sample simultaneous-band coverage diagnostic.}
\label{tab:supt_finite_sample}
\begin{tabularx}{\textwidth}{@{}>{\raggedright\arraybackslash}Xrrrrrrr@{}}
\toprule
Method & $r$ & $T$ & Reps & \shortstack{Pointwise\\coverage} & \shortstack{Marginal\\simultaneous} & \shortstack{Sup-$t$\\critical} & \shortstack{Sup-$t$\\simultaneous} \\
\midrule
Gaussian covariance & 0.00 & 2000 & 2000 & 95.0\% & 62.3\% & 2.76 & 94.7\% \\
Gaussian covariance & 0.70 & 2000 & 2000 & 94.8\% & 66.0\% & 2.74 & 95.3\% \\
Gaussian covariance & 0.90 & 2000 & 2000 & 94.9\% & 66.3\% & 2.74 & 95.2\% \\
residual multiplier bootstrap & 0.00 & 2000 & 1000 & 94.7\% & 62.1\% & 2.74 & 93.9\% \\
residual multiplier bootstrap & 0.70 & 2000 & 1000 & 95.0\% & 66.7\% & 2.72 & 95.1\% \\
residual multiplier bootstrap & 0.90 & 2000 & 1000 & 94.9\% & 68.6\% & 2.70 & 94.5\% \\
\bottomrule
\end{tabularx}

\begin{minipage}{0.94\linewidth}Notes: The Gaussian-covariance rows use 2,000 Gaussian switchback simulations with $T=2{,}000$, $H=8$, true response zero, and the balanced Markov assignment law. The residual-multiplier rows add 1,000 outer simulations and 149 multiplier draws per replication, using the realized design matrix and OLS residuals to compute a conditional sup-$t$ critical value. Pointwise coverage averages horizon-by-horizon marginal coverage; the two simultaneous columns report whether all nine horizons are covered in a replication. The exercise checks the finite-sample cost of presenting marginal LP bands as curve-level bands.
\end{minipage}
\end{table}

\begin{table}[htbp]
\centering
\caption{Sparse-budget cumulative-target finite randomization-path risk dispersion.}
\label{tab:sparse_finite_path}
\begin{tabular}{rrrrrrr}
\toprule
$p_0$ & $T$ & $r_{\max}$ & median & p90 & p90 cond. no. & p10 active spells \\
\midrule
0.05 & 2{,}000 & 0.825 & 1.034 & 1.817 & 86.9 & 12 \\
0.05 & 17{,}520 & 0.825 & 1.007 & 1.197 & 64.6 & 131 \\
\addlinespace
0.10 & 2{,}000 & 0.815 & 1.022 & 1.417 & 70.8 & 27 \\
0.10 & 17{,}520 & 0.815 & 1.003 & 1.111 & 57.2 & 272 \\
\addlinespace
0.14 & 2{,}000 & 0.806 & 1.009 & 1.303 & 61.6 & 40 \\
0.14 & 17{,}520 & 0.806 & 1.004 & 1.091 & 52.2 & 386 \\
\addlinespace
0.25 & 2{,}000 & 0.778 & 1.006 & 1.155 & 46.4 & 75 \\
0.25 & 17{,}520 & 0.778 & 0.999 & 1.047 & 41.0 & 704 \\
\bottomrule
\end{tabular}

\begin{minipage}{0.94\linewidth}Notes: Each row uses 2,000 stationary unbalanced Markov paths, $H=8$, the cumulative target $c=\mathbf 1$, and the largest persistence allowed by a six-half-hour expected active-spell cap. Risk ratios divide the exact realized-path risk, computed from centered assignments $A_t-p_0$, by the corresponding population Markov risk. The condition number is for the realized lagged-assignment Gram matrix. Active spells count entries into the active state, including an initially active path. No simulated Gram matrix was singular.
\end{minipage}
\end{table}

\subsection{Multi-objective and multi-arm diagnostics}

\begin{figure}[htbp]
    \centering
    \includegraphics[width=0.86\textwidth]{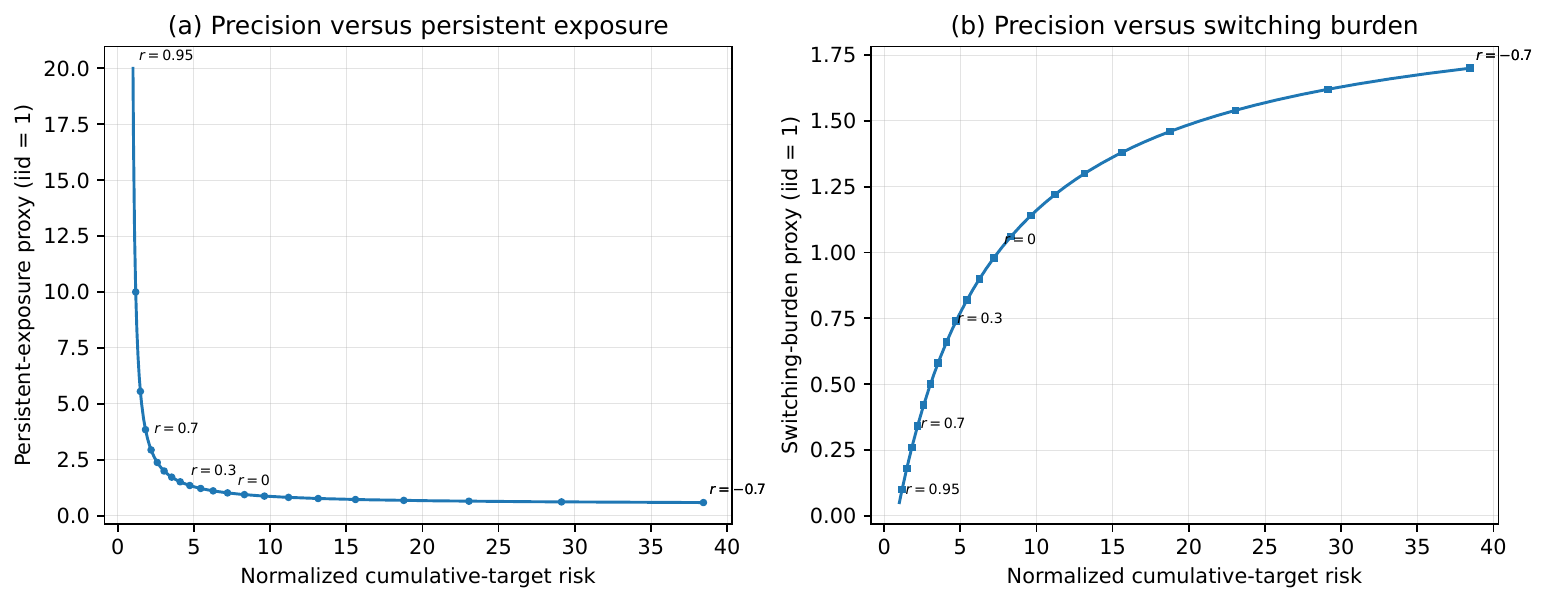}
    \caption{Illustrative Pareto frontier for cumulative-target precision versus non-statistical costs. Precision risk is normalized to the best cumulative-target Markov risk on the grid. The persistent-exposure proxy is normalized to iid assignment and rises with expected episode length; the switching-burden proxy is normalized to iid assignment and falls with persistence. The frontier is a planning diagnostic, not a structural business-loss estimate.}
    \label{fig:pareto-frontier}
\end{figure}

\begin{figure}[htbp]
    \centering
    \includegraphics[width=0.86\textwidth]{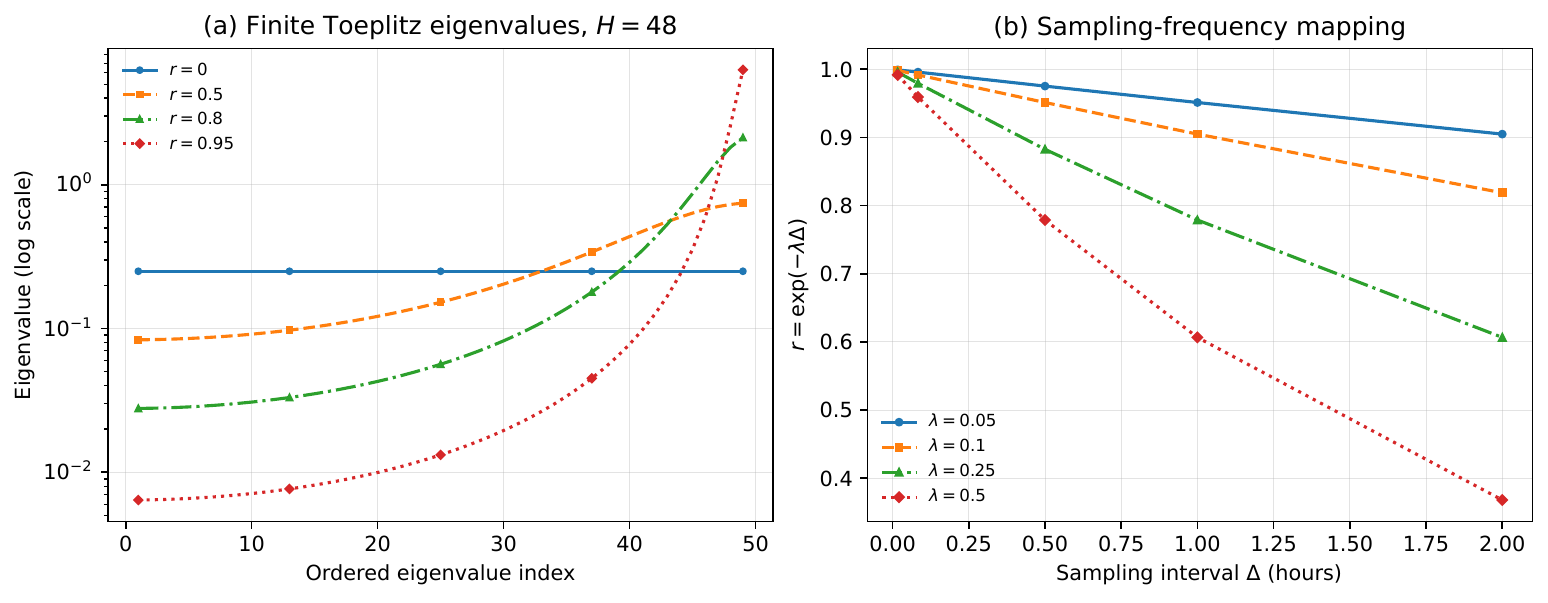}
    \caption{Spectral and sampling-frequency diagnostics. The left panel shows ordered eigenvalues of the finite AR(1)-Toeplitz information matrix $Q_{48}(r)$ for several persistence values. The right panel translates fixed continuous-time switching rates into the discrete persistence $r=\exp(-\lambda\Delta)$ as the sampling interval $\Delta$ changes.}
    \label{fig:spectral-sampling}
\end{figure}

For a stationary three-arm chain with uniform arm shares, stack two orthonormal Helmert contrasts at each lag. Under uniform switching among the two alternative arms, the lag-$k$ block is $(1/3)r_s^k I_2$, where $r_s=(3s-1)/2$ indexes the common stay probability $s$. A directional-bias parameter $\delta=0.80$, which allocates 90 percent of switches clockwise and 10 percent counterclockwise, preserves uniform shares but makes the lag blocks matrix-valued. Figure~\ref{fig:multiarm-information} compares the two cases at $H=8$ over the common feasible grid $r_s\in[-0.45,0.90]$. For any fixed unit arm-contrast direction, applying the symmetric-chain grid rule to the directional chain raises risk by 4.8 percent for the immediate target, 0.0 percent for the cumulative target, 2.3 percent for the delayed and rebound targets, and 45.1 percent for the alternating target. A sensitivity grid $\delta\in\{0,0.2,0.4,0.6,0.8\}$ shows that the alternating-target penalty is 1.00 through $\delta=0.4$, rises to 1.04 at $\delta=0.6$, and reaches 1.45 at $\delta=0.8$. Monte Carlo averages of the realized block Gram matrix, using 1,000 paths at $T=2{,}000$ and 250 paths at $T=17{,}520$, differ from the analytical block-Toeplitz matrix by at most 0.46 percent in relative Frobenius norm. This diagnostic establishes the need for a matrix-valued calibrated selector in sufficiently asymmetric multi-arm designs; it is not a general multi-arm optimality theorem.

\begin{figure}[htbp]
    \centering
    \includegraphics[width=0.98\textwidth]{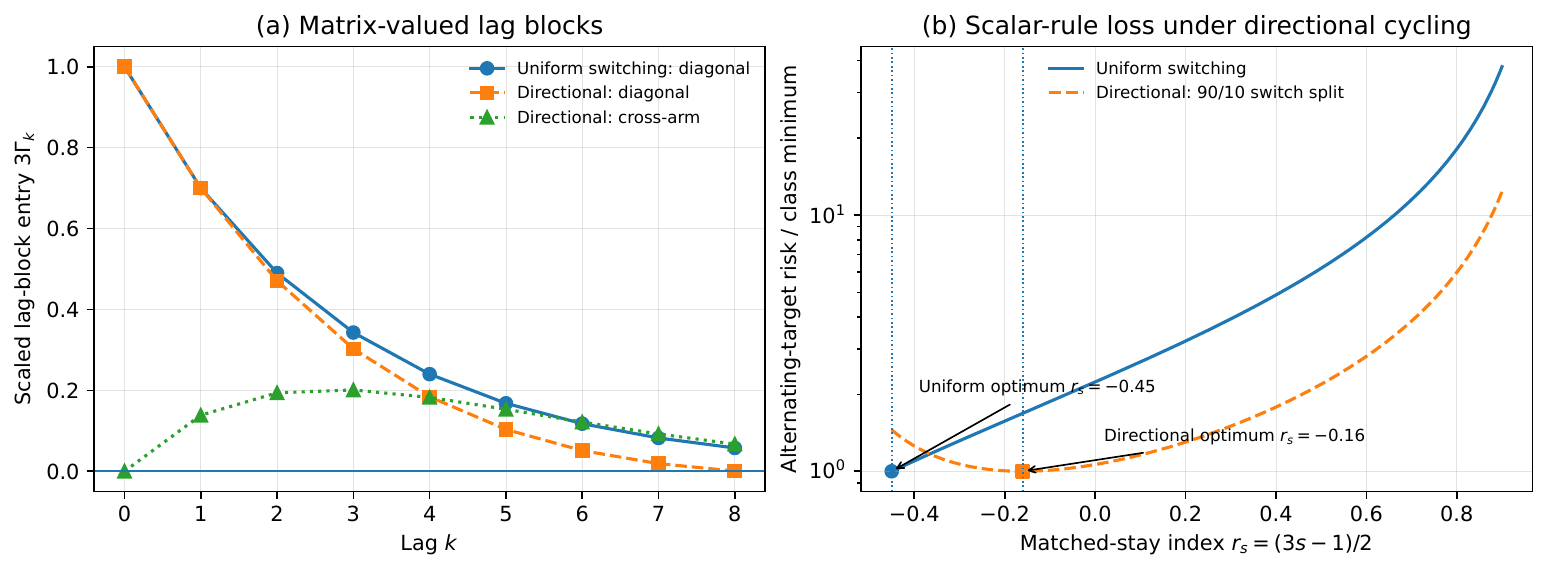}
    \caption{Three-arm block-Toeplitz information diagnostic. Panel (a) plots selected entries of the contrast-coded lag block $3\Gamma_k$ at matched-stay index $r_s=0.70$. Uniform switching yields scalar blocks, whereas $\delta=0.80$ (90 percent clockwise and 10 percent counterclockwise, conditional on switching) induces nonzero cross-arm entries. Panel (b) reports alternating-target risk normalized by the minimum on $r_s\in[-0.45,0.90]$ within each design class. The symmetric-chain grid rule selects $r_s=-0.45$; the directional chain selects $r_s=-0.16$, and imposing the symmetric rule raises directional-chain risk by 45.1 percent.}
    \label{fig:multiarm-information}
\end{figure}

The deterministic large-horizon conditioning grid supports the same conclusion: high-persistence, large-$H$ targets may require regularization or a shorter pre-specified target menu.

\section{Assumptions and proofs}\label{app:proofs}

\begin{assumption}[Benchmark assignment process]
The assignment process is a stationary balanced two-state Markov chain with $\Pr(A_t=1)=1/2$, $\Pr(A_t=A_{t-1})=s$, and $r=2s-1\in(-1,1)$.
\end{assumption}

\begin{assumption}[Finite-memory response benchmark]
For the analytical benchmark, the outcome satisfies
\[
    Y_t=\mu_t+\sum_{\ell=0}^{H}g_\ell A_{t-\ell}+\varepsilon_t,
\]
where $E[\varepsilon_t\mid A_{-\infty:\infty}]=0$. For the closed-form covariance result, we impose the conditional benchmark moments $E[\varepsilon_t^2\mid A_{-\infty:\infty}]=\sigma^2$ and $E[\varepsilon_t\varepsilon_s\mid A_{-\infty:\infty}]=0$ for $t\ne s$. The HAC extension replaces these benchmark moments with a standard mixing and long-run covariance condition for $X_t\varepsilon_t$.
\end{assumption}

\paragraph{Proof of Lemma \ref{lem:uniform_markov_mixing}.}
For a two-state stationary Markov chain with transition matrix $P(p,r)$, the eigenvalues are $1$ and $r$. On the centered one-dimensional subspace, $P^k$ therefore contracts at rate $|r|^k$. Because $p$ is bounded away from zero and one and the feasible set is bounded away from the nonergodic endpoints, there is a finite constant $C$ such that the total-variation distance between the conditional distribution of $A_{t+k}$ given $A_t$ and the stationary distribution is at most $C \rho_{\max}^k$ uniformly over the class. The alpha-mixing coefficient is bounded by this total-variation distance, giving the displayed bound for $A_t$. The lag vector $X_t=(\widetilde A_t,\ldots,\widetilde A_{t-H})$ only shifts the separation between sigma fields, so its coefficient is bounded by the same expression with $k-H$ for $k>H$. Standard product and measurable-transform inequalities for alpha-mixing processes give the stated score-process implication when the residual process is independent of assignment conditional on controls or jointly alpha-mixing with uniformly summable coefficients; see \citet[Ch.~14]{Davidson1994} and \citet[Ch.~3]{White2001} for these closure arguments.

\paragraph{Proof of Lemma \ref{lem:spell_lower_bound}.}
Let $B_t=\mathbf 1\{A_{t-1}=0,A_t=1\}$. Under stationarity,
\[
    E B_t=P(A_{t-1}=0)P(A_t=1\mid A_{t-1}=0)=(1-p)p(1-r).
\]
The finite-state Markov chain is ergodic for every fixed interior parameter, so the ergodic theorem gives $T^{-1}\sum_{t=2}^T B_t\to (1-p)p(1-r)$ almost surely. Uniformly over a compact interior class, $(1-p)p(1-r)$ is bounded below by a positive constant; hence $N_{1T}\ge cT$ with probability approaching one for any smaller constant $c$. If $r=r_T\to1$, the expected number of active-spell starts is of order $T p(1-p)(1-r_T)$. This motivates the operational many-spell diagnostic used above: require this quantity to diverge when appealing to many-episode normal approximations. If it is bounded, the assignment path has only a bounded expected number of active episodes.

\paragraph{Proof of Lemma \ref{lem:target_invariance}.}
The population least-squares coefficient in the path-adjusted regression is
\[
    \bar g=Q_X^{-1}E[X_tY_t].
\]
Substituting $Y_t=\alpha+X_t'g+\varepsilon_t$ and using $E[X_t\varepsilon_t]=0$ and the residualized convention for constants and controls gives
\[
    \bar g=Q_X^{-1}E[X_tX_t']g=g.
\]
The assignment mechanism therefore affects $Q_X$ and the covariance of the estimator, but not the population coefficient as long as $Q_X$ is nonsingular.

\paragraph{Proof of Proposition \ref{prop:naive_mixture}.}
By the centered finite-memory representation in Proposition~\ref{prop:naive_mixture},
\[
    Y_{t+h}=\mu_h+\sum_{\ell=0}^{H}g_\ell \widetilde A_{t+h-\ell}+\varepsilon_{t+h}.
\]
After residualizing constants and deterministic controls, the population coefficient from regressing $Y_{t+h}$ on $\widetilde A_t$ alone is
\[
    \beta_h^{\rm naive}(r)=\frac{\operatorname{Cov}(Y_{t+h},\widetilde A_t)}{\operatorname{Var}(\widetilde A_t)}.
\]
The covariance between $\widetilde A_t$ and $\widetilde A_{t+h-\ell}$ is $\sigma_A^2r^{|h-\ell|}$, and the error term is orthogonal to assignment. Dividing by $\sigma_A^2$ gives the stated mixture formula.

\paragraph{Proof of Proposition \ref{prop:unbalanced}.}
For a stationary two-state Markov chain with stationary treatment share $p$ and second eigenvalue $r$, the centered process $\widetilde A_t=A_t-p$ satisfies $E\widetilde A_t=0$, $\operatorname{Var}(\widetilde A_t)=p(1-p)$, and $E[\widetilde A_t\mid \widetilde A_{t-1}]=r\widetilde A_{t-1}$. Iterating gives $E[\widetilde A_t\widetilde A_{t-k}]=p(1-p)r^{|k|}$, which yields \eqref{eq:unbalanced_Q}. The transition probabilities displayed above are obtained by imposing the stationary probability $p$ and eigenvalue $r$ on a two-state transition matrix.

\paragraph{Proof of Corollary \ref{cor:budgeted_sparse}.}
Proposition~\ref{prop:unbalanced} gives $Q_H(p_0,r)=p_0(1-p_0)R_H(r)$. Inverting yields $Q_H(p_0,r)^{-1}=R_H(r)^{-1}/\{p_0(1-p_0)\}$. Hence the factor $p_0(1-p_0)$ affects the scale of the contrast variance but not the persistence minimizer when $p_0$ is fixed. The active-state transition probability is $p_0+(1-p_0)r$, so the probability of leaving the active state is $(1-p_0)(1-r)$. The active spell length, conditional on starting in the active state and including the first active period, is geometric on $\{1,2,\ldots\}$ with this exit probability. Its expectation is therefore $\{(1-p_0)(1-r)\}^{-1}$. Imposing an upper bound $L_{\max}$ gives the displayed inequality.

\paragraph{Proof of Corollary \ref{cor:active_share_budget}.}
For fixed persistence and target weights, Corollary~\ref{cor:budgeted_sparse} implies that the variance criterion is a positive constant times $[p(1-p)]^{-1}$. The function $p(1-p)$ is increasing on $(0,1/2]$ and maximized at $p=1/2$. The stated conclusions follow immediately.

\paragraph{Proof of Proposition \ref{prop:info}.}
Because the chain is balanced, $E[\widetilde A_t]=0$ and $\operatorname{Var}(\widetilde A_t)=1/4$. The transition matrix has second eigenvalue $r=2s-1$, hence $\operatorname{Corr}(\widetilde A_t,\widetilde A_{t-k})=r^k$. This yields $E[X_tX_t']=Q_H(r)$. The displayed inverse is the standard inverse of an AR(1) correlation matrix; multiplying it by $(r^{|i-j|})_{i,j}$ gives the identity.

\paragraph{Proof of Proposition \ref{prop:covariance}.}
Under the Markov assignment process, $T^{-1}\sum_tX_tX_t'\to_p Q_H(r)$. Conditional serial uncorrelatedness implies that the score process has zero autocovariances given the assignment path, so its limiting covariance is $\sigma^2Q_H(r)$. Hence the score term satisfies a central limit theorem,
\[
    T^{-1/2}\sum_tX_t\varepsilon_t\Rightarrow N(0,\sigma^2Q_H(r))
\]
in the homoskedastic martingale-difference benchmark. The usual least-squares expansion gives
\[
    \sqrt T(\hat g-g)=\left(T^{-1}\sum_tX_tX_t'\right)^{-1}T^{-1/2}\sum_tX_t\varepsilon_t+o_p(1),
\]
so the limiting covariance is $\sigma^2Q_H(r)^{-1}$. With serial correlation, $\sigma^2Q_H(r)$ is replaced by the long-run covariance $\Omega_H(r)$.

\paragraph{Proof of Proposition \ref{prop:optimal}.}
Substituting the tridiagonal inverse into $c'Q_H(r)^{-1}c$ gives
\[
    \mathcal R_c(r)\propto \frac{a+br^2-2dr}{1-r^2}.
\]
Differentiating yields
\[
    \mathcal R_c'(r)\propto (a+b)r-d(1+r^2).
\]
Thus any interior stationary point solves $dr^2-(a+b)r+d=0$. If $d=0$, the derivative has the sign of $r$ and the unique interior minimizer is $r=0$. If $0<|d|<(a+b)/2$, the quadratic has exactly one root in $(-1,1)$,
\[
    r^{\mathrm{int}}_c=\frac{(a+b)-\sqrt{(a+b)^2-4d^2}}{2d}.
\]
The second root lies outside $(-1,1)$, so the feasible optimum is the projection of $r^{\mathrm{int}}_c$ onto $\mathcal R$. If $|d|=(a+b)/2$, the stationary point is the boundary value $r=\operatorname{sign}(d)$ and is not feasible in the open Markov-chain parameter space. The objective approaches its infimum as $r$ approaches that boundary, so the feasible optimum is the nearest endpoint of $\mathcal R$. Finally, $|d|\le(a+b)/2$ follows from $2|c_{j-1}c_j|\le c_{j-1}^2+c_j^2$ summed over adjacent pairs.

\paragraph{Proof of Lemma \ref{lem:spectral_endpoint}.}
By Parseval's identity,
\[
\frac{1}{2\pi}\int_{-\pi}^{\pi}\frac{|C(e^{-i\omega})|^2}{f_A(\omega;r)}\,d\omega
=\frac{4}{1-r^2}\frac{1}{2\pi}\int_{-\pi}^{\pi}|1-re^{-i\omega}|^2|C(e^{-i\omega})|^2\,d\omega.
\]
Expanding $|1-re^{-i\omega}|^2=1+r^2-r(e^{i\omega}+e^{-i\omega})$ gives
\[
\frac{4}{1-r^2}\left\{(1+r^2)\sum_{j=0}^{H}c_j^2-2r\sum_{j=1}^{H}c_{j-1}c_j\right\}.
\]
Subtracting the tridiagonal inverse expression in \eqref{eq:closed_form_risk} leaves $4r^2(c_0^2+c_H^2)/(1-r^2)$, which proves the identity.

\paragraph{Proof of Proposition \ref{prop:omitted_bias}.}
The population coefficient from the truncated regression is
\[
    \bar g_H(r)=Q_{XX}(r)^{-1}E[X_tY_t].
\]
Since $Y_t=X_t'g+Z_t'\gamma+\varepsilon_t$ and $E[X_t\varepsilon_t]=0$,
\[
    \bar g_H(r)=g+Q_{XX}(r)^{-1}Q_{XZ}(r)\gamma.
\]
For a contrast $c'g$, the omitted-lag bias is
\[
    c'Q_{XX}(r)^{-1}Q_{XZ}(r)\gamma.
\]
Maximizing its square over $\|\gamma\|_2\le M$ gives the stated norm expression.

\paragraph{Closed-form part of Proposition \ref{prop:omitted_bias}.}
For the AR(1)-Toeplitz covariance matrix, the linear projection of $\widetilde A_{t-H-m}$ on $X_t=(\widetilde A_t,\ldots,\widetilde A_{t-H})'$ depends only on the last included lag:
\[
    E[\widetilde A_{t-H-m}\mid X_t]_{\rm lin}=r^m\widetilde A_{t-H}.
\]
Equivalently, the $m$th row of $Q_{ZX}(r)Q_{XX}(r)^{-1}$ is $r^m e_H'$. Multiplying by $c$ gives $c_Hr^m$. The worst-case squared bias over $\|\gamma\|_2\le M$ is the squared norm of this vector times $M^2$, which gives the displayed geometric sum.

\paragraph{Proof of Corollary \ref{cor:buffer}.}
If the reported weights are zero on the final estimated horizon, then $c_H=0$. Proposition \ref{prop:omitted_bias} implies that the omitted-tail projection vector has norm zero for the reported contrast. Hence the worst-case omitted-tail bias term is zero.

\paragraph{Proof of Proposition \ref{prop:calibrated_selector}.}
Let $r_W^\star\in\arg\min_{r\in\mathcal R}R_W(r)$. Approximate optimality and uniform convergence imply
\[
R_W(\hat r_W)\le \widehat R_W(\hat r_W)+o_p(1)
\le \widehat R_W(r_W^\star)+o_p(1)
\le R_W(r_W^\star)+o_p(1).
\]
The reverse inequality $R_W(\hat r_W)\ge \inf_{r\in\mathcal R}R_W(r)$ is deterministic, giving oracle-risk convergence. If the minimizer is unique and separated, standard argmin consistency gives $\hat r_W\to_p r_W^\star$.

\paragraph{Proof of Proposition \ref{prop:menu_robust}.}
Continuity of each $\mathcal R_c$ and compactness of $\mathcal R$ imply that the normalized loss $L_{\mathcal C}(r)$ is continuous and attains a minimum. Uniform convergence over the finite menu implies uniform convergence of the plug-in normalized loss because the denominators $\inf_{u\in\mathcal R}\mathcal R_c(u)$ are positive and their plug-in counterparts converge uniformly. Approximate optimality of $\widehat r_{\mathcal C}$ then gives
\[
    L_{\mathcal C}(\widehat r_{\mathcal C})
    \le \widehat L_{\mathcal C}(\widehat r_{\mathcal C})+o_p(1)
    \le \inf_{r\in\mathcal R}\widehat L_{\mathcal C}(r)+o_p(1)
    \le \inf_{r\in\mathcal R}L_{\mathcal C}(r)+o_p(1),
\]
while the reverse inequality holds by definition of the infimum. This proves oracle menu-risk convergence.

\paragraph{Proof of Proposition \ref{prop:realized_schedule}.}
Conditional on a fixed realized assignment path, the residualized least-squares estimator based on the Moore--Penrose inverse satisfies
\[
    \hat g-g=(X'X)^+X'\varepsilon
\]
for estimable linear combinations of $g$, where \(X'X=T\hat Q_T\). In the homoskedastic benchmark, \(E[\varepsilon\varepsilon'\mid X]=\sigma^2 I\), so the covariance on the estimable subspace is
\[
    \operatorname{Var}(\hat g\mid X)=\sigma^2(X'X)^+=\frac{\sigma^2}{T}\hat Q_T^+.
\]
If $c\in\operatorname{col}(X'X)$, the conditional variance of \(c'\hat g\) is therefore \((\sigma^2/T)c'\hat Q_T^+c\). If $c$ is not in that column space, the contrast is not fully estimable without regularization or additional identifying information; the ridge diagnostic above is then the conservative object to report.

\section{Pre-analysis protocol}\label{app:pap-checklist}
The pre-analysis plan should record the target, design menu, pilot or calibration source, tuning constants, estimator, inference route, and simulation design before the confirmatory experiment.
\begin{itemize}
    \item \emph{Target and loss.} State the primary LP contrast or curve, the weight matrix $W$, any secondary targets, and the target-level multiplicity rule.
    \item \emph{Design menu.} Record the feasible persistence interval, active-share or run-length constraints, non-Markov alternatives, and any threshold for switching from the benchmark to calibrated selection.
    \item \emph{Calibration and tuning.} Specify the pilot or historical source for $\widehat V(d)$, HAC bandwidth caps, ridge or regularization grid, buffer horizon, and numerical conditioning checks.
    \item \emph{Estimator and inference.} Pre-specify the LP estimator, residualization controls, construction of simultaneous bands, the randomization-first trigger, and a selection-adjusted or sample-split protocol.
    \item \emph{Documentation and governance.} Record the assignment seed, exclusions, preprocessing rules, analysis specification, numerical diagnostics, and finite-sample balance reports.
\end{itemize}

\section{Low Carbon London calibration details}\label{app:lcl-details}

\subsection{Calendar, tariff-state, and information diagnostics}
Residential demand response is often state dependent: weekday and weekend behavior, seasonal load, and peak-period load can imply different dynamic response shapes. To make this limitation quantitative, we use a small design-sensitivity calculation in which the reporting object is a weighted weekday/weekend menu. The calculation is not an estimate of heterogeneous treatment effects in the LCL data. It asks how the common persistence recommendation would move if the pre-analysis reporting target assigned different horizon-weight vectors to two calendar states. The common-persistence recommendation ranges from iid assignment ($r^*=0$) to persistent assignment ($r^*=0.693$), confirming that state-specific reporting objects can materially change the design rule.

The calculation clarifies the implementation choice. If a single common assignment process must be used, the researcher should pre-specify weights over state-specific targets and minimize the weighted HW-LP risk. If the platform can randomize with state-dependent persistence, the design problem becomes a time-inhomogeneous Markov assignment problem and must be evaluated with the realized state process and the corresponding residualized covariance matrix. Because those states are calendar states known before assignment, this extension can preserve exogeneity. It does, however, replace the stationary Toeplitz information matrix with a time-inhomogeneous design matrix, so closed-form persistence rules are generally lost and HAC inference must be based on the realized or calibrated score process. The rule should therefore be planned and logged rather than chosen adaptively from outcomes.

The observed Low Carbon London tariff schedule is not used for causal identification, but it is useful as a descriptive design benchmark. The tariff workbook yields high-price, low-price, and non-normal binary active-half-hour indicators. Their active shares, lag-one persistence, event counts, and active-spell lengths reveal signals that are sparse and persistent relative to balanced Markov assignment, precisely the kind of design feature that the HW-LP criterion is meant to evaluate.

We translate these observed schedules into the information-matrix language developed above. For each active-half-hour indicator, we compute
\[
    \widehat Q_{\mathrm{obs}}=(T-H)^{-1}\sum_t X_tX_t'
\]
after centering by the observed active share, and evaluate the same horizon-weighted contrasts used in the design comparisons. The entries are descriptive information diagnostics, not causal estimates. They show that the actual tariff signals are operationally realistic but information-poor for some LP contrasts because the active shares are small and the signals are highly persistent. This diagnostic is preferable to treating the historical tariff as an optimized experiment: the historical trial was not designed to minimize the present HW-LP criterion.

\begin{proposition}[Realized-schedule information diagnostic]\label{prop:realized_schedule}
Condition on a fixed assignment path and let $\widehat Q_T=T^{-1}\sum_t X_tX_t'$ be the empirical second-moment matrix of the residualized lagged assignment vector. In the homoskedastic benchmark, for any estimable contrast $c\in\operatorname{col}(\widehat Q_T)$, the conditional variance of the least-squares estimator of $c'g$ is proportional to $c'\widehat Q_T^+c/T$. If the contrast has components outside $\operatorname{col}(\widehat Q_T)$, the Moore--Penrose diagnostic is not a full variance for that contrast; a ridge-stabilized diagnostic $c'(\widehat Q_T+\lambda I)^{-1}c$ should be reported instead.
\end{proposition}

Proposition~\ref{prop:realized_schedule} is the basis for the observed-schedule diagnostics: the historical tariff schedules are treated as fixed assignment paths and their lag matrices are evaluated against the same horizon-weighted targets as the candidate designs. For the ridge diagnostic we use $\lambda=\eta\lambda_{\max}(\widehat Q_T)$ with $\eta=10^{-3}$ as the default and report sensitivity at $10^{-4}$ and $10^{-2}$. This scale-free rule regularizes nearly unobserved lag directions without changing well-identified directions. A replay design that injects response paths under the historical indicators is useful as an additional sensitivity check, but we use the information diagnostic because it cleanly separates design informativeness from the particular replay normalization.

Finite-sample design diagnostics for representative candidate designs confirm the information-matrix mechanism behind the empirical results: alternating and persistent designs have larger condition numbers than iid assignment, while sparse designs have smaller active shares and therefore lower information scale. This compact diagnostic complements the scalar-MSE and selector tables without duplicating them.

The injected paths are calibrated to be economically visible but not extreme relative to the baseline load series. The response-path scale diagnostic measures the maximum absolute horizon effect relative to the normal-tariff baseline mean. Across scenarios, the maximum absolute horizon effect ranges from 7.0 to 16.4 percent of the baseline mean. The calibrated evaluation therefore translates the analytical design rules into a realistic high-frequency load environment while preserving known ground truth for the target response path.

Taken together, the empirical results imply a target-contingent decision rule. If the reported object is a single scalar target, the relevant comparison is Table~\ref{tab:lcl-key-results}. If the experiment must satisfy an active-share budget, the sparse Markov table provides the more relevant benchmark. If calendar-state response heterogeneity is central, the appendix state-specific sensitivity calculation indicates whether a common-persistence rule is robust. If the question is how informative an already implemented tariff schedule would have been for a modern LP target, the observed-schedule diagnostic is the appropriate object. These are related design questions, but they are not interchangeable.

The empirical lesson is the same as the theoretical one: specify the dynamic object first, then choose the assignment path. ATE-oriented randomization is appropriate when the reported object is contemporaneous. It is not generally appropriate when the object is cumulative, delayed, or rebound-shaped. Conversely, persistent assignment is useful for smooth dynamic responses but can be inefficient for contrasts that require separating adjacent horizons.

\subsection{Pilot-window stability and external validity}
The calibrated selector is intended to be estimated on pilot information and then used for a confirmatory design. To check whether the aggregate LCL calibration is dominated by one part of the year, Appendix Table~\ref{tab:residualized_split_stability} splits the residualized aggregate baseline into two pilot/evaluation windows. The baseline series is residualized on hour-of-day, day-of-week, and month indicators before the residual autocovariance is estimated. Selecting $r$ on one half and evaluating target risk on the other half leaves the calibrated recommendation unchanged for immediate and cumulative targets, moves the delayed recommendation only between $0.77$ and $0.79$, and moves the rebound recommendation only between $0.27$ and $0.29$. The largest evaluation loss in the table is below one percent. The fine grid therefore reduces the concern that the earlier stability result was driven by a coarse candidate menu. This does not establish household-level transportability, but it supports the use of the LCL exercise as an aggregate covariance stress test rather than a result specific to one sample split. Because the response paths are known by construction, the exercise evaluates covariance and design risk rather than structural demand responses. It therefore cannot validate structural demand recovery, household-level elasticities, or the historical tariff effect. The LCL results calibrate design risk under realistic half-hourly covariance; they are not forecasts of response magnitudes in contemporary demand-response programs.

Appendix Table~\ref{tab:pilot_length_sensitivity} reports a more granular pilot-length diagnostic using one-month, three-month, six-month, and full-year pilot windows. The one-month pilots are a demanding stress test: they occasionally under-select the cumulative target's upper-endpoint persistence and can lose up to 17.5 percent relative to the complementary-window oracle. For delayed and rebound targets, however, even the one-month windows have maximum evaluation losses of only 2.3 percent and 0.4 percent. Three-month or longer pilots are stable for all four targets, with delayed and rebound recommendations remaining in narrow ranges around the full-year calibrated values. This diagnostic supports a practical rule: use very short pilots for screening, but use at least a quarter of high-frequency baseline data when the cumulative target can bind at the feasible persistence cap.

Appendix Table~\ref{tab:loadstate_calibration_sensitivity} adds a second calibration check using load-state slices of the same residualized aggregate series. The pooled recommendation remains essentially stable across low-, middle-, and high-load periods: immediate targets select iid assignment, cumulative targets select the upper fine-grid endpoint, delayed targets select high persistence around $0.77$--$0.81$, and rebound targets select moderate persistence around $0.25$--$0.28$. Using the pooled recommendation within each load state costs at most about four tenths of one percent relative to the load-state oracle in this aggregate diagnostic. Because the exercise is based on aggregate baseline dynamics rather than household-level causal contrasts, the table is a residual-covariance stability check for the aggregate target rather than household-level heterogeneity evidence.

Appendix Table~\ref{tab:synthetic_transportability} adds a synthetic transportability diagnostic that is independent of the London baseline. It holds the Markov assignment class and four reporting targets fixed, but changes the residual score autocovariance across iid, AR(1), AR(2), seasonal, volatility-clustered, heavy-tailed, and regime-switching proxies. The calibrated selector equals the closed-form benchmark under iid residuals, but delayed-window recommendations move toward higher persistence under serially dependent or regime-switching residual dynamics. The heavy-tail, two-regime, and compound rows stress-test non-Gaussian tails, state-dependent variance, and their joint effect rather than adding a second field dataset; they are meant to probe transportability of the covariance logic, not contemporary demand-response magnitudes. The table therefore supports the calibrated design argument: the closed form explains the target-shape margin, while calibrated covariance selection decides whether realistic dynamics move the field recommendation.

The Low Carbon London exercise is a design evaluation rather than a causal analysis of the observed tariff. The observed data come from a 2013 residential demand-response setting that predates widespread electric-vehicle charging, behind-the-meter batteries, and smart-thermostat automation. They supply realistic half-hourly load dynamics and observed tariff schedules that can be evaluated as fixed design paths, but they are not a forecast of response magnitudes in contemporary demand-response programs with widespread electric-vehicle charging, behind-the-meter storage, or smart-thermostat automation. Appendix Table~\ref{tab:synthetic_transportability} partially addresses sensitivity to residual dynamics; external validity to current technologies still requires fresh field data. The treatment effects used to compare alternative assignment designs are injected and known. This structure lets the design risk be evaluated against a fixed target. The empirical section uses replication-level response estimates to compute full-covariance selector diagnostics, paired Monte Carlo standard errors, sparse active-share design comparisons, observed-tariff information diagnostics, and finite design-matrix diagnostics.

The Low Carbon London exercise is a semi-synthetic calibrated design evaluation. It uses the 2013 half-hourly residential-load panel to calibrate baseline dynamics and then injects known response paths under alternative assignment designs. The observed dynamic time-of-use tariff is not used as a source of causal identification. The normal-tariff households are used as the main baseline source to avoid building the calibration on households exposed to the dynamic tariff. Aggregation is by timestamp and tariff group, using available household readings at each half-hour. The evaluation fixes the panel frequency, household-reading ranges, target horizon, Monte Carlo replication counts, common-random-number policy, scenario definitions, and feasible persistence sets before comparing designs.

The calendar-state sensitivity and compact finite-design diagnostics show that the qualitative LCL selector conclusions are not driven by a single finite assignment matrix.

\section*{Data Availability Statement}
The empirical application uses the Low Carbon London Project: Data from the Dynamic Time-of-Use Electricity Pricing Trial, 2013, catalogued by the UK Data Service as Study 7857 and cited as \citet{StrbacEtAl2024}. The data contain half-hourly household electricity-consumption readings, timestamps, tariff-group identifiers, and recorded tariff schedules from the 2013 residential trial. We use the full calendar year and average non-missing readings within timestamp and tariff group, producing 17,520 half-hour observations. The normal-tariff group, with roughly 4,000--4,400 readings in most half-hours, provides the baseline series. The dynamic-tariff sample has roughly 1,000--1,100 readings in most half-hours; its recorded high-price, low-price, and non-normal tariff indicators are used only as fixed assignment paths in information diagnostics. The observed tariff is not used to identify its historical causal effect. Known response paths are injected into the normal-tariff baseline dynamics to compare assignment designs. No new observational data were collected for this study.

\bibliographystyle{plainnat}
\bibliography{references}

@article{AitSahalia2002,
  author = {A{\"\i}t-Sahalia, Y.},
  title = {{{Maximum likelihood estimation of discretely sampled diffusions: A closed-form approximation approach}}},
  journal = {{{Econometrica}}},
  year = {2002},
  volume = {70},
  number = {1},
  pages = {223--262}
}

@article{Akaike1974,
  author = {Akaike, H.},
  title = {{{A new look at the statistical model identification}}},
  journal = {{{IEEE Transactions on Automatic Control}}},
  year = {1974},
  volume = {19},
  number = {6},
  pages = {716--723}
}

@article{Allcott2011,
  author = {Allcott, H.},
  title = {{{Social norms and energy conservation}}},
  journal = {{{Journal of Public Economics}}},
  year = {2011},
  volume = {95},
  number = {9--10},
  pages = {1082--1095}
}

@article{AllcottRogers2014,
  author = {Allcott, H. and Rogers, T.},
  title = {{{The short-run and long-run effects of behavioral interventions: Experimental evidence from energy conservation}}},
  journal = {{{American Economic Review}}},
  year = {2014},
  volume = {104},
  number = {10},
  pages = {3003--3037}
}

@book{Amari2016,
  author = {Amari, S.},
  title = {{{Information Geometry and Its Applications}}},
  year = {2016},
  publisher = {Springer}
}

@article{Andrews1991,
  author = {Andrews, D. W. K.},
  title = {{{Heteroskedasticity and autocorrelation consistent covariance matrix estimation}}},
  journal = {{{Econometrica}}},
  year = {1991},
  volume = {59},
  number = {3},
  pages = {817--858}
}

@article{AndrewsMonahan1992,
  author = {Andrews, D. W. K. and Monahan, J. C.},
  title = {{{An improved heteroskedasticity and autocorrelation consistent covariance matrix estimator}}},
  journal = {{{Econometrica}}},
  year = {1992},
  volume = {60},
  number = {4},
  pages = {953--966}
}

@article{AronowSamii2017,
  author = {Aronow, P. M. and Samii, C.},
  title = {{{Estimating average causal effects under general interference, with application to a social network experiment}}},
  journal = {{{Annals of Applied Statistics}}},
  year = {2017},
  volume = {11},
  number = {4},
  pages = {1912--1947}
}

@article{AtheyEcklesImbens2018,
  author = {Athey, S. and Eckles, D. and Imbens, G. W.},
  title = {{{Exact p-values for network interference}}},
  journal = {{{Journal of the American Statistical Association}}},
  year = {2018},
  volume = {113},
  number = {521},
  pages = {230--240}
}

@incollection{AtheyImbens2017,
  author = {Athey, S. and Imbens, G. W.},
  title = {{{The econometrics of randomized experiments}}},
  booktitle = {{{Handbook of Economic Field Experiments}}},
  editor = {Banerjee, A. V. and Duflo, E.},
  volume = {1},
  pages = {73--140},
  publisher = {Elsevier},
  year = {2017}
}

@book{AtkinsonDonevTobias2007,
  author = {Atkinson, A. C. and Donev, A. N. and Tobias, R. D.},
  title = {{{Optimum Experimental Designs, with SAS}}},
  year = {2007},
  publisher = {Oxford University Press}
}

@article{AuerbachGorodnichenko2013,
  author = {Auerbach, A. J. and Gorodnichenko, Y.},
  title = {{{Output spillovers from fiscal policy}}},
  journal = {{{American Economic Review}}},
  year = {2013},
  volume = {103},
  number = {3},
  pages = {141--146}
}

@inproceedings{BakshyEcklesBernstein2014,
  author = {Bakshy, E. and Eckles, D. and Bernstein, M. S.},
  title = {{{Designing and deploying online field experiments}}},
  booktitle = {{{Proceedings of the 23rd International Conference on World Wide Web}}},
  pages = {283--292},
  year = {2014}
}

@article{BarnichonBrownlees2019,
  author = {Barnichon, R. and Brownlees, C.},
  title = {{{Impulse response estimation by smooth local projections}}},
  journal = {{{Review of Economics and Statistics}}},
  year = {2019},
  volume = {101},
  number = {3},
  pages = {522--530}
}

@article{BasseDingToulis2023,
  author = {Basse, G. W. and Ding, Y. and Toulis, P.},
  title = {{{Minimax designs for causal effects in temporal experiments with treatment habituation}}},
  journal = {{{Biometrika}}},
  year = {2023},
  volume = {110},
  number = {1},
  pages = {155--168}
}

@article{BergmeirHyndmanKoo2018,
  author = {Bergmeir, C. and Hyndman, R. J. and Koo, B.},
  title = {{{A note on the validity of cross-validation for evaluating autoregressive time series prediction}}},
  journal = {{{Computational Statistics \& Data Analysis}}},
  year = {2018},
  volume = {120},
  pages = {70--83}
}

@article{BerkEtAl2013,
  author = {Berk, R. and Brown, L. and Buja, A. and Zhang, K. and Zhao, L.},
  title = {{{Valid post-selection inference}}},
  journal = {{{Annals of Statistics}}},
  year = {2013},
  volume = {41},
  number = {2},
  pages = {802--837}
}

@article{BojinovRambachanShephard2021,
  author = {Bojinov, I. and Rambachan, A. and Shephard, N.},
  title = {{{Panel experiments and dynamic causal effects: A finite population perspective}}},
  journal = {{{Quantitative Economics}}},
  year = {2021},
  volume = {12},
  number = {4},
  pages = {1171--1196}
}

@article{BojinovSimchiLeviZhao2023,
  author = {Bojinov, I. and Simchi-Levi, D. and Zhao, J.},
  title = {{{Design and analysis of switchback experiments}}},
  journal = {{{Management Science}}},
  year = {2023},
  volume = {69},
  number = {7},
  pages = {3759--3777}
}

@article{BojinovShephard2019,
  author = {Bojinov, I. and Shephard, N.},
  title = {{{Time series experiments and causal estimands: Exact randomization tests and trading}}},
  journal = {{{Journal of the American Statistical Association}}},
  year = {2019},
  volume = {114},
  number = {528},
  pages = {1665--1682}
}

@misc{CarlsonShephard2026,
  author = {Carlson, J. and Shephard, N.},
  title = {{{When are time series predictions causal? The potential system and dynamic causal effects}}},
  year = {2026},
  eprint = {2603.20394},
  archivePrefix = {arXiv},
  note = {arXiv preprint arXiv:2603.20394.}
}

@article{ChalonerVerdinelli1995,
  author = {Chaloner, K. and Verdinelli, I.},
  title = {{{Bayesian experimental design: A review}}},
  journal = {{{Statistical Science}}},
  year = {1995},
  volume = {10},
  number = {3},
  pages = {273--304}
}

@book{CharnesCooper1961,
  author = {Charnes, A. and Cooper, W. W.},
  title = {{{Management Models and Industrial Applications of Linear Programming}}},
  year = {1961},
  publisher = {Wiley}
}

@book{Davidson1994,
  author = {Davidson, J.},
  title = {{{Stochastic Limit Theory}}},
  year = {1994},
  publisher = {Oxford University Press}
}

@article{FaruquiSergici2010,
  author = {Faruqui, A. and Sergici, S.},
  title = {{{Household response to dynamic pricing of electricity: A survey of 15 experiments}}},
  journal = {{{Journal of Regulatory Economics}}},
  year = {2010},
  volume = {38},
  number = {2},
  pages = {193--225}
}

@book{Fedorov1972,
  author = {Fedorov, V. V.},
  title = {{{Theory of Optimal Experiments}}},
  year = {1972},
  publisher = {Academic Press}
}

@article{FerreiraMirandaAgrippinoRicco2025,
  author = {Ferreira, L. N. and Miranda-Agrippino, S. and Ricco, G.},
  title = {{{Bayesian local projections}}},
  journal = {{{Review of Economics and Statistics}}},
  year = {2025},
  volume = {107},
  number = {5},
  pages = {1424--1438}
}

@book{GoodwinPayne1977,
  author = {Goodwin, G. C. and Payne, R. L.},
  title = {{{Dynamic System Identification: Experiment Design and Data Analysis}}},
  year = {1977},
  publisher = {Academic Press}
}

@book{HampelEtAl1986,
  author = {Hampel, F. R. and Ronchetti, E. M. and Rousseeuw, P. J. and Stahel, W. A.},
  title = {{{Robust Statistics: The Approach Based on Influence Functions}}},
  year = {1986},
  publisher = {Wiley}
}

@article{Hansen1982,
  author = {Hansen, L. P.},
  title = {{{Large sample properties of generalized method of moments estimators}}},
  journal = {{{Econometrica}}},
  year = {1982},
  volume = {50},
  number = {4},
  pages = {1029--1054}
}

@article{HansenHeatonYaron1996,
  author = {Hansen, L. P. and Heaton, J. and Yaron, A.},
  title = {{{Finite-sample properties of some alternative GMM estimators}}},
  journal = {{{Journal of Business \& Economic Statistics}}},
  year = {1996},
  volume = {14},
  number = {3},
  pages = {262--280}
}

@article{HansenLunde2005,
  author = {Hansen, P. R. and Lunde, A.},
  title = {{{A forecast comparison of volatility models: Does anything beat a GARCH(1,1)?}}},
  journal = {{{Journal of Applied Econometrics}}},
  year = {2005},
  volume = {20},
  number = {7},
  pages = {873--889}
}

@inproceedings{HohnholdOBrienTang2015,
  author = {Hohnhold, H. and O'Brien, D. and Tang, D.},
  title = {{{Focusing on the long-term: It's good for users and business}}},
  booktitle = {{{Proceedings of the 21st ACM SIGKDD International Conference on Knowledge Discovery and Data Mining}}},
  pages = {1849--1858},
  year = {2015},
  note = {doi:10.1145/2783258.2788583}
}

@article{HowardEtAl2021,
  author = {Howard, S. R. and Ramdas, A. and McAuliffe, J. and Sekhon, J.},
  title = {{{Time-uniform, nonparametric, nonasymptotic confidence sequences}}},
  journal = {{{Annals of Statistics}}},
  year = {2021},
  volume = {49},
  number = {2},
  pages = {1055--1080}
}

@article{Huber1964,
  author = {Huber, P. J.},
  title = {{{Robust estimation of a location parameter}}},
  journal = {{{Annals of Mathematical Statistics}}},
  year = {1964},
  volume = {35},
  number = {1},
  pages = {73--101}
}

@article{HudgensHalloran2008,
  author = {Hudgens, M. G. and Halloran, M. E.},
  title = {{{Toward causal inference with interference}}},
  journal = {{{Journal of the American Statistical Association}}},
  year = {2008},
  volume = {103},
  number = {482},
  pages = {832--842}
}

@book{ImbensRubin2015,
  author = {Imbens, G. W. and Rubin, D. B.},
  title = {{{Causal Inference for Statistics, Social, and Biomedical Sciences: An Introduction}}},
  year = {2015},
  publisher = {Cambridge University Press}
}

@article{InoueJordaKuersteiner2026,
  author = {Inoue, A. and Jord{\`a}, {\`O}. and Kuersteiner, G. M.},
  title = {{{Inference for local projections}}},
  journal = {{{The Econometrics Journal}}},
  year = {2026},
  volume = {29},
  number = {1},
  pages = {2--26}
}

@article{Ito2014,
  author = {Ito, K.},
  title = {{{Do consumers respond to marginal or average price? Evidence from nonlinear electricity pricing}}},
  journal = {{{American Economic Review}}},
  year = {2014},
  volume = {104},
  number = {2},
  pages = {537--563}
}

@article{ItoIdaTanaka2018,
  author = {Ito, K. and Ida, T. and Tanaka, M.},
  title = {{{Moral suasion and economic incentives: Field experimental evidence from energy demand}}},
  journal = {{{American Economic Journal: Economic Policy}}},
  year = {2018},
  volume = {10},
  number = {1},
  pages = {240--267}
}

@misc{JiaKallusYu2023,
  author = {Jia, S. and Kallus, N. and Yu, C. L.},
  title = {{{Clustered switchback experiments: Near-optimal rates under spatiotemporal interference}}},
  year = {2023},
  eprint = {2312.15574},
  archivePrefix = {arXiv},
  note = {arXiv preprint arXiv:2312.15574.}
}

@article{Jorda2005,
  author = {Jord{\`a}, {\`O}.},
  title = {{{Estimation and inference of impulse responses by local projections}}},
  journal = {{{American Economic Review}}},
  year = {2005},
  volume = {95},
  number = {1},
  pages = {161--182}
}

@article{Kiefer1959,
  author = {Kiefer, J.},
  title = {{{Optimum experimental designs}}},
  journal = {{{Journal of the Royal Statistical Society: Series B}}},
  year = {1959},
  volume = {21},
  number = {2},
  pages = {272--304}
}

@article{KilianKim2011,
  author = {Kilian, L. and Kim, Y. J.},
  title = {{{How reliable are local projection estimators of impulse responses?}}},
  journal = {{{Review of Economics and Statistics}}},
  year = {2011},
  volume = {93},
  number = {4},
  pages = {1460--1466}
}

@article{KoenkerBassett1978,
  author = {Koenker, R. and Bassett, G.},
  title = {{{Regression quantiles}}},
  journal = {{{Econometrica}}},
  year = {1978},
  volume = {46},
  number = {1},
  pages = {33--50}
}

@article{KohaviLongbothamSommerfieldHenne2009,
  author = {Kohavi, R. and Longbotham, R. and Sommerfield, D. and Henne, R. M.},
  title = {{{Controlled experiments on the web: Survey and practical guide}}},
  journal = {{{Data Mining and Knowledge Discovery}}},
  year = {2009},
  volume = {18},
  number = {1},
  pages = {140--181}
}

@book{KohaviTangXu2020,
  author = {Kohavi, R. and Tang, D. and Xu, Y.},
  title = {{{Trustworthy Online Controlled Experiments: A Practical Guide to A/B Testing}}},
  year = {2020},
  publisher = {Cambridge University Press}
}

@article{KolesarPlagborgMoller2025,
  author = {Koles{\'a}r, M. and Plagborg-M{\o}ller, M.},
  title = {{{Dynamic causal effects in a nonlinear world: The good, the bad, and the ugly}}},
  journal = {{{Journal of Business \& Economic Statistics}}},
  year = {2025},
  volume = {43},
  number = {4},
  pages = {737--754}
}

@article{LeeEtAl2016,
  author = {Lee, J. D. and Sun, D. L. and Sun, Y. and Taylor, J. E.},
  title = {{{Exact post-selection inference, with application to the lasso}}},
  journal = {{{Annals of Statistics}}},
  year = {2016},
  volume = {44},
  number = {3},
  pages = {907--927}
}

@article{Leung2022,
  author = {Leung, M. P.},
  title = {{{Causal inference under approximate neighborhood interference}}},
  journal = {{{Econometrica}}},
  year = {2022},
  volume = {90},
  number = {1},
  pages = {267--293}
}

@article{LiPlagborgMollerWolf2024,
  author = {Li, D. and Plagborg-M{\o}ller, M. and Wolf, C. K.},
  title = {{{Local projections vs. VARs: Lessons from thousands of DGPs}}},
  journal = {{{Journal of Econometrics}}},
  year = {2024},
  volume = {244},
  number = {2},
  pages = {105722}
}

@misc{LinDing2025,
  author = {Lin, Z. and Ding, P.},
  title = {{{Unifying regression-based and design-based causal inference in time-series experiments}}},
  year = {2025},
  eprint = {2510.22864},
  archivePrefix = {arXiv},
  note = {arXiv preprint arXiv:2510.22864.}
}

@article{Lindley1956,
  author = {Lindley, D. V.},
  title = {{{On a measure of the information provided by an experiment}}},
  journal = {{{Annals of Mathematical Statistics}}},
  year = {1956},
  volume = {27},
  number = {4},
  pages = {986--1005}
}

@book{Ljung1999,
  author = {Ljung, L.},
  title = {{{System Identification: Theory for the User}}},
  year = {1999},
  edition = {2nd},
  publisher = {Prentice Hall}
}

@article{Mehra1974,
  author = {Mehra, R. K.},
  title = {{{Optimal inputs for linear system identification}}},
  journal = {{{IEEE Transactions on Automatic Control}}},
  year = {1974},
  volume = {19},
  number = {3},
  pages = {192--200}
}

@book{Miettinen1999,
  author = {Miettinen, K.},
  title = {{{Nonlinear Multiobjective Optimization}}},
  year = {1999},
  publisher = {Springer}
}

@article{Mikusheva2007,
  author = {Mikusheva, A.},
  title = {{{Uniform inference in autoregressive models}}},
  journal = {{{Econometrica}}},
  year = {2007},
  volume = {75},
  number = {5},
  pages = {1411--1452}
}

@article{MontielOleaPlagborgMoller2021,
  author = {Montiel Olea, J. L. and Plagborg-M{\o}ller, M.},
  title = {{{Local projection inference is simpler and more robust than you think}}},
  journal = {{{Econometrica}}},
  year = {2021},
  volume = {89},
  number = {4},
  pages = {1789--1823}
}

@article{NewshamBowker2010,
  author = {Newsham, G. R. and Bowker, B. G.},
  title = {{{The effect of utility time-varying pricing and load control strategies on residential summer peak electricity use: A review}}},
  journal = {{{Energy Policy}}},
  year = {2010},
  volume = {38},
  number = {7},
  pages = {3289--3296}
}

@article{OBrienFleming1979,
  author = {O'Brien, P. C. and Fleming, T. R.},
  title = {{{A multiple testing procedure for clinical trials}}},
  journal = {{{Biometrics}}},
  year = {1979},
  volume = {35},
  number = {3},
  pages = {549--556}
}

@article{Owen1988,
  author = {Owen, A. B.},
  title = {{{Empirical likelihood ratio confidence intervals for a single functional}}},
  journal = {{{Biometrika}}},
  year = {1988},
  volume = {75},
  number = {2},
  pages = {237--249}
}

@article{Ozaki2018,
  author = {Ozaki, R.},
  title = {{{Follow the price signal: People's willingness to shift household practices in a dynamic time-of-use tariff trial in the United Kingdom}}},
  journal = {{{Energy Research \& Social Science}}},
  year = {2018},
  volume = {46},
  pages = {10--18},
  note = {doi:10.1016/j.erss.2018.06.008}
}

@article{PesaventoRossi2007,
  author = {Pesavento, E. and Rossi, B.},
  title = {{{Impulse response confidence intervals for persistent data: What have we learned?}}},
  journal = {{{Journal of Economic Dynamics and Control}}},
  year = {2007},
  volume = {31},
  number = {7},
  pages = {2398--2412}
}

@article{Phillips1998,
  author = {Phillips, P. C. B.},
  title = {{{Impulse response and forecast error variance asymptotics in nonstationary VARs}}},
  journal = {{{Journal of Econometrics}}},
  year = {1998},
  volume = {83},
  number = {1--2},
  pages = {21--56}
}

@article{PlagborgMollerWolf2021,
  author = {Plagborg-M{\o}ller, M. and Wolf, C. K.},
  title = {{{Local projections and VARs estimate the same impulse responses}}},
  journal = {{{Econometrica}}},
  year = {2021},
  volume = {89},
  number = {2},
  pages = {955--980}
}

@article{Pocock1977,
  author = {Pocock, S. J.},
  title = {{{Group sequential methods in the design and analysis of clinical trials}}},
  journal = {{{Biometrika}}},
  year = {1977},
  volume = {64},
  number = {2},
  pages = {191--199}
}

@book{PronzatoPazman2013,
  author = {Pronzato, L. and P{\'a}zman, A.},
  title = {{{Design of Experiments in Nonlinear Models: Asymptotic Normality, Optimality Criteria and Small-Sample Properties}}},
  year = {2013},
  publisher = {Springer}
}

@book{Pukelsheim2006,
  author = {Pukelsheim, F.},
  title = {{{Optimal Design of Experiments}}},
  year = {2006},
  publisher = {SIAM}
}

@misc{RambachanShephard2025,
  author = {Rambachan, A. and Shephard, N.},
  title = {{{When do common time series estimands have nonparametric causal meaning?}}},
  year = {2025},
  eprint = {1903.01637},
  archivePrefix = {arXiv},
  note = {arXiv preprint arXiv:1903.01637, revised January 2025}
}

@incollection{Ramey2016,
  author = {Ramey, V. A.},
  title = {{{Macroeconomic shocks and their propagation}}},
  booktitle = {{{Handbook of Macroeconomics}}},
  editor = {Taylor, J. B. and Uhlig, H.},
  volume = {2},
  pages = {71--162},
  publisher = {Elsevier},
  year = {2016}
}

@article{RameyZubairy2018,
  author = {Ramey, V. A. and Zubairy, S.},
  title = {{{Government spending multipliers in good times and in bad: Evidence from U.S. historical data}}},
  journal = {{{Journal of Political Economy}}},
  year = {2018},
  volume = {126},
  number = {2},
  pages = {850--901}
}

@book{RueHeld2005,
  author = {Rue, H. and Held, L.},
  title = {{{Gaussian Markov Random Fields: Theory and Applications}}},
  year = {2005},
  publisher = {Chapman \& Hall/CRC}
}

@article{SavjeAronowHudgens2021,
  author = {S{\"a}vje, F. and Aronow, P. M. and Hudgens, M. G.},
  title = {{{Average treatment effects in the presence of unknown interference}}},
  journal = {{{Annals of Statistics}}},
  year = {2021},
  volume = {49},
  number = {2},
  pages = {673--701}
}

@phdthesis{Schofield2015,
  author = {Schofield, J. R.},
  title = {{{Dynamic time-of-use electricity pricing for residential demand response: Design and analysis of the Low Carbon London smart-metering trial}}},
  year = {2015},
  school = {Imperial College London}
}

@article{Schwarz1978,
  author = {Schwarz, G.},
  title = {{{Estimating the dimension of a model}}},
  journal = {{{Annals of Statistics}}},
  year = {1978},
  volume = {6},
  number = {2},
  pages = {461--464}
}

@misc{StrbacEtAl2024,
  author = {Strbac, G. and Tindemans, S. H. and Woolf, M. and Bilton, M. and Carmichael, R. and Schofield, J. R.},
  title = {{{Low Carbon London Project: Data from the Dynamic Time-of-Use Electricity Pricing Trial, 2013 [data collection]}}},
  year = {2024},
  howpublished = {UK Data Service, Study 7857},
  note = {[dataset], doi:10.5255/UKDA-SN-7857-2}
}

@article{Wald1945,
  author = {Wald, A.},
  title = {{{Sequential tests of statistical hypotheses}}},
  journal = {{{Annals of Mathematical Statistics}}},
  year = {1945},
  volume = {16},
  number = {2},
  pages = {117--186}
}

@misc{Wang2026,
  author = {Wang, E.},
  title = {{{Local projections identify the same policy counterfactuals as empirical and structural models}}},
  year = {2026},
  eprint = {2409.09577},
  archivePrefix = {arXiv},
  note = {arXiv preprint arXiv:2409.09577, revised February 2026}
}

@book{White2001,
  author = {White, H.},
  title = {{{Asymptotic Theory for Econometricians}}},
  year = {2001},
  edition = {Revised},
  publisher = {Academic Press}
}

@article{Wolak2011,
  author = {Wolak, F. A.},
  title = {{{Do residential customers respond to hourly prices? Evidence from a dynamic pricing experiment}}},
  journal = {{{American Economic Review}}},
  year = {2011},
  volume = {101},
  number = {3},
  pages = {83--87}
}

@inproceedings{WuEtAl2026,
  author = {Wu, X. and Wen, Q. and Zhang, Y. and Zhu, H. and Li, T. and Shi, C.},
  title = {{{Designing time series experiments in A/B testing with Transformer reinforcement learning}}},
  booktitle = {{{International Conference on Learning Representations}}},
  year = {2026}
}

@misc{XiongChinTaylor2024,
  author = {Xiong, R. and Chin, A. and Taylor, S. J.},
  title = {{{Data-driven switchback experiments: Theoretical tradeoffs and empirical Bayes designs}}},
  year = {2024},
  eprint = {2406.06768},
  archivePrefix = {arXiv},
  note = {arXiv preprint arXiv:2406.06768.}
}

@article{YunusovTorriti2021,
  author = {Yunusov, T. and Torriti, J.},
  title = {{{Distributional effects of Time of Use tariffs based on electricity demand and time use}}},
  journal = {{{Energy Policy}}},
  year = {2021},
  volume = {156},
  pages = {112412},
  note = {doi:10.1016/j.enpol.2021.112412}
}

@book{Zarrop1979,
  author = {Zarrop, M. B.},
  title = {{{Optimal Experiment Design for Dynamic System Identification}}},
  year = {1979},
  publisher = {Springer}
}

@misc{ZengEtAl2026,
  author = {Zeng, Z. and Adjaho, C. and Bucarey, A. and Qin, C. and Zhang, R. and Hoban, P. and Johari, R. and Wager, S.},
  title = {{{Sequentially-rerandomized switchback experiments}}},
  year = {2026},
  eprint = {2604.02489},
  archivePrefix = {arXiv},
  note = {arXiv preprint arXiv:2604.02489.}
}

\end{document}